\documentclass[11pt]{article}

\pdfoutput=1 %

\usepackage{amsmath,amsthm,amssymb} %
\usepackage[affil-it]{authblk} %
\usepackage[english]{babel} %
\usepackage{caption} %
\usepackage{color} %
\usepackage{csquotes} %
\usepackage{dsfont} %
\usepackage{enumitem} %
\usepackage{flushend} %
\usepackage[T1]{fontenc} %
\usepackage{mathpazo} 
\usepackage{framed} %
\usepackage[margin=1in]{geometry} %
\usepackage{graphicx} %
\usepackage{hyphenat} %
\usepackage[breaklinks]{hyperref} %
\usepackage{mathabx} %
\usepackage{mathtools} %
\usepackage{microtype} %
\usepackage[utf8]{inputenc} %
\usepackage{soul} %
\usepackage{subcaption} %
\usepackage{subdepth} %
\usepackage{suffix} 
\usepackage{tikz} %
\usepackage{xspace} %

\hypersetup{colorlinks=true, linkcolor=blue, citecolor=magenta} %
\definecolor{shadecolor}{rgb}{.95,.95,.95} %
\setul{1ex}{.5pt} %
\setlist[enumerate]{nolistsep,itemsep=3pt,topsep=3pt} %

\definecolor{White}{rgb}{1,1,1} %
\definecolor{Black}{rgb}{0,0,0} %
\definecolor{LightGray}{rgb}{.8,.8,.8} %
\colorlet{ChannelColor}{LightGray} %
\colorlet{ChannelTextColor}{Black} %
\colorlet{ReadoutColor}{White} %

\newtheorem{theorem}{Theorem} %
\newtheorem{lemma}[theorem]{Lemma} %
\newtheorem{corollary}[theorem]{Corollary} %
\newtheorem{definition}[theorem]{Definition} %

\newcommand{\microspace}{\mspace{.5mu}} %
\newcommand{\ket}[1]{\ensuremath{\lvert\microspace #1
    \microspace\rangle}} %
\newcommand{\bigket}[1]{\bigl\lvert\microspace #1
  \microspace\bigr\rangle} %
\newcommand{\bra}[1]{\ensuremath{\langle\microspace #1
    \microspace\rvert}} %
\newcommand{\ip}[2]{\ensuremath{\left\langle#1,#2\right\rangle}} %
\newcommand{\norm}[1]{\ensuremath{\left\lVert #1 \right\rVert}} %
\newcommand{\abs}[1]{\ensuremath{\left\lvert #1 \right\rvert}} %
\newcommand{\ceil}[1]{\ensuremath{\left\lceil #1 \right\rceil}} %
\newcommand{\defeq}{\stackrel{\smash{\text{\tiny\rm def}}}{=}} %

\newcommand{\complex}{\mathbb{C}} %
\renewcommand{\Re}{\operatorname{Re}} %

\newcommand{\class}[1]{\textup{#1}\xspace} %

\newcommand{\NP}{\class{NP}} %
\newcommand{\IP}{\class{IP}} %
\newcommand{\NEXP}{\class{NEXP}} %
\newcommand{\QMA}{\class{QMA}} %
\newcommand{\QMIP}{\class{QMIP}} %
\WithSuffix\newcommand\QMIP*{\ensuremath{\class{QMIP}^*}} %
\newcommand{\PSPACE}{\class{PSPACE}} %
\newcommand{\PCP}{\class{PCP}} %
\newcommand{\MIP}{\class{MIP}} %
\WithSuffix\newcommand\MIP*{\ensuremath{\class{MIP}^*}} %
\newcommand{\QIP}{\class{QIP}} %
\newcommand{\QMAM}{\class{QMAM}} %

\newcommand{\reg}[1]{\textsf{#1}} %
\newcommand{\game}[1]{\textit{#1}} %
\newcommand{\dtr}{\operatorname{D}} 
\newcommand{\dis}{d} 
\newcommand{\con}{\operatorname{C}} 

\newcommand{\setft}[1]{\mathrm{#1}} %
\newcommand{\Density}{\setft{D}} %
\newcommand{\Pos}{\setft{Pos}} %
\newcommand{\Unitary}{\setft{U}} %
\newcommand{\Herm}{\setft{Herm}} %
\newcommand{\Lin}{\setft{L}} %

\newcommand{\poly}{\mathit{poly}} %

\newcommand{\xap}{\ensuremath{\coAsterisk}} 

\renewcommand{\hat}[1]{\widehat{#1}} 
\renewcommand{\tilde}[1]{\widetilde{#1}} 
\renewcommand{\check}[1]{\widecheck{#1}} 

\DeclareMathOperator{\tr}{tr} %
\DeclareMathOperator{\SWAP}{SWAP} %
\DeclareMathOperator{\CNOT}{CNOT} %
\DeclareMathOperator{\TOFFOLI}{TOFFOLI} %
\DeclareMathOperator{\MAP}{MAP} %

\def\B{\mathcal{B}} 
\def\C{\mathcal{C}} 
\def\H{\mathcal{H}} 
\def\I{\mathds{1}} 
\def\M{\mathcal{M}} 
\def\P{\mathcal{P}} 
\def\V{\mathcal{V}} 
\def\R{\mathcal{R}} 
\def\S{\mathcal{S}} 
\def\T{\mathcal{T}} 
\def\X{\mathcal{X}} 
\def\Y{\mathcal{Y}} 
\def\Z{\mathcal{Z}} 

\newcommand{\E}{\mathop{\mathbb{E}}\displaylimits} 
\newcommand{\Pauli}{\mathbb{P}} %
\newcommand{\Power}{\mathbb{Q}} %
\newcommand{\Stabilizer}{\mathcal{S}} %
\newcommand{\Strategy}{\mathfrak{S}} %

\newcommand{\legal}{\text{legal}} %
\newcommand{\prop}{\text{prop}} %
\newcommand{\propv}{\text{prop,v}} %
\newcommand{\propp}{\text{prop,p}} %
\newcommand{\GHZ}{\text{GHZ}} %
\newcommand{\cons}{\text{cons}} %

\begin{document}


\title{\Large\bf Compression of Quantum Multi-Prover Interactive
  Proofs}

\author{Zhengfeng Ji}

\affil{\small Centre for Quantum Software and Information, School of
  Software, Faculty of Engineering and Information Technology,
  University of Technology Sydney, NSW, Australia}


\date{\today}

\maketitle

\begin{abstract}
  We present a protocol that transforms any quantum multi-prover
  interactive proof into a nonlocal game in which questions consist of
  logarithmic number of bits and answers of constant number of bits.
  As a corollary, this proves that the promise problem corresponding
  to the approximation of the nonlocal value to inverse polynomial
  accuracy is complete for \QMIP*, and therefore \NEXP-hard.
  This establishes that nonlocal games are provably harder than
  classical games without any complexity theory assumptions.
  Our result also indicates that gap amplification for nonlocal games
  may be impossible in general and provides a negative evidence for
  the possibility of the gap amplification approach to the
  multi-prover variant of the quantum \PCP conjecture.
\end{abstract}

\section{Introduction}

The notion of the efficient proof verification is one of the
fundamental concepts in the theory of computing.
Proof verification models and corresponding complexity classes ranging
from \NP, to \IP, \MIP and \PCP greatly enrich the theory of
computing.
The class \NP~\cite{Coo71,Lev73,Kar72}, one of the cornerstones of
theoretical computer science, corresponds to the proof verification of
a proof string by an efficient deterministic computer.
Interactive models of proof verification were first proposed by
Babai~\cite{Bab85} and Goldwasser, Micali, and Rackoff~\cite{GMR85}.
It is generalized to the multi-prover setting by Ben-Or, Goldwasser,
Kilian and Wigderson~\cite{BGKW88}.
The study of different proof systems through the computational lens
has led to a blossom of celebrated results in computational complexity
theory and
cryptography~(e.g.,~\cite{LFKN92,Sha92,BFL90,AS98,ALM+98,GMW91,GMR89}).

The efforts of understanding proof systems in the context of quantum
computing have also been fruitful (see
e.g.,~\cite{Kit99,KKR06,GI13,OT08,KSV02,Wat09,AN02,Wat99,JJUW11,IV12,Vid13,RUV13}).
These interesting results are nicely summarized in the recent survey
on quantum proofs by Vidick and Watrous~\cite{VW16}.
We emphasize that, in the development of quantum proofs, entanglement
has played a dramatic role---it is both the cause of the problems and
the key to the solutions as well.

A quantum analog of \NP was proposed by
Kitaev~\cite{Kit99,KSV02,AN02}.
In this generalization, a quantum witness state plays the role of the
proof string and a polynomial-time quantum computer checks whether the
witness state is valid for the input.
Kitaev introduces the class \QMA of problems that admit efficient
verifiable quantum proofs.
He also establishes the quantum analog of the Cook-Levin theorem by
showing that the local Hamiltonian problem, the natural quantum
version of the constraint satisfaction problems, is complete for \QMA.
As elaborated in~\cite{AN02}, the difficulty is to perform local
propagation checks on the snapshot states which may be highly
entangled.
The circuit-to-Hamiltonian construction, the key technique for the
quantum Cook-Levin theorem, demonstrates how one can locally check the
propagation of quantum computation, by introducing an extra clock
system that entangles with the computational system.
A generalization of this construction to the interactive setting will
be one of the key ingredients of our result.

Entanglement also has unexpected use in single-prover quantum
interactive proof systems, \QIP, in which an efficient quantum
verifier exchanges quantum messages with a quantum prover before
making decisions.
Watrous presented a constant-round quantum interactive proof system
for \PSPACE~\cite{Wat99,KW00}, in which entanglement is exploited to
enforce the correct temporal structure in a classical interactive
proof for \PSPACE.
Alternatively, one can view this parallelization technique as dividing
the interactive computation into two halves and check either forward
or backward from the middle point.
This idea also gave rise to a simple public coin characterization of
\QIP called \QMAM~\cite{MW05}, which in turn helps in the final proof
that $\QIP = \PSPACE$~\cite{JJUW11}.
The technique can be extended to the multi-prover setting and show
that quantum multi-prover interactive proof systems also parallelize
to constant-rounds~\cite{KKMV08}.
This will provide a starting point for our work.

This paper is about quantum multi-prover interactive proofs and
nonlocal games, the scaled-down version of one-round quantum
multi-prover proofs with classical messages.
The class of languages that have quantum multi-prover interactive
proofs is denoted as \QMIP*.
In the multi-prover setting, shared entanglement among the provers
becomes the natural focus of the study, a topic that has received
continuing interests in physics foundations since
1960's~\cite{Bel64,KS67,Tsi80,Wer89,Mer90,Per90}.
From the complexity perspective, it is known that, without shared
entanglement, or with limited amount of entanglement, the collection
of languages that have quantum multi-prover interactive proof systems
equals to the classical counterpart, \MIP~\cite{KM03} (and, hence,
also equals to \NEXP~\cite{BFL90}).
It was pointed out in~\cite{CHTW04} that provers with shared
entanglement may break the soundness condition of a classically sound
protocol.
One striking example is given by the so-called magic square
game~\cite{Mer90,Per90}, which has nonlocal value\footnote{The
  nonlocal value of a multi-player one-round game is the supremum of
  the probability that entangled players can make the verifier
  accept.}
one even though it corresponds to a system of constraints with no
classical solution~\cite{CHTW04}.
Strong evidences are also given in that paper that the entanglement
between the players may indeed weaken the power of two-player XOR
games.

Several methods have been proposed to control the cheating ability of
entangled provers and recover soundness in certain cases.
It is proved that approximating the nonlocal value of a multi-player
game to inverse-polynomial precision is \NP-hard~\cite{KKM+08,IKM09},
and therefore at least as hard as approximating the classical
value~\cite{FRS94}.
Several natural problems arise from the study of nonlocality,
including the binary constraint system game~\cite{CM12}, the quantum
coloring game~\cite{CMN+07,RM12} and the game corresponding to the
Kochen-Specker sets~\cite{KS67}, are shown to be \NP-hard
in~\cite{Ji13}.
By proving that the multi-linearity test~\cite{BFL90} is sound against
entangled provers, Ito and Vidick proved the containment of \NEXP in
\MIP*~\cite{IV12}.
This was later improved to the result that three-player XOR games are
\NP-hard to approximate even to constant precision~\cite{Vid13}.
Very recently, techniques introduced in~\cite{FV15,Ji16} allow us to
go beyond the \NP-hardness type of results and prove that nonlocal
games are \QMA-hard.
The problem of the existence of perfect commuting-operator strategy
for binary constraint system games was shown to be undecidable in a
recent breakthrough~\cite{Slo16}.
It is, however, not comparable to the above results mainly because it
does not tolerate approximation errors.

In this paper, we significantly improve the understanding of quantum
multi-prover interactive proofs and nonlocal games by showing that any
quantum multi-prover interactive proof can be \emph{compressed} in the
sense that the resulting protocol, a nonlocal game, has one round of
classical communication with messages consisting of logarithmic number
of bits.
It has perfect completeness and an inverse polynomial completeness and
soundness gap.
Our result is made possible by combining and exploiting the unique
features of entanglement that have already led to intriguing
understandings of quantum proof systems as discussed above.

\begin{theorem}
  \label{thm:main}
  For $r\in \poly$, any problem $A$ that has an $r$-prover quantum
  interactive proof, and any instance $x$ of the problem, there exists
  an $(r+8)$-player one-round game and real numbers $s\in 1 -
  \poly^{-1}(\abs{x})$, such that
  \begin{enumerate}
  \item The questions are classical bit strings of length
    $O(\log(\abs{x}))$.
  \item The answers are classical bit strings of length $O(1)$.
  \item If $x\in A$, then the nonlocal value of the game is $1$.
  \item If $x\not\in A$, then the nonlocal value of the game is at
    most $s$.
  \end{enumerate}
\end{theorem}

We mention that a corresponding claim in the classical case does not
hold since $\MIP = \NEXP$, the approximation of classical value is in
\NP, and $\NEXP \ne \NP$ by a diagonalization argument~\cite{Coo73}.
The approximation problem of nonlocal value is obviously in \QMIP* by
designing a multi-prover interactive protocol that sequentially
repeats the multi-player game polynomially many times.
This observation and Theorem~\ref{thm:main} imply that the problem is
in fact complete for the class \QMIP*.
As \NEXP is contained in $\QMIP*$~\cite{IV12}, a direct corollary of
Theorem~\ref{thm:main} is that approximating the nonlocal value of a
multi-player game is \NEXP-hard, improving the \QMA-hardness result
of~\cite{Ji16}.

\begin{corollary}
  Given a multi-player one-round game in which the questions are
  strings of $O(\log n)$ bits and answers are of strings of $O(1)$
  bits, it is \QMIP*-complete, and hence \NEXP-hard, to approximate
  the nonlocal value of the game to inverse polynomial precision.
\end{corollary}

The same problem for the classical value is obviously in \NP.
This means that the nonlocal value of multi-player one-round games is
provably harder to approximate than the classical value without any
complexity theory assumptions.

Our main theorem has the following consequence for the quantum
multi-prover interactive proofs with inverse exponential completeness
and soundness gap by scaling up the problem size.
Let $\class{NEEXP}$ be the class of nondeterministic
double-exponential time.
Let $\MIP^*(r,m,c,s)$ (and $\MIP(r,m,c,s)$) be the class of languages
that have $r$-prover, $m$-round interactive proofs with a classical
polynomial-time verifier, entangled provers (classical provers
respectively), completeness $c$, and soundness $s$.
We mention that $\MIP(\poly,\poly\,1,s) \subseteq \NEXP$ even for $s =
1 - \Omega(\exp(-p(n)))$ as a nondeterministic exponential time
machine may first guess all the interactions and compute the value for
this interaction to a precision of polynomially many bits.

\begin{corollary}
  \label{cor:MIP*}
  There exists a constant $r_0$ such that for $r \ge r_0$, there exist
  choices of soundness $s = 1 - \Omega(\exp(-p(n)))$ where $p(n)$ is
  some polynomial, such that
  \begin{equation*}
    \class{NEEXP} \subseteq \MIP^{*}(r,1,1,s),
  \end{equation*}
  and therefore, by the nondeterministic hierarchy
  theorem~\cite{Coo73},
  \begin{equation*}
    \MIP(r,1,1,s) \neq \MIP^{*}(r,1,1,s).
  \end{equation*}
\end{corollary}

Our result also indicates that the gap amplification for nonlocal
games may not be possible.
For classical multi-player game, one can reduce the inverse polynomial
approximation problem of the game value to the constant approximation
problem of some derived game, a procedure known as gap amplification.
This is an equivalent formulation of the classical \PCP theorem and
the approach of the alternative proof of the \PCP theorem given by
Dinur~\cite{Din07}.
Whether one can also amplify the gap of nonlocal game in a similar way
has been an interesting open problem.
Our result implies that it may not be possible at all.
If gap amplification works for nonlocal games, then one can start with
any nonlocal game, first perform gap amplification, and then scale up
the instance size (assuming that the resulting referee after gap
amplification has $\mathrm{polylog}$ time) and use our protocol to
transform it back into a nonlocal game with eight extra players.
This series of transformations will prove that nonlocal games with a
constantly many more players are exponentially harder, a situation
which does not seem to be plausible.
This provides negative evidence for the strong form of the quantum
\PCP conjecture that asks whether constant approximation of the
nonlocal value is as hard as inverse polynomial approximation.
It may still be possible to prove, and even using the gap
amplification approach for some special nonlocal games with certain
structure, that constant approximation to the nonlocal value is
\QMA-hard, a weaker form of the multi-player variant of quantum \PCP
conjecture.

Historically in the study of classical proof systems, we have started
from \NP, generalized it to \IP and \MIP~\cite{LFKN92,Sha92,BFL90},
motivated the study of \PCP and come back to \NP with the celebrated
\PCP theorem~\cite{AS98,ALM+98,Din07}.
Our result indicates that the landscape of quantum proof systems may be
very different.

Two important open questions are left open by this work. First, it is
an intriguing problem to understand the complexity of constant
approximation of the nonlocal value of a multi-player game. Second,
it is important to provide upper bounds for the class \QMIP*, a
problem that remains widely open.

\subsection{Techniques and Proof Overview}

Our proof is motivated by, and reuses many techniques from, the
previous work in~\cite{FV15,Ji16} but requires several new techniques
that we now discuss.

First, we recall that it is crucial in~\cite{FV15,Ji16} that we encode
the quantum witness state with certain quantum error
correcting/detecting code and distribute the encoded state among the
players so that we can prove rigidity
theorems~\cite{MY98,DMMS00,RUV13,McK14} that are helpful to enforce
the behavior of the players.
This can be thought of as the quantum analog of the oracularization
technique~\cite{FRS94}.
This technique alone, however, does not work anymore when we are
dealing with quantum interactive proofs instead of quantum witness
states as in the case for \QMA for the following reason.
In order to check the correct propagation for the provers' step, we
will ask the players to simulate the provers' actions, applications of
unitary circuits on their private qubits and the message qubits.
This will require that the provers' circuits are \emph{transversal}
over the underlying quantum code.
It is however well known that no quantum error correcting code
supports transversal universal quantum computation~\cite{ZCC11,EK09}.
To this end, we need to find a different way to encode and distribute
the qubits used in the interactive proof.

We introduce extended nonlocal games called propagation games and
constraint propagation games.
Propagation games exploit the idea of propagation checks in the proof
of \QMA-completeness for the local Hamiltonian problem and define a
corresponding game so that the shared state between the referee and
the player, who possess the clock and computation system respectively,
must be approximately close to the history state with respect to the
player's measurement strategies.
The constraint propagation game then adds the constraint checks to the
propagation game.
The constraints can be of any product form and can represent
commutativity and anti-commutativity as special cases.
We then define a variant of the constraint propagation game using a
constraint system satisfied by the Pauli operators on $n$ qubits of
weight $k$.
With this game, we avoid the problem of transversality and obtain
rigidity at the same time.

The use of extended nonlocal games for obtaining rigidity provides
great flexibility and largely simplifies the structure of the game.
This is the reason that constraint propagation games work with single
player, and also the reason that we can check the propagation of the
provers' step in an interactive proof system.
In particular, an extended nonlocal game defined by the stabilizer of
the GHZ state serves as a nonlocal game implementation of the
forward-backward checking technique in quantum interactive proofs
discussed in the introduction.

Our resulting nonlocal game for \QMIP* has perfect completeness.
To achieve this, we modify the stabilizer game introduced
in~\cite{Ji16} so that the new stabilizer game has prefect quantum
strategies.
An eight-qubit code is used to define the stabilizer game and a much
simpler proof of rigidity is provided for it.
We also need a proof technique first used in the construction of
zero-knowledge proofs for \QMA~\cite{BJSW16}, with which we design a
propagation verification procedure for the verifier's circuits so that
it suffices to measure commuting Pauli operators with $X$ and $Z$
factors only.

Our proof has the following overall structure.
First, we generalize Kitaev's circuit\hyp{}to\hyp{}Hamiltonian
construction for \QMA to the interactive setting and turn a quantum
multi-prover interactive proof system into an honest player game.
This honest player game plays the role of the random checking protocol
of the local Hamiltonian problem.
We then use the rigidity of the constraint propagation game based on a
constraint system satisfied by the Pauli operators to remove the
requirement that the players must measure honestly.
This gives rise to an extended nonlocal game for \QMIP*.
Finally, we turn this extended nonlocal game into a nonlocal game by
using eight extra players who encode and simulate the Pauli
measurements on the quantum system of the referee in the extended
nonlocal game.

\section{Preliminaries}

\label{sec:prel}

\subsection{Notions}

In this paper, a \emph{quantum register} refers to a named collection
of qubits that we view as a single unit.
Register names are represented by capital letters in a \emph{sans
  serif} font, such as $\reg{X}$, $\reg{Y}$, and $\reg{Z}$.
The associated Hilbert spaces are denoted by the same letters used in
the register names in a calligraphic font.
For example, $\X,\Y,\Z$ are the associated Hilbert spaces of registers
$\reg{X}$, $\reg{Y}$, and $\reg{Z}$ respectively.
To refer to some specific qubits in a register $\reg{X}$, we use qubit
index followed by the register name.
For example, $\bigl(\reg{X}, i_1, i_2\bigr)$ represents the $i_1$ and
$i_2$-th qubits of register $\reg{X}$ and the parentheses are omitted
when this is used in subscripts.
Hilbert spaces named with letter $\B$ are two-dimensional unless
stated otherwise.

We use $\Density(\X)$, $\Lin(\X)$, $\Herm(\X)$, $\Pos(\X)$ to denote
the set of density operators, bounded linear operators, Hermitian
operators and positive semidefinite operators on $\X$.
The adjoint of matrix $M$ is denoted as $M^*$.
For two Hermitian operators $M,N\in \Herm(\X)$, we write $M
\preccurlyeq N$ to mean $N-M\in \Pos(\X)$.
For matrix $M$, $\abs{M}$ is defined to be $\sqrt{M^* M}$.
The operator norm $\norm{M}$ of matrix $M$ is the largest eigenvalue
of $\abs{M}$.
The trace norm $\norm{M}_1$ of $M$ is the trace of $\abs{M}$.
An operator $R\in \Herm(\X)$ is a reflection if $R^2 = \I$.
An operator $R\in \Lin(\X)$ is a contraction if $\norm{R}\le 1$.
An operator $M\in \Lin(\X)$ is called traceless if $\tr(M) = 0$.

A positive-operator valued measure (POVM) is described by the
collection $\bigl\{ M^a \bigr\}$ for $M^a \in \Pos(\X)$.
Recall that the Naimark's theorem states that, for any POVM $\bigl\{
M^a \bigr\}$, there exists an isometry
\begin{equation*}
  V = \sum_a \sqrt{M^a} \otimes \ket{a},
\end{equation*}
such that
\begin{equation*}
  M^a = V^* (\I \otimes \ket{a}\bra{a}) V.
\end{equation*}
It will be technically convenient to apply the Naimark's theorem and
assume without loss of generality that the measurements considered in
this paper are projective measurements, and that each measurement
operator has the same rank.

For each reflection $R$, there naturally associates a two-outcome
projective quantum measurement $\{R^a\}$ where
\begin{equation*}
  R^a = \frac{\I+(-1)^a R}{2},
\end{equation*}
for $a=0,1$.
Conversely, for any two-outcome projective measurement $\{ R^a \}$,
there associates a reflection $R = R^0 - R^1$.
If the measurement operators have the same rank, the associated
reflection is traceless.

For any projective measurement $M = \bigl\{ M^a \bigr\}$ with $k$-bit
outcome, define reflections
\begin{equation*}
  R_i = \sum_{a\in\{0,1\}^k} (-1)^{a_i} M^a,
\end{equation*}
for $i\in [k]$.
The reflections $R_1, R_2, \ldots, R_k$ will be called the
\emph{derived reflections} of $M$.
It is easy to see that the a projective measurement $M$ with $k$-bit
outcome has a one-to-one correspondence with the tuple of $k$ derived
reflections $\bigl( R_i \bigr)_{i=1}^k$.

For a Hermitian operator $H\in \Herm(\X)$, and a subspace $S\subseteq
\X$, the restriction of $H$ to $S$ is
\begin{equation*}
  H\restriction_S = \Pi_S \,H\, \Pi_S,
\end{equation*}
where $\Pi_S$ is the projection onto the space $S$.

We will refer to the following elementary quantum gates in the paper:

\begin{enumerate}
\item The four single-qubit Pauli operators
  \begin{equation*}
    I =
    \begin{bmatrix}
      1 & 0\\
      0 & 1
    \end{bmatrix}, %
    \quad X =
    \begin{bmatrix}
      0 & 1\\
      1 & 0
    \end{bmatrix}, %
    \quad Y =
    \begin{bmatrix}
      0 & -i\\
      i & \phantom{-}0
    \end{bmatrix}, %
    \quad Z =
    \begin{bmatrix}
      1 & \phantom{-}0\\
      0 & -1
    \end{bmatrix},
  \end{equation*}
  which may also be denoted as $\sigma_0$, $\sigma_1$, $\sigma_2$, and
  $\sigma_3$ respectively sometimes.

\item The Hadamard gate
  \begin{equation*}
    H = \frac{1}{\sqrt{2}} %
    \begin{bmatrix} %
      1 & \phantom{-}1\\
      1 & -1
    \end{bmatrix}.
  \end{equation*}

\item Two-qubit unitary gates $\CNOT$ and $\SWAP$
  \begin{equation*}
    \CNOT \,\ket{j,k} = \ket{j,j\oplus k},\quad
    \SWAP \,\ket{j,k} = \ket{k,j},\quad \text{ for } j,k\in\{0,1\}.
  \end{equation*}

\item The Toffoli gate
  \begin{equation*}
    \TOFFOLI \,\ket{j,k,l} = \ket{j,k,l\oplus jk},\quad \text{ for }
    j,k,l\in\{0,1\}.
  \end{equation*}
\end{enumerate}

For unitary gate $U$, define $\Lambda_c(U)$ to be the controlled gate
\begin{equation*}
  \Lambda_c(U) = \ket{0}\bra{0}_c \otimes \I + \ket{1}\bra{1}_c
  \otimes U.
\end{equation*}

A quantum channel is a physically admissible transformation of quantum
states.
Mathematically, a quantum channel
\begin{equation*}
  \mathfrak{E}:\Lin(\X) \rightarrow \Lin(\Y)
\end{equation*}
is a completely positive, trace-preserving linear map.

The trace distance of two quantum states $\rho_0, \rho_1 \in
\Density(\X)$ is
\begin{equation*}
  \dtr(\rho_0, \rho_1) \defeq \frac{1}{2} \norm{\rho_0 - \rho_1}_1.
\end{equation*}
The monotonicity of the trace distance states that for quantum states
$\rho_0, \rho_1 \in \Density(\X)$ and all quantum channel
$\mathfrak{E} : \Lin(\X) \rightarrow \Lin(\Y)$,
\begin{equation*}
  \dtr \bigl(\mathfrak{E}(\rho_0), \mathfrak{E}(\rho_1) \bigr) \le
  \dtr(\rho_0, \rho_1).
\end{equation*}

For a string $x$, $\abs{x}$ denotes its length.
For a positive integer $k$, $[k]$ is the abbreviation of the set
$\{1,2,\ldots, k\}$.
For a set $A$, $\abs{A}$ denotes the number of elements in $A$.
We use $\poly$ to denote the collection of polynomially bounded
functions of either $n$ or $\abs{x}$ depending on the context.
For two complex numbers $a$, $b$, we use $a \approx_\epsilon b$ as a
shorthand notion for $\abs{a-b}\le O(\epsilon)$.

For quantum state $\rho\in \Density(\X)$ and operators $M,N\in
\Lin(\X)$, introduce the following notions
\begin{subequations}
  \begin{align}
    \tr_\rho (M) & = \tr(M\rho),\\
    \ip{M}{N}_\rho & = \tr_\rho (M^* N),\\
    \norm{M}_\rho & = \sqrt{\ip{M}{M}_\rho}.
  \end{align}
\end{subequations}
It is straightforward to verify that $\ip{\cdot}{\cdot}_\rho$ is a
semi-inner-product, $\norm{\cdot}_\rho$ is a seminorm and they become
an inner product and a norm, respectively, when $\rho$ is a full-rank
state.
By the Cauchy-Schwarz inequality,
\begin{equation*}
  \abs{\ip{M}{N}_\rho} \le \norm{M}_\rho \norm{N}_\rho,
\end{equation*}
or more explicitly,
\begin{equation*}
  \abs{\tr_\rho (M^* N)} \le \left[ \tr_\rho(M^* M) \,
    \tr_\rho(N^* N) \right]^{1/2}.
\end{equation*}

For state $\rho\in \Density(\X \otimes \Y)$, and an operator $M\in
\Lin(\X)$, we may also write $\tr_\rho(M)$ even though the state
$\rho$ and the operator $M$ do not act on the same space.
In this case, it is understood that $\tr_\rho(M) =
\tr_{\rho_{\reg{X}}} (M)$ where $\rho_{\reg{X}}$ is the reduced state
of $\rho$ on register $\reg{X}$.
This is one reason that makes $\tr_\rho(\cdot)$ easy to use as it is
not necessary to specify the correct reduced state explicitly all the
time.

We use the following version of the gentle measurement
lemma~\cite{Win99}:
\begin{lemma}
  \label{lem:gentle}
  Let $\rho \in \Density(\X)$ be a state on $\X$, $\Pi\in \Pos(\X)$ a
  projection and $\epsilon \in [0,1]$ a real number.
  If $\ip{\Pi}{\rho} \ge 1 - \epsilon$, then for $\rho_\Pi = \Pi \rho
  \Pi/\ip{\Pi}{\rho}$,
  \begin{equation*}
    \dtr(\rho,\rho_\Pi) \le O(\sqrt{\epsilon}).
  \end{equation*}
\end{lemma}

A simple corollary of the gentle measurement lemma is the following.

\begin{lemma}
  \label{lem:gapped}
  Let $H \in \Pos(\X)$ be a Hamiltonian that has a $0$-eigenspace $S$
  and all other eigenvalues are at least $\Delta$.
  Let $\Pi$ be the projection onto the $0$-eigenspace $S$.
  Let $\rho\in \Density(\X)$ be a quantum state.
  If $\tr_\rho (H) \le \epsilon$, then
  \begin{equation*}
    \dtr(\rho, \rho_\Pi) \le (\sqrt{\epsilon/\Delta}),
  \end{equation*}
  where $\rho_\Pi = \Pi \rho \Pi / \ip{\Pi}{\rho}$.
\end{lemma}

\subsection{Quantum Multi-Player Proof Systems}

An $r$-prover quantum interactive proof system consists of a verifier
$V$ and $r$ provers $P_1, P_2, \ldots, P_r$.
The verifier $V$ possesses a private quantum register $\reg{V}$
consisting of $q_V \in \poly$ qubits, each prover $P_i$ possesses a
private quantum register $\reg{P}_i$.
There are also $r$ message registers $\reg{M}_i$ each of which
contains $q_M \in \poly$ qubits.
Before the interaction starts, all qubits in register $\reg{V}$ are
initialized to $\ket{0}$.
The proves are not allowed to communicate after the interaction
starts.
The interaction consists of $m \in \poly$ alternating turns of the
applications of the verifier and the provers' circuits.
The verifier of an $m$-turn quantum multi-prover interactive proof
system is described by a tuple $V = \bigl( V^i
\bigr)_{i=1}^{\ceil{(m+1)/2}}$, where each $V^i$ is a polynomial-time
uniformly generated quantum circuit from input $x$ acting on registers
$\reg{V}, \reg{M}_1, \ldots, \reg{M}_r$.
For $l\in [r]$, the prover $P_l$ for an $m$-turn quantum multi-prover
interactive proof system is described by a tuple $W^l = \bigl( W^{i,l}
\bigr)_{i=1}^{\ceil{m/2}}$.
For odd $m$, the provers start by choosing a state $\ket{\psi} \in
\V\otimes \M$ then the verifier and the provers applies the circuits
$V^1$, $\bigotimes_{l=1}^r W^{1,l}$, $\ldots$, $\bigotimes_{l=1}^r
W^{(m-1)/2,l}$, $V^{(m+1)/2}$ in order.
For even $m$, the verifier initializes all qubits in $\reg{M}_l$ to
$\ket{0}$ and the verifier and provers apply the circuit $V^1$,
$\bigotimes_{l=1}^r W^{1,l}$, $\ldots$, $\bigotimes_{l=1}^r
W^{m/2,l}$, $V^{(m+2)/2}$.
The verifier then measures the first qubit in $\reg{V}$, accepts if
the outcome is $1$ and rejects otherwise.
The maximum acceptance probability $\MAP(V)$ for a given verifier
circuit $V$ is the maximum of the verifier's acceptance probability
for all possible quantum provers described by $\bigl( W^l
\bigr)_{l=1}^r$ and all correctly initialized state.

A language $A\in \QMIP*(r,m,c,s)$ if and only if there exists an
$r$-prover, $m$-turn quantum interactive proof systems with verifier
$V$ such that the following conditions hold:
\begin{enumerate}
\item (\textit{Completeness}) If $x\in A$, $\MAP(V) \ge c$,
\item (\textit{Soundness}) If $x\not\in A$, $\MAP(V) \le s$.
\end{enumerate}
Define \QMIP* to be $\QMIP*(\poly,\poly,2/3,1/3)$.
If the exchanged messages in a quantum multi-prover interactive proof
system are classical while the provers may still share entanglement
before the interaction starts, the corresponding complexity class will
be denoted as \MIP*.
It is now known that $\QMIP* = \MIP*$~\cite{RUV13}.

\subsection{Nonlocal Games and Extended Nonlocal Games}

Multi-player games have similar structure as multi-prover interactive
proofs with the main difference in the length of the messages.
In multi-prover interactive proofs, messages can consist of
polynomially many bit (or qubits), while in multi-player games, the
messages consist of logarithmic number of bits.
In a multi-player one-round game, a referee communicates with two or
more players classically in one round.
The referee samples questions and sends them out to the players and
expects to receive answers back.
He then accepts or rejects based on the questions and answers.
The players are allowed to agree on a strategy before the game starts,
but cannot communicate with each other during the game.

Let there be $r$ players, $(1), (2), \ldots, (r)$.
Let $\Gamma^{(i)}$ be a finite set of questions for player $(i)$ and
$\Lambda^{(i)}$ be a finite set of possible answers from player $(i)$.
An $r$-player game is defined by a distribution $\pi$ over
$\prod_{i=1}^r \Gamma^{(i)}$ and a function $V: \prod_{i=1}^r
\Lambda^{(i)} \times \prod_{i=1}^r \Gamma^{(i)} \rightarrow [0,1]$,
specifying the acceptance probability.
By a convexity argument, it suffices to consider the strategy of
classical players described by functions $f^{(i)} : \Gamma^{(i)}
\rightarrow \Lambda^{(i)}$.
The value of the strategy is the acceptance probability
\begin{equation*}
  \omega = \E_{q\sim \pi} V (a(q),q),
\end{equation*}
for $q=(q_1,q_2,\ldots,q_r)$ distributed according to $\pi$ and $a(q)
= \bigl( f^{(1)}(q_1), f^{(2)}(q_2), \ldots, f^{(r)}(q_r) \bigr)$.
The classical value of the game is the maximum of the values of all
classical strategies.

In a nonlocal game, the players are allowed to share an arbitrary
entangled state before the game starts.
A quantum strategy $\Strategy$ for the nonlocal game is described by
the shared state $\rho$, the measurements $\bigl\{ M^{(i)}_{q_i}
\bigr\}$ that player $(i)$ performs when the question is $q_i\in
\Gamma^{(i)}$.
The value of the strategy is defined as
\begin{equation*}
  \omega^*(\Strategy) = \E_{q\sim \pi} \sum_a \biggl[ \tr_\rho
  \bigl(\bigotimes_{i=1}^r M^{(i),a_i}_{q_i} \bigr) V(a,q) \biggr],
\end{equation*}
for $a=(a_1,a_2,\ldots,a_r)$ and $q=(q_1,q_2,\ldots,q_r)$.
The nonlocal value of the game is the supremum of the values of all
quantum strategies.

Extended nonlocal games were introduced in~\cite{JMRW16} as an
extension of the nonlocal game.
It is originally defined in a multi-player setting while we found that
it is also interesting to consider extended nonlocal games with a
single player.
An extended nonlocal game generalizes the nonlocal game in the
following sense.
The referee possesses a quantum register $\reg{S}$ but otherwise
samples and sends out questions similarly as in a nonlocal game.
The players can choose an initial state shared between the referee's
register and their private quantum systems.
The referee performs measurement on his quantum register and may
depend his acceptance also on the measurement outcome.
More formally, an $r$-player extended nonlocal game is defined by a
distribution $\pi$ over $\prod_{i=1}^r \Gamma^{(i)}$ and a function
$V: \prod_{i=1}^r \Lambda^{(i)} \times \prod_{i=1}^r \Gamma^{(i)}
\rightarrow [0,\I]$, where $[0,\I]$ is the set
\begin{equation*}
  \{ V\in \Pos(\S) \mid V \le \I \}.
\end{equation*}
A quantum strategy $\Strategy$ for the extended nonlocal game is
described by the shared state $\rho$, the measurements $\bigl\{
M^{(i)}_{q_i} \bigr\}$ that player $(i)$ performs when the question is
$q_i\in \Gamma^{(i)}$.
The value of the strategy is defined as
\begin{equation*}
  \omega^*(\Strategy) = \E_{q\sim \pi} \sum_a \biggl[ \tr_\rho
  \bigl(\bigotimes_{i=1}^r M^{(i),a_i}_{q_i} \otimes V(a,q) \bigr)
  \biggr],
\end{equation*}
for $a=(a_1,a_2,\ldots,a_r)$ and $q=(q_1,q_2,\ldots,q_r)$.
The value of the game is the supremum of the values of all quantum
strategies.

\subsection{Pauli Operators and Stabilizer Codes}

Let $\Pauli_n$ be the group generated by the $n$-fold tensor product
of Pauli operators
\begin{equation*}
  \Pauli_n = \biggl\{ e^{i\phi} \bigotimes_{j=1}^n D_j, \text{ for
  }\phi \in \{0,\pi/2,\pi,3\pi/2\},\; D_j \in \{I,X,Y,Z\} \biggr\}.
\end{equation*}
The weight of a Pauli operator in $\Pauli_n$ is the number of
non-identity tensor factors in it.
An Pauli operator is of $XZ$-form if each tensor factor is one of $I$,
$X$ and $Z$.
For Pauli operator $P$ of the form $(-1)^{\tau} \bigotimes_{j} D_j$,
the bit $\tau\in \{0,1\}$ is called the sign bit of the operator.

We present several relevant definitions and facts about the stabilizer
codes and refer the reader to the thesis of Gottesman~\cite{Got97} for
more details.
A stabilizer $\Stabilizer$ is an abelian subgroup of $\Pauli_n$ not
containing $-\I$.
It provides a succinct description of a corresponding subspace of
$(\complex^2)^{\otimes n}$---the simultaneous $+1$-eigenspace of the
operators in the stabilizer.
Let $C(\Stabilizer)$ be the centralizer of $\Stabilizer$ in
$\Pauli_n$, the set of operators in $\Pauli_n$ that commutes with all
operators in $\Stabilizer$.
The distance of the stabilizer code is $d$ if there is no operator of
weight less than $d$ in $C(\Stabilizer)-\Stabilizer$.
The logical $X$ and $Z$ operators $L_X$ and $L_Z$ are a pair of
anti-commuting operators in $C(\Stabilizer)-\Stabilizer$.

As a simple example, the operators $XXX$, $ZZI$ and $IZZ$ generate a
stabilizer for the $\GHZ$ state.
Operators $XXXX$ and $ZZZZ$ generate the stabilizer for the four-qubit
quantum error detecting code, which has distance two and encodes two
qubits.

\subsection{Distance Measures of Quantum Strategies}

\label{sec:dist}

We review the state-dependent distance measure $\dis_\rho$ and
consistency measure $\con_\rho$ of quantum measurements discussed in
detail in~\cite{Ji16}.
We also introduce the notion that an operator approximately stabilizes
a state and prove several related facts.
These concepts play an important role in analyzing the behavior of the
entangled players and are used extensively in our proofs.

Consider the situation where two players, Alice and Bob, share a
quantum state $\rho\in\Density(\X\otimes \Y)$ and Alice measures
$\reg{X}$ with either $\{M_0^a\}$ or $\{M_1^a\}$.
The post-measurement states are
\begin{equation}
  \label{eq:postms}
  \rho_i = \sum_a \ket{a}\bra{a} \otimes M_i^a \rho (M_i^a)^*,
\end{equation}
for $i = 0,1$, respectively, depending on which measurement is
performed.
By the monotonicity of the trace distance, the difference is bounded
by $\dtr(\rho_0, \rho_1)$.
As a special case, if Bob measures on $\reg{Y}$ and then the referee
makes the decision, the acceptance probabilities will differ by at
most $\dtr(\rho_0, \rho_1)$.
The state-dependent distance defined next provides a bound on the
distance $\dtr(\rho_0, \rho_1)$.
The claim is stated in Lemma~\ref{lem:drho} whose proof can be found
in~\cite{Ji16}.

\begin{definition}
  For two quantum measurements $M_i = \bigl\{ M_i^a \bigr\}$ with
  $i=0,1$ that have the same set of possible outcomes, define
  \begin{equation}
    \label{eq:drho}
    \dis_\rho(M_0,M_1) \defeq \Bigl[ \sum_a \norm{M_0^a-M_1^a}_\rho^2
    \Bigr]^{1/2}.
  \end{equation}
  More explicitly,
  \begin{equation}
    \label{eq:drho2}
    \dis_\rho(M_0,M_1) = \biggl[2-2\Re \sum_a \tr_\rho \bigl(
    (M_0^a)^* M_1^a \bigr) \biggr]^{1/2}.
  \end{equation}
\end{definition}

\begin{lemma}
  \label{lem:drho}
  Let $M_i = \bigl\{ M_i^a \bigr\}$ for $i=0,1$ be two quantum
  measurements with the same set of possible outcomes, and $\rho_i$ be
  the post-measurement states in Eq.~\eqref{eq:postms}.
  Then
  \begin{equation*}
    \dtr(\rho_0, \rho_1) \le \dis_\rho (M_0, M_1).
  \end{equation*}
\end{lemma}

A direct corollary of the above lemma is that replacing measurement
$M_0$ with $M_1$ in a strategy for a nonlocal game changes the value
by at most $\dis_\rho (M_0, M_1)$.
As this claim works for the general quantum measurement, it
generalizes to the special cases such as positive-operator valued
measures (POVM), projective measurements and measurements
corresponding to reflections.

In analysis of games, the post-measurement state is not important and
hence POVMs are the suitable formulation for quantum measurements.
It is technically convenient to adopt the following definition for the
state-dependent distance between to POVMs $M_i = \bigl\{ M_i^a
\bigr\}$ with $i=0,1$,
\begin{equation*}
  \dis_\rho (M_0, M_1) \defeq \inf_{N_i = \{ N_i^a \}} \dis_\rho (\{
  N_0^a \} , \{N_1^a \})
\end{equation*}
where the infimum is taken over all possible measurement operators
$N_i^a$ such that $M_i^a = (N_i^a)^* (N_i^a)$ for all $a$ and $i=0,1$.
The freedom to use arbitrary measurement operators, instead of
$\sqrt{M_i^a}$, provides simpler ways to upper bound the distance.

We focus on the case of reflections which will be extensively used in
later sections.

For two reflections $R_0, R_1$, let
\begin{equation*}
  R_i^a = \frac{\I+(-1)^a R_i}{2}
\end{equation*}
be the projective measurement operators correspond to $R_i$.
Define
\begin{equation*}
  \dis_\rho(R_0,R_1) \defeq \dis_\rho(\{R_0^a\}, \{R_1^a\}) = \bigl[1
  - \Re \tr_\rho \bigl( R_0 R_1 \bigr) \bigr]^{1/2}.
\end{equation*}

It is easy to verify that $\dis_\rho$ satisfy the triangle inequality.

\begin{lemma}
  Let $M_0$, $M_1$, $M_2$ be three measurements on state $\rho$.
  Then
  \begin{equation*}
    \dis_\rho(M_0,M_2) \le \dis_\rho(M_0,M_1) + \dis_\rho(M_1,M_2).
  \end{equation*}
\end{lemma}

Next, we recall the consistency measure for two quantum measurements
that act on two \emph{different} quantum systems.

\begin{definition}
  Let $\rho \in \Density(\X \otimes \Y)$ be a quantum state on
  $\reg{X}, \reg{Y}$, let $M= \bigl\{ M^a \bigr\}$, $N= \bigl\{ N^a
  \bigr\}$ be two POVMs on registers $\reg{X}, \reg{Y}$ respectively
  having the same set of possible outcomes.
  Define the consistency of $M$, $N$ on state $\rho$ as
  \begin{equation}
    \label{eq:cons}
    \operatorname{C}_\rho(M,N) \defeq \sum_a \tr_\rho (M^a\otimes N^a).
  \end{equation}
  $M$ and $N$ are called $\epsilon$-consistent on state $\rho$ if
  $\con_\rho(M,N) \ge 1 - \epsilon$.
\end{definition}

The consistency of measurements puts strong structural constraints on
the strategies of nonlocal game, which greatly simplify the analysis.
In this paper, we will be mostly interested in the consistency of two
reflections.
For two reflections $R$, $S$, let $\{R^a\}$, $\{S^a\}$ be their
corresponding projective measurements.
Define
\begin{equation}
  \label{eq:consref}
  \con_\rho(R,S) \defeq \con_\rho(\{R^a\}, \{S^a\}) = \frac{1+\tr_\rho(R\otimes
    S)}{2}.
\end{equation}
The condition $\tr_\rho(R\otimes S) \approx_\epsilon 1$, or
equivalently, $R,S$ are $O(\epsilon)$-consistent on $\rho$, can be
thought of as a quantitative way of saying that $\rho$ is
approximately stabilized by $R\otimes S$.

More generally, we introduce the notion of $\epsilon$-stabilizer as
follows.
It is crucial for later applications that we define this concept not
only for reflections but for the more general notion of contractions.

\begin{definition}
  Let $R\in \Lin(\X)$ be a contraction and $\rho \in \Density(\X)$ be
  a quantum state.
  We say that $R$ $\epsilon$-stabilizes $\rho$ if
  \begin{equation*}
    \Re \tr_{\rho} R \ge 1 - \epsilon.
  \end{equation*}
\end{definition}

\begin{lemma}
  \label{lem:approx-stab}
  Let $R_0,R_1 \in \Lin(\X)$ be two contractions such that $\Re
  \tr_\rho R_i = 1 - \epsilon_i$ for $i=0, 1$.
  Then the product $R_0R_1$ satisfies
  \begin{equation*}
    \Re\tr_{\rho} (R_0R_1) \ge 1 - \epsilon,
  \end{equation*}
  for $\epsilon = \bigl( \epsilon_0^{1/2} + \epsilon_1^{1/2}
  \bigr)^2$.
  As a special case, if $\rho$ is $O(\epsilon)$-stabilized by both
  $R_0$ and $R_1$, it is also $O(\epsilon)$-stabilized by $R_0R_1$.
\end{lemma}

\begin{proof}
  We first prove that
  \begin{equation}
    \label{eq:approx-stab-1}
    \Re\tr_{\rho} (\I-R_0) (\I-R_1) \ge -2
    \epsilon_0^{1/2}\epsilon_1^{1/2}.
  \end{equation}
  In fact, by Cauchy-Schwarz inequality, the absolute value of the
  left hand side is at most
  \begin{equation*}
    \Bigl[ \tr_{\rho} \bigl((\I-R_0)(\I-R_0)^* \bigr) \tr_\rho
    \bigl((\I-R_1)^*(\I-R_1) \bigr) \Bigr]^{1/2},
  \end{equation*}
  which is bounded by $2\epsilon_0^{1/2}\epsilon_1^{1/2}$ using the
  conditions for contractions $R_0, R_1$.
  By Eq.~\eqref{eq:approx-stab-1},
  \begin{equation*}
    \Re \tr_{\rho} (R_0R_1) \ge \Re \tr_{\rho} R_0 + \Re
    \tr_{\rho} R_1 - 1 - 2\epsilon_0^{1/2}\epsilon_1^{1/2} = 1 -
    \bigl( \epsilon_0^{1/2} + \epsilon_1^{1/2} \bigr)^2.
  \end{equation*}
\end{proof}

\begin{lemma}
  \label{lem:approx-stab-2}
  Let $R\in \Lin(\X)$ be a contraction that $\epsilon$-stabilizes
  state $\rho$, then for any contraction $S$,
  \begin{equation*}
    \Re\tr_\rho (SR) \approx_{\sqrt{\epsilon}} \Re\tr_\rho(S).
  \end{equation*}
\end{lemma}

\begin{proof}
  The absolute value of the difference of the two terms in the
  equation is upper bounded by the Cauchy-Schwarz inequality as
  \begin{equation*}
    \abs{\tr_\rho \bigl( S(\I - R) \bigr)} \le \Bigl[ \tr_\rho
    (SS^*) \cdot \tr_\rho \bigl( (\I-R)^*(\I-R) \bigr) \Bigr]^{1/2}
    \le O(\sqrt{\epsilon}).
  \end{equation*}
\end{proof}

\begin{lemma}
  \label{lem:approx-stab-3}
  Let $R_0, R_1, R_3 \in \Lin(\X)$ be contractions such that
  \begin{equation*}
    \Re\tr_\rho (R_0R_1) \approx_\epsilon 1,\quad \Re\tr_\rho (R_1^* R_2)
    \approx_\epsilon 1,
  \end{equation*}
  then
  \begin{equation*}
    \Re\tr_\rho (R_0R_2) \approx_ \epsilon 1.
  \end{equation*}
\end{lemma}

\begin{proof}
  By Lemma~\ref{lem:approx-stab}, we have
  \begin{equation*}
    \Re\tr_\rho (R_0 R_1 R_1^* R_2) \approx_\epsilon 1.
  \end{equation*}
  It therefore suffices to prove that
  \begin{equation*}
    \abs{\tr_\rho \bigl( R_0 (\I - R_1R_1^*) R_2 \bigr)} \le O(\epsilon).
  \end{equation*}
  By the Cauchy-Schwarz inequality, the above term is bounded by
  \begin{equation*}
    \Bigl[ \tr_\rho \bigl( R_0R_0^* - R_0R_1R_1^*R_0^* \bigr) \tr_\rho
    \bigl(R_2^*R_2 - R_2^*R_1R_1^*R_2 \bigr) \Bigr]^{1/2}.
  \end{equation*}

  Applying Lemma~\ref{lem:approx-stab} again to the condition for $R_0,
  R_1$, we have
  \begin{equation*}
    \tr_\rho \bigl( R_0R_1R_1^*R_0^* \bigr) \approx_\epsilon 1.
  \end{equation*}
  By the fact that $R_0$ is a contraction, it implies that
  \begin{equation*}
    \tr_\rho \bigl( R_0R_0^* - R_0R_1R_1^*R_0^* \bigr) \le O(\epsilon).
  \end{equation*}
  A similar bound applies to the second term in the square root and
  this completes the proof.
\end{proof}

In the analysis, it is important to have a quantity characterizing the
approximate commutativity and anti-commutativity of two reflections.
Two reflections $R_0, R_1$ is said to be $\epsilon$-commutative on
state $\rho$ if
\begin{equation}
  \Re\tr_\rho \bigl( R_0R_1R_0R_1 \bigr) \ge 1 - \epsilon,
\end{equation}
and $\epsilon$-anti-commutative on $\rho$ if
\begin{equation}
  \Re\tr_\rho \bigl( R_0R_1R_0R_1 \bigr) \le \epsilon - 1.
\end{equation}

We prove the following lemma which roughly says that
$\epsilon$-anti-commutative reflections are close to $X,Z$
respectively up to a change of basis.
It will be used multiple times in later analysis to establish rigidity
theorems.
\begin{lemma}
  \label{lem:XZ}
  Let $\rho\in \Density(\X)$ be a quantum state and $R_0,R_1\in
  \Herm(\X)$ be two traceless reflections such that
  \begin{equation*}
    \Re \tr_\rho (R_0R_1R_0R_1) \approx_\epsilon -1.
  \end{equation*}
  There exists a unitary $V\in \Lin(\X,\B \otimes \X')$ such that $R_1
  = V^* (Z\otimes \I) V$ and
  \begin{equation*}
    \Re \tr_\rho \bigl( R_0 \, V^* (X\otimes \I) V \bigr)
    \approx_\epsilon 1.
  \end{equation*}
  Equivalently,
  \begin{equation*}
    \dis_\rho \bigl( R_0, V^* (X\otimes \I) V \bigr) \le
    O(\sqrt{\epsilon}).
  \end{equation*}
  The choice of $V$ is independent of state $\rho$ and is determined
  solely by the operators $R_0, R_1$.
\end{lemma}

The proof of the lemma relies on the Jordan's Lemma.

\begin{lemma}[Jordan's Lemma~\cite{Jor75}]
  For any two reflections $R_0,R_1$ acting on a finite dimensional
  Hilbert space $\H$, there exists a decomposition of $\H$ into
  orthogonal one- and two-dimensional subspaces invariant under both
  $R_0$ and $R_1$.
\end{lemma}

\begin{proof}[Proof of Lemma~\ref{lem:XZ}]
  Using Jordan's Lemma and the condition that the reflections are
  traceless, one get simultaneous $2$-by-$2$ block diagonalizations of
  $R_0$ and $R_1$ such that each $2$-by-$2$ block is a reflection
  having both $\pm 1$ eigenvalues.
  Hence, there is a unitary operator $V\in \Lin(\X, \B\otimes \X')$
  such that\footnote{The traceless condition simplifies the discussion
    here.
    Otherwise, the dimension of $\X$ may not be even and one has to
    take $V$ to be an isometry instead of a unitary operator, which
    will in turn make later discussions more complicated.}
  \begin{equation*}
    R_1 = V^* (Z\otimes I) V,
  \end{equation*}
  and
  \begin{equation*}
    R_0 = V^* \sum_l \left[
      \begin{pmatrix}
        \cos\theta_l & \phantom{-}\sin\theta_l\\
        \sin\theta_l & -\cos\theta_l
      \end{pmatrix}
      \otimes \ket{l}\bra{l}\right] V,
  \end{equation*}
  where $\theta_l \in [0,\pi]$ and $l$ is the index of the
  two-dimensional invariant subspaces obtained by Jordan's lemma.

  By a direct calculation, we have
  \begin{equation*}
    R_0 R_1 R_0 R_1 = V^* \sum_l \left[
      \begin{pmatrix}
        1 - 2\sin\theta_l^2 & -2\cos\theta_l\sin\theta_l\\
        2\cos\theta_l\sin\theta_l & 1 - 2\sin\theta_l^2
      \end{pmatrix}
      \otimes \ket{l}\bra{l}\right] V.
  \end{equation*}

  The condition, $\Re\tr_\rho \bigl( R_0R_1R_0R_1 \bigr)
  \approx_\epsilon -1$, then simplifies to
  \begin{equation}
    \label{eq:XZ-1}
    \E_{l} \sin\theta_l^2 \approx_\epsilon 1,
  \end{equation}
  where the expectation $\E_l$ is over the probability distribution
  \begin{equation*}
    \Pr(l) = \tr_{\rho} \bigl[ V^* (\I\otimes \ket{l}\bra{l}) V\bigr].
  \end{equation*}

  We then have,
  \begin{equation*}
    \Re\tr_\rho \bigl( R_0 \, V^*(X\otimes \I)V \bigr) = \E_l \sin\theta_l
    \ge \E_l \sin\theta_l^2 \approx_\epsilon 1,
  \end{equation*}
  which completes the proof by the definition of $\dis_\rho$ for two
  reflections.
\end{proof}

\subsection{Rigidity Using Extended Nonlocal Games}

Rigidity of a nonlocal game states that if the players win the game
with probability that is close to optimal, then they have to
approximately follow the optimal strategy up to an isometry, including
the initialization of a shared state and the application of the
quantum measurement for each question.
It has found a wide range of applications in self-testing of quantum
apparatus~\cite{MY98,DMMS00,MS12,MYS12} and quantum multi-player
interactive proofs~\cite{RUV13,Ji16,NV15}.

We demonstrate that it is easier both to construct extended nonlocal
games that have rigidity properties and to establish rigidity for
them.
Extended nonlocal games are in some sense a variant of the nonlocal
games with one honest player, who honestly follows a prescribed
measurement strategy.
This may explain the reason behind the advantages of the use of the
extended nonlocal games.

In~\cite{Ji16}, the CHSH game was revisited in the framework of
stabilizers starting from the fact that the EPR state is stabilized by
$XX$ and $ZZ$.
There, one need to rotate the basis for one of players by $45$ degree,
a mysterious twist that one has to perform in the case for nonlocal
games.
Consider the following extended nonlocal game for the EPR stabilizer
between the referee and one player that literally translates the
generators of stabilizer into random questions.
The referee possesses a single qubit register $\reg{B}$ and samples a
random bit $q\in \{0,1\}$, and send it to the player.
He then measures $X$ or $Z$ on his qubit for $q=0$ or $q=1$
respectively and accepts if and only if the measurement outcome equals
the answer bit $a$ from the player.
The game value for this simple extended nonlocal game is one, which
can be achieved by a player who shares the EPR state and measures
$X,Z$ for question $0,1$ respectively.
The strategy of the player can be described by the tuple $(\rho,
\hat{X}, \hat{Z})$ where $\rho\in \Density(\B\otimes \R)$ is the state
the player chooses, and $\hat{X}$, $\hat{Z} \in \Herm(\R)$ are
traceless reflections that describe the player's two-outcome
measurements for question $0$ and $1$ respectively.
The value of the strategy is
\begin{equation*}
  \frac{1}{2}\sum_{D\in \{X,Z\}} \tr_\rho \frac{\I + D \otimes
    \hat{D}}{2}.
\end{equation*}
If the value is at least $1 - \epsilon$, it then follows that
\begin{equation*}
  \tr_\rho \bigl( D \otimes \hat{D} \bigr) \approx_\epsilon 1,
\end{equation*}
for $D = X, Z$.
By Lemma~\ref{lem:approx-stab}, we have
\begin{equation*}
  \Re \tr_\rho \bigl( XZXZ\otimes \hat{X} \hat{Z} \hat{X} \hat{Z}
  \bigr) \approx_\epsilon 1,
\end{equation*}
which simplifies to
\begin{equation*}
  \Re \tr\rho \bigl( \hat{X} \hat{Z} \hat{X} \hat{Z} \bigr)
  \approx_\epsilon -1.
\end{equation*}
Lemma~\ref{lem:XZ} then establishes the rigidity for this game.

A similar construction based on the stabilizer for the GHZ state is
used in Sec.~\ref{sec:proof} to check the correct propagation of the
provers' actions.

\section{Stabilizer Games, Redefined}
\label{sec:stab}

\subsection{Stabilizer Games}

Stabilizer games were first defined in~\cite{Ji16} as an extension of
the CHSH game\cite{CHSH69}.
In this section, we introduce yet another type of stabilizer games.
Their advantage over the stabilizer games defined in~\cite{Ji16} is
that they have perfect quantum strategies, a property that is crucial
to obtain perfect completeness.
The analysis of this new stabilizer game is also arguably simpler.

Consider the stabilizer in Fig.~\ref{fig:8code}.
It is an eight-qubit code encoding two logical qubits and has distance
two, and we will refer to it as the eight-qubit code in this paper.
The operators $g_1, g_2, \ldots, g_6$ are the generators for the
stabilizer.
The operators $L_X$ and $L_Z$ are the logical $X,Z$ operators for one
of the logical qubits.
One can derive this code by concatenating the $[4,2,2]$ code
stabilized by $X^{\otimes 4}, Z^{\otimes 4}$ and the $[2,1,1]$ code
stabilized by $Y^{\otimes 2}$.
We note that the construction of the stabilizer game generalizes to
other stabilizer codes with generators of $XZ$-form.
It suffices for our purpose to consider the game defined by the
eight-qubit code only.

The stabilizer of the eight-qubit code consists of $64$ operators,
$\prod_{i=1}^6 g_i^{\mu_i}$, for $\mu_i \in \{0,1\}$.
There are $32$ of them that have $XZ$-form.
These are the operators when one and only one of $\mu_1, \mu_2$ is
$1$.
Let $\Xi$ be the set of these $XZ$-form operators.
Examples of operators in $\Xi$ contain $g_1, g_2$ and $g_{1,3}=g_1g_3,
g_{2,3}=g_2g_3$ in Fig.~\ref{fig:8code-XZ}.
Note that we have listed $g_2$ before $g_1$ in Fig.~\ref{fig:8code-XZ}
on purpose for reasons to be clear later.

\begin{figure}[!htb]
  \centering
  \begin{tabular}[c]{c|c@{}c@{}c@{}c@{}c@{}c@{}c@{}c}
    \hline\hline
    Name & \multicolumn{8}{c}{Operator}\\
    \hline
    $g_1$ & $X$ & $X$ & $X$ & $X$ & $X$ & $X$ & $X$ & $X$\\
    $g_2$ & $X$ & $Z$ & $X$ & $Z$ & $X$ & $Z$ & $X$ & $Z$\\
    $g_3$ & $Y$ & $Y$ & $I$ & $I$ & $I$ & $I$ & $I$ & $I$\\
    $g_4$ & $I$ & $I$ & $Y$ & $Y$ & $I$ & $I$ & $I$ & $I$\\
    $g_5$ & $I$ & $I$ & $I$ & $I$ & $Y$ & $Y$ & $I$ & $I$\\
    $g_6$ & $I$ & $I$ & $I$ & $I$ & $I$ & $I$ & $Y$ & $Y$\\
    \hline
    $L_X$ & $X$ & $X$ & $X$ & $X$ & $I$ & $I$ & $I$ & $I$\\
    $L_Z$ & $X$ & $Z$ & $I$ & $I$ & $X$ & $Z$ & $I$ & $I$\\
    \hline\hline

  \end{tabular}
  \caption{An eight-qubit stabilizer code used in the stabilizer game
    in Fig.~\ref{fig:stab-game}.}
  \label{fig:8code}
\end{figure}

\begin{figure}[!htb]
  \centering
  \begin{tabular}[c]{c|r@{}c@{}c@{}c@{}c@{}c@{}c@{}c@{}c}
    \hline\hline
    Name & & \multicolumn{8}{c}{Operator}\\
    \hline
    $g_2$ & & $X$ & $Z$ & $X$ & $Z$ & $X$ & $Z$ & $X$ & $Z$\\
    $g_1$ & & $X$ & $X$ & $X$ & $X$ & $X$ & $X$ & $X$ & $X$\\
    \hline
    $g_{1,3}$ & ${-}$ & $Z$ & $Z$ & $X$ & $X$ & $X$ & $X$ & $X$ & $X$\\
    $g_{2,3}$ & & $Z$ & $X$ & $X$ & $Z$ & $X$ & $Z$ & $X$ & $Z$\\
    \hline\hline

  \end{tabular}
  \caption{Four examples of $XZ$-form operators in the stabilizer of
    the eight-qubit code}
  \label{fig:8code-XZ}
\end{figure}

The stabilizer game considered in this paper is an eight-player XOR
game as specified in Fig.~\ref{fig:stab-game}.
We abuse the notion and use the Pauli operators $X,Z$ as labels for
the questions.
Intuitively, an $X$ question requests that the player measures Pauli
$X$ operator on his system and reply with the outcome and similarly
for $Z$ questions.

\begin{figure}[!htb]
  \begin{shaded}
    \ul{Stabilizer Game}\\[1em]
    Let $\Xi$ be the subset of stabilizer operators of $XZ$-form for
    the eight-qubit code.
    The stabilizer game for the eight-qubit code is the eight-player
    nonlocal game defined as follows.
    \begin{enumerate}
    \item The referee selects one of the $32$ operators from $\Xi$
      uniformly at random.
      Let $D^{(i)} \in \{X,Z\}$, $s\in \{0,1\}$ be the $i$-th tensor
      factor and the sign of the chosen operator respectively.
    \item For $i\in [8]$, the referee sends $D^{(i)}$ to player $(i)$
      and receive a bit $a^{(i)}$ back;
    \item Accepts if $\bigoplus_{i=1}^8 a^{(i)} = s$ and rejects
      otherwise.
    \end{enumerate}
  \end{shaded}
  \caption{Stabilizer game defined by the eight-qubit code in
    Fig.~\ref{fig:8code}.}
  \label{fig:stab-game}
\end{figure}

A strategy of the stabilizer is specified by the state $\rho$ shared
between the eight players and the measurement the players perform for
question indexed by $X,Z$.
We will use $\hat{D}^{(i)}$ to denote the reflections that correspond
to the measurement player $(i)$ performs for question $D\in\{X,Z\}$.
Without loss of generality, we assume that $\hat{D}^{(i)}$ are
traceless reflections.
When there is no ambiguity, the player index in the superscript may be
omitted.

We prove the following rigidity theorem for the stabilizer game.

\begin{theorem}
  \label{thm:stab-game}
  The nonlocal value of stabilizer game in Fig.~\ref{fig:stab-game} is
  $1$.
  Furthermore, the game has the following rigidity property.
  Let $\Strategy = \bigl(\rho, \bigl\{ \hat{D}^{(i)} \bigr\} \bigr)$
  be a strategy for the stabilizer game where $\rho$ is the state
  shared between the players before the game starts and
  $\hat{D}^{(i)}\in \Herm(\R_i)$ is the traceless reflection
  corresponding to the measurements the player $(i)$ performs for
  question $D \in \{X,Z\}$.
  If the value of strategy $\Strategy$ is at least $1 - \epsilon$,
  then there are unitary operators $V_i\in \Lin \bigl(\R_i, \B_i
  \otimes \R_i' \bigr)$ for $i\in [8]$, such that the following
  properties hold
  \begin{itemize}
  \item For all $i\in [8]$, $\hat{Z}^{(i)} = \check{Z}^{(i)}$ and
    \begin{equation*}
      \dis_\rho \bigl( \hat{X}^{(i)}, \check{X}^{(i)} \bigr) \le
      O(\sqrt{\epsilon}),
    \end{equation*}
    where $\check{D}^{(i)} = V_i^*(D\otimes \I) V_i$ for $D\in
    \{X,Z\}$.

  \item Let $\Pi$ be the projection to the code space of the eight
    qubit code, and let $V$ be the unitary operator
    $\bigotimes_{i=1}^8 V_i$, then
    \begin{equation*}
      \ip{\Pi\otimes \I}{V\rho V^*} \ge 1-O(\epsilon),
    \end{equation*}
    where $\Pi$ acts on the eight qubits in registers $\bigl(
    \reg{B}_i \bigr)_{i=1}^8$.
  \end{itemize}
\end{theorem}

\begin{proof}
  It is obvious that if the players share an encoded state of the
  eight qubit code and perform the $X,Z$ measurements to obtain the
  answers for questions $X,Z$ respectively, the referee accepts with
  certainty.

  For each operator Pauli $P \in \Xi$ of the form
  \begin{equation*}
    P = (-1)^{\nu} \bigotimes_{i=1}^8 D^{(i)},
  \end{equation*}
  where $D^{(i)} \in \{X,Z\}$, define reflection
  \begin{equation*}
    \hat{P} = (-1)^{\nu} \bigotimes_{i=1}^8 \hat{D}^{(i)},
  \end{equation*}
  by replacing the $X$ and $Z$'s with $\hat{X}^{(i)}$ and
  $\hat{Z}^{(i)}$ from the strategy respectively.
  For a strategy $\Strategy = \bigl(\rho, \bigl\{ \hat{D}^{(i)}
  \bigr\} \bigr)$, its value of the game can be expressed as
  \begin{equation}
    \label{eq:stab-game-value}
    \frac{1}{32} \sum_{P\in \Xi} \tr_\rho \frac{\I + \hat{P}}{2}.
  \end{equation}
  If strategy $\Strategy$ has value $1-\epsilon$, we have
  \begin{equation*}
    \tr_\rho \hat{P} \approx_\epsilon 1,
  \end{equation*}
  for each operator $P\in \Xi$.

  These conditions for the four operators in Fig.~\ref{fig:8code-XZ}
  and a repeated application of Lemma~\ref{lem:approx-stab} conclude
  that
  \begin{equation*}
    \Re\tr_\rho \bigl( \, \hat{g_2} \, \hat{g_1} \, \hat{g_{1,3}} \,
    \hat{g_{2,3}} \bigr) \ge 1 - O(\epsilon).
  \end{equation*}
  Observing that, as in Fig.~\ref{fig:8code-XZ}, all the reflections
  cancel out for all players except those for player $(2)$, we have
  \begin{equation*}
    \Re\tr_\rho \bigl( \hat{Z}^{(2)} \hat{X}^{(2)} \hat{Z}^{(2)}
    \hat{X}^{(2)} \bigr) \approx_\epsilon -1.
  \end{equation*}
  Lemma~\ref{lem:XZ} then proves the first item for $i=2$.
  A similar argument proves the case for $i=1$ by considering
  \begin{equation*}
    \Re\tr_{\rho} \bigl( \, \hat{g_2} \, \hat{g_{1,3}} \, \hat{g_1} \,
    \hat{g_{2,3}} \bigr).
  \end{equation*}
  The symmetry of the game then completes the proof of the first item
  for all $i\in [8]$.

  To prove the second item, consider strategy $\check{\Strategy} =
  \bigl( \rho, \bigl\{ \check{D}^{(i)} \bigr\} \bigr)$ where
  $\check{D}^{(i)}$ are reflections as defined in the first item of
  the theorem for $D\in \{X,Z\}$.
  By the claim in the first item, Lemma~\ref{lem:approx-stab} and the
  expression for the game value in Eq.~\eqref{eq:stab-game-value}, it
  is easy to see that the value of strategy $\check{\Strategy}$ is at
  least $1 - O(\epsilon)$.
  That is
  \begin{equation*}
    \frac{1}{32} \tr_{\rho'} \sum_{P\in \Xi} P \ge 1 - O(\epsilon)
  \end{equation*}
  where $\rho' = V\rho V^*$.
  It is easy to see that the operator $\sum_{P\in \Xi} P$ has the code
  space as its eigenspace of eigenvalue $32$, and all other
  eigenvalues are at most $0$.
  It then follows that $\sum_{P\in \Xi} P \le 32\Pi$, where $\Pi$ is
  the projection to the code space.
  A direct calculation then proves the second item of the theorem.
\end{proof}

\subsection{Multi-Qubit Stabilizer Game}

In this section, we consider a multi-qubit variant of the stabilizer
game called the $(n,k)$-stabilizer game.
It is again an eight-player game and the referee sends questions in
the form of measurement instructions to the players.
The players are expected to hold a quantum register of $n$ qubits and
follow the measurement instructions that encode what quantum
measurements the honest players are supposed to perform.

A $(n,k)$-stabilizer game was defined in~\cite{Ji16} where the players
receive instructions of at most $k$ single-qubit measurement.
The $(n,k)$-stabilizer game considered here is more general in the
sense that the measurement instructions to the player may include a
set of $k$ pairwise commuting $XZ$-form Pauli operators of weight at
most $k$.
For example, in the case of $k=2$, a possible question may be
$\{X_1X_2, Z_1Z_2\}$, asking the player to measure both $X_1X_2$ and
$Z_1Z_2$ simultaneously, and reply with the two measurement outcome
bits.
In this section, we will follow the convention that the subscripts for
$X,Z$ are the index for the qubits these operators act on.

Let $\Pauli_{n,k}$ be the set of $XZ$-form Pauli operators on $n$
qubits of weight at most $k$.
Let $\Power_{n,k}$ be the collection of size-$k$ subsets of
$\Pauli_{n,k}$ of pairwise commuting operators, each of which acts on
the same set of $k$ qubits.
The sizes of $\Pauli_{n,k}$ and $\Power_{n,k}$ are at most polynomial
in $n$ for constant $k$.
The measurement specification that the players receive will be either
a single Pauli operator $P \in \Pauli_{n,k}$ or a set $Q \in
\Power_{n,k}$ .
In the first case, the player is supposed to measure $P$ and respond
with a single bit, while in the latter case, the player is supposed to
measure all operators in $Q$ and reply with $k$ outcome bits.
Pauli operators of weight one in $\Pauli_{n,1}$ is usually denoted as
$D_u$ for qubit index $u\in [n]$ and $D_u\in \{X_u, Z_u\}$.

\begin{figure}[!htb]
  \begin{shaded}
    \ul{Multi-Qubit Stabilizer Game}\\[1em]
    Let $[n]$ be the index of $n$ qubits and let $k\ge 2$ be a
    constant.
    The $(n,k)$-stabilizer game is an eight-player nonlocal game where
    the referee does the following with equal probability:
    \begin{enumerate}
    \item \game{Stabilizer Check}.
      The referee plays the stabilizer game on a randomly selected
      qubit.
      That is, he
      \begin{enumerate}
      \item Samples a qubit $u\in [n]$ uniformly at random; samples
        $D^{(i)}_u \in \{X_u,Z_u\}$ as in the stabilizer game.
      \item Sends $D^{(i)}_u$ to player $(i)$ and receive an answer
        bit $a^{(i)}$.
      \item Accepts if the referee for the stabilizer game accepts on
        questions $D^{(i)}_u$ and answers $a^{(i)}$.
      \end{enumerate}

    \item \game{Confusion Check}.
      The referee plays the stabilizer game but confuses one of the
      players by hiding the index of the qubit the stabilizer game
      checks against in a set of qubits.
      \begin{enumerate}
      \item The referee selects a subset $J\subset [n]$ of size $k$,
        an index $u\in J$, and a player $t\in [8]$, all uniformly at
        random.
        For each qubit $v\in J$, indecently samples questions
        $D_v^{(i)}\in \{X_v,Z_v\}$ as in the stabilizer game.
      \item Sends $D^{(i)}_u$ to player $(i)$ and receives an answer
        bit $a^{(i)}$ if $i\ne t$; sends $Q = \bigl\{ D^{(t)}_v
        \bigr\}_{v\in J}$ to player $(t)$, and receives a $k$-bit
        string $b=(b_v)_{v\in J}$.
        Define $a^{(t)} = b_u$ and $a = \bigl( a^{(1)}, a^{(2)},
        \ldots, a^{(r)} \bigr)$.
      \item Accepts if and only if the referee for the stabilizer game
        accepts when the questions are $D_u$ and answer bits are $a$.
      \end{enumerate}

    \item \game{Parity Check}.
      The referee tests the consistency of a multi-qubit $XZ$-form
      Pauli measurement and individual $X$ and $Z$ measurements.
      \begin{enumerate}
      \item Samples $P\in \Pauli_{n,k}$ and $t\in [8]$ uniformly at
        random.
        Let $J$ be the support of $P$ and $P=\prod_{v\in J} D_v$ for
        $D_v\in \{X_v,Z_v\}$.
        If $\abs{J} < k$, randomly chooses $J'\supset J$ of size $k$
        and randomly choose $D_v\in \{X_v,Z_v\}$ for $v\in J'\setminus
        J$.
      \item Sends $Q=\{D_v\}_{v\in J'}$ to player $(i)$ and receives
        $k$ bits $\bigl( a_v^{(i)} \bigr)_{v\in J'}$ if $i\ne t$ ;
        sends $P$ to player $(t)$ and receive a bit $a^{(t)}$.
        Define $a^{(i)} = \bigoplus_{v\in J} a_v^{(i)}$ for $i\ne t$.
      \item Accepts if $\bigoplus_{i=1}^8 a^{(i)} = 0$; rejects
        otherwise.
      \end{enumerate}

    \item \game{Pauli Check}.
      The referee tests the consistency of answers between multiple
      and single Pauli questions.
      \begin{enumerate}
      \item Samples $Q \in \Power_{n,k}$, $P\in Q$ and $t\in [8]$
        uniformly at random.
      \item Sends $P$ to player $(i)$ and receives a bit $a^{(i)}$ if
        $i\ne t$ ; sends $Q$ to player $(t)$ and receive a $k$-bit
        answer $\bigl( a_{P'}^{(t)} \bigr)_{P' \in Q}$.
        Define $a^{(t)} = a_P^{(t)}$.
      \item Accepts if $\bigoplus_{i=1}^8 a^{(i)} = 0$; rejects
        otherwise.
      \end{enumerate}
    \end{enumerate}
  \end{shaded}
  \caption{Multi-qubit stabilizer game.}
  \label{fig:ms-game}
\end{figure}

The $(n,k)$-stabilizer game is given in Fig.~\ref{fig:ms-game}.
It is easy to see that the nonlocal value of the game equals one,
which can be achieved by players who share an correctly encoded state
and follow measurement specifications honestly.
Let $\R_i$ be the state space of player $(i)$.
A strategy for the $k$-qubit stabilizer game,
\begin{equation*}
  \Strategy = \bigl(\rho, \bigl\{ R_P^{(i)} \bigr\}, \bigl\{
  M^{(i)}_Q \bigr\} \bigr),
\end{equation*}
consists of a state $\rho \in \Density\bigl( \bigotimes_{i=1}^r \R_i
\bigr)$, reflections $R_P^{(i)}$ the players measure for question $P$
and measurements $M^{(i)}_Q$ with $k$-bit outcomes for question $Q$.
The superscripts of the measurements indexing the players are
sometimes omitted if there will be no ambiguity.

The reflection $R_P$ and measurement $M_Q$ in the strategy will
sometimes be denoted as $\hat{P}$ and $\hat{Q}$ respectively.
Without loss of generality, it is assumed that the measurements
$\hat{Q}$ are projective measurements and each measurement operator
has the same rank.
For $P\in Q$, define the derived reflections
\begin{equation*}
  \hat{P|Q} = \sum_{b\in \{0,1\}^Q} (-1)^{b_P} \hat{Q}^{\,b}.
\end{equation*}
By the assumption on measurement $\hat{Q}$, the derived reflections
are traceless.

We prove the following rigidity property of the $(n,k)$-stabilizer
game.

\begin{theorem}
  \label{thm:ms-game}
  For any constant integer $k \ge 2$, there exists a constant
  $\kappa>0$ that depends only on $k$ such that the $(n,k)$-stabilizer
  game in Fig.~\ref{fig:ms-game} has the following rigidity property.
  For any quantum strategy $\Strategy = \bigl(\rho, \bigl\{
  \hat{P}^{(i)} \bigr\}, \bigl\{ \hat{Q}^{(i)} \bigr\} \bigr)$ that
  has value at least $1 - \epsilon$, there are isometries $V_i \in
  \Lin(\R_i,\B_i^{\otimes n} \otimes \R_i')$, such that the following
  properties hold
  \begin{itemize}
  \item For all $i\in[8]$, $P \in \Pauli_{n,k}$, and $Q \in
    \Power_{n,k}$,
    \begin{subequations}
      \begin{align}
        \dis_\rho \bigl( \hat{P}^{(i)}, \check{P}^{(i)} \bigr)
        & \le O(n^\kappa \epsilon^{1/\kappa}),\label{eq:ms-1}\\
        \dis_\rho \bigl( \hat{Q}^{(i)}, \check{Q}^{(i)} \bigr)
        & \le O(n^\kappa \epsilon^{1/\kappa}),\label{eq:ms-2}
      \end{align}
    \end{subequations}
    where $\check{P}^{(i)} = V_i^* (P\otimes \I) V_i$ and
    $\check{Q}^{(i)}$ is the measurement that first performs isometry
    $V_i$ and then measures the $k$ Pauli operators in $Q$.
  \item Let $\Pi$ be the projection to the code space of the
    stabilizer code, $V$ be the isometry $\bigotimes_{i=1}^r V_i$,
    then
    \begin{equation}
      \label{eq:ms-3}
      \ip{\Pi^{\otimes n}\otimes I}{V\rho V^*} \ge
      1-O(n^\kappa \epsilon^{1/\kappa}),
    \end{equation}
    where the $t$-th tensor factor of $\; \Pi^{\otimes n}$ acts on
    eight qubits, each of which is the $t$-th qubit of each player's
    system after the application of $V$.
  \end{itemize}
\end{theorem}

The proof of Theorem~\ref{thm:ms-game} relies on the following lemmas.

\begin{lemma}
  \label{lem:commute}
  Let $\rho \in \Density(\X\otimes \Y)$ be a quantum state, $R_1, R_2,
  S_1, S_2 \in \Herm(\X)$ be four reflections on $\reg{X}$, and $U_1,
  U_2\in \Herm(\Y)$ be two reflections on $\reg{Y}$.
  If $S_1, S_2$ commute, both $R_1, S_1$ are $\epsilon$-consistent
  with $U_1$, and both $R_2, S_2$ are $\epsilon$-consistent with
  $U_2$, then
  \begin{equation}
    \Re\tr_\rho \bigl( R_1R_2R_1R_2 \bigr) \approx_\epsilon 1.
  \end{equation}
\end{lemma}

\begin{proof}
  First, by the commutativity of $S_1, S_2$, we have
  \begin{equation*}
    \Re \tr_\rho (S_1 S_2 S_1 S_2) = 1.
  \end{equation*}
  Lemma~\ref{lem:approx-stab} and the $\epsilon$-consistency of
  $R_1,S_1$ with $U_1$ then imply
  \begin{equation*}
    \Re\tr_\rho \Bigl[ \bigl( S_1\otimes U_1\bigr)
    \bigl((S_1S_2S_1S_2) \otimes \I_{\reg{Y}} \bigr) \bigl(R_1\otimes
    U_1\bigr) \Bigr] \approx_\epsilon 1,
  \end{equation*}
  which simplifies to
  \begin{equation*}
    \Re\tr_\rho \bigl( S_2S_1S_2R_1 \bigr) \approx_\epsilon 1.
  \end{equation*}
  That is, we can move $S_1$ in the front to the end and replace it
  with $R_1$ without causing too much error in the expression.
  Repeating similar arguments three more times, we have
  \begin{equation*}
    \Re\tr_\rho \bigl( R_1R_2R_1R_2 \bigr) \approx_\epsilon 1.
  \end{equation*}
\end{proof}

\begin{lemma}
  \label{lem:W}
  Let $\X$, $\Y$, $\B$ be two-dimensional Hilbert spaces.
  Let $V\in \Lin(\R, \B \otimes \R')$ be a unitary operator, $R \in
  \Lin(\R)$ be any operator and $\ket{\Phi}$ be the EPR state on
  $\X\otimes \Y$.
  Define isometry $W\in \Lin(\R, \X\otimes \Y\otimes \R)$ as
  \begin{equation*}
    W = (\I \otimes V^*) \SWAP (\ket{\Phi} \otimes V),
  \end{equation*}
  where the $\SWAP$ acts on $\X$ and $\B$.
  Then
  \begin{equation*}
    W^* R W = \frac{1}{4}\sum_{i=0}^3 \bigl( V^*(\sigma_i \otimes \I)V
    \bigr) \,R\, \bigl( V^*(\sigma_i\otimes \I)V \bigr),
  \end{equation*}
  where $\sigma_0, \sigma_1, \sigma_2, \sigma_3$ are the Pauli
  operators.
\end{lemma}

\begin{lemma}
  \label{lem:overline}
  Let $R_1, R_2, \ldots, R_k\in \Herm(\Y)$ be $k$ pairwise commuting
  reflections, $V\in \Lin(\X,\Y)$ an isometry, $\rho \in
  \Density(\X\otimes \Z)$ a quantum state.
  Define operators $\check{R}_i = V^* R_i V \in \Herm(\X)$.
  If $\check{R}_i$ has $\epsilon$-consistent reflections on $\Z$ for
  $i\in [k]$, then
  \begin{equation*}
    \Re\tr_\rho \biggl[ V^* \Bigl( \prod_{i=1}^k R_i \Bigr) V \,
    \prod_{i=1}^k \check{R}_i \biggr]
    \approx_\epsilon 1.
  \end{equation*}
\end{lemma}

\begin{proof}
  We prove by induction on $k$.
  For $k=1$, the claim follows from the fact that $\check{R}_1$ has
  $\epsilon$-consistent reflections and Lemma~\ref{lem:approx-stab}.
  We now assume that the claim holds for $k-1$ and prove it for $k$.
  By the induction hypothesis and Lemma~\ref{lem:approx-stab},
  \begin{equation*}
    \Re\tr_\rho \biggl[ \check{R}_k \, V^* \Bigl(\prod_{i=1}^{k-1} R_i
    \Bigr) V \, \Bigl( \prod_{i=1}^{k-1} \check{R}_i \Bigr) \,
    \check{R}_k \biggr] \approx_\epsilon 1.
  \end{equation*}
  It therefore suffices to prove that the difference on the left hand
  sides of the above equation and the equation in the lemma is at most
  $O(\epsilon)$.
  The absolute value of the difference can be bounded by
  Cauchy-Schwarz inequality as follows
  \begin{equation}
    \label{eq:VV*}
    \begin{split}
      & \abs{\tr_\rho \biggl[ V^* \Bigl( R_k (\I-VV^*)
        \prod_{i=1}^{k-1} R_i \Bigr) V \, \Bigl( \prod_{i=1}^k
        \check{R}_i \Bigr) \biggr]}\\
      & \qquad \le \Bigl( 1 - \tr_\rho \check{R}_k^2 \Bigr)^{1/2}
      \biggl[ 1 - \tr_\rho \Bigl\lvert V^* \Bigl( \prod_{i=1}^{k-1}
      R_i \Bigr) V \Bigl( \prod_{i=1}^k \check{R}_i \Bigr)
      \Bigr\rvert^2 \biggr]^{1/2}.
    \end{split}
  \end{equation}
  It follows from the induction hypothesis and
  Lemma~\ref{lem:approx-stab} that both terms in the sequare roots on
  the right hand side are at most $O(\epsilon)$.
\end{proof}

\begin{lemma}
  \label{lem:derived}
  Let $M = \bigl\{ M^a \bigr\}$ be a projective measurement of $k$-bit
  outcome on quantum register $\reg{X}$ and $R_1, R_2, \ldots, R_k$ be
  its derived reflections.
  Let $N = \bigl\{ N^a \bigr\}$ be a projective measurement of $k$-bit
  outcome on quantum register $\reg{Y}$ and $S_1, S_2, \ldots, S_k$ be
  its derived reflections.
  Let $V\in \Lin(\X,\Y)$ be an isometry and $\rho\in
  \Density(\X\otimes \Z)$ a quantum state.
  For $i\in [k]$, define $\check{S}_i = V^* S_i V \in \Herm(\X)$.
  Let $\check{N}$ be the quantum measurement that measures $N$ after
  the application of isometry $V$.

  If $\Re \tr_\rho \bigl( R_i \check{S_i} \bigr) \approx_\epsilon 1$
  and $R_i$ has $\epsilon$-consistent reflections on $\Z$ for $i\in
  [k]$, then
  \begin{equation*}
    \dis_\rho (M, \check{N}) \le O(\epsilon^{1/2}).
  \end{equation*}
\end{lemma}

\begin{proof}
  By the definition of measurement $\check{N}$, we have
  \begin{equation*}
    \begin{split}
      \check{N}^a & = V^* \biggl[ \prod_{i=1}^k \frac{\I +
        (-1)^{a_i}S_i}{2} \biggr] V\\
      & = \frac{1}{2^k} \sum_{x\in \{0,1\}^k} (-1)^{\ip{a}{x}} V^*
      \Bigl( \prod_{i=1}^k S_i^{x_i} \Bigr) V.
    \end{split}
  \end{equation*}
  This implies that
  \begin{equation*}
    \begin{split}
      \sum_a \Re \tr_\rho (M^a \check{N}^a) & = \frac{1}{2^{2k}}
      \sum_{a,x,y} (-1)^{\ip{a}{x\oplus y}} \Re \tr_\rho \biggl[
      \Bigl(\prod_{i=1}^k R_i^{x_i} \Bigr) \Bigl( V^* \prod_{i=1}^k
      S_i^{y_i} V\Bigr) \biggr]\\
      & = \frac{1}{2^k} \sum_x \Re \tr_\rho \biggl[ \Bigl(
      \prod_{i=1}^k R_i^{x_i} \Bigr) V^* \Bigl( \prod_{i=1}^k
      S_i^{x_i} \Bigr) V \biggr].
    \end{split}
  \end{equation*}
  Note that in the this proof, the superscript of an operator
  represents the corresponding power of the operator and is not the
  index for measurement outcome as in the other parts of the paper.

  We claim that for all $x\in \{0,1\}^k$, the term in the summand
  \begin{equation*}
    \Re \tr_\rho \biggl[ \Bigl(\prod_{i=1}^k R_i^{x_i} \Bigr) V^*
    \Bigl( \prod_{i=1}^k S_i^{x_i} \Bigr) V \biggr] \approx_\epsilon 1.
  \end{equation*}
  This will concludes the proof by the definition of $\dis_\rho$ for
  POVMs by choosing $\{VM^a\}$ and
  \begin{equation*}
    \biggl\{ \biggl[ \prod_{i=1}^k \frac{\I + (-1)^{a_i}S_i}{2}
    \biggr] V \biggr\}
  \end{equation*}
  as the measurement operators for the two POVMs $M$ and $\check{N}$
  respectively.
  We prove this claim by an induction on $k$.
  For $k=1$, the claim is exactly the condition $\Re \tr_\rho \bigl(
  R_1 \check{S}_1 \bigr) \approx_\epsilon 1$.
  Assume now the claim holds for $k-1$, and we prove the case for $k$.

  We have
  \begin{subequations}
    \begin{align}
      \Re \tr_\rho \biggl[ \Bigl(\prod_{i=1}^k R_i^{x_i} \Bigr) V^*
      \Bigl(\prod_{i=1}^{k-1} S_i^{x_i} \Bigr) V R_k^{x_k} \biggr]
      & \approx_\epsilon 1,\\
      \Re \tr_\rho \biggl[ \Bigl(\prod_{i=1}^k R_i^{x_i} \Bigr) V^*
      \Bigl(\prod_{i=1}^{k-1} S_i^{x_i} \Bigr) V \check{S}_k^{x_k}
      \biggr] & \approx_\epsilon 1,
    \end{align}
  \end{subequations}
  where the first approximation follows from
  Lemma~\ref{lem:approx-stab} and the induction hypothesis, the second
  approximation follows from the condition $\Re \tr_\rho \bigl( R_k
  \check{S}_k \bigr) \approx_\epsilon 1$.
  Then, by the use of Cauchy-Schwarz inequality as in
  Eq.~\eqref{eq:VV*}, it follows that
  \begin{equation*}
    \Re \tr_\rho \biggl[ \Bigl(\prod_{i=1}^k R_i^{x_i} \Bigr) V^*
    \Bigl(\prod_{i=1}^{k} S_i^{x_i} \Bigr) V \biggr]
    \approx_\epsilon 1.
  \end{equation*}
\end{proof}

\begin{proof}[Proof of Theorem~\ref{thm:ms-game}]
  We first prove the claims in Eqs.~\eqref{eq:ms-1}
  and~\eqref{eq:ms-2} of the theorem for Pauli operators $D_u\in
  \Pauli_{n,k}$ of weight one and $Q\in \Power_{n,k}$ containing $k$
  such operators acting on $k$ different qubits.
  For this, we only need to consider the first two tests of the game
  and the analysis is almost the same as the proof for the multi-qubit
  stabilize game defined in~\cite{Ji16}.
  We prove the claim for player $(t)$, which suffices to conclude the
  claims for all players by the symmetry of the game.

  In the proof, we will refer to operators with different accents.
  Reflections with a hat are from the measurement strategy of the
  players, reflections with a tilde represent intermediate operators,
  and operators with a check correspond honest players' measurements
  up to an isometry.
  Our goal is therefore bound the distance between operators with a
  hat to those with a check.
  Keeping this convention in mind may help understanding the proof.

  As the strategy $\Strategy$ has value at least $1-\epsilon$, the
  referee rejects in the \game{Stabilizer Check} with probability at
  most $4\epsilon$.
  It follows that, for $u \in [n]$, strategy $\Strategy_u =
  \bigl(\rho, \bigl\{ \hat{D}_u^{(i)} \bigr\} \bigr)$ has value at
  least $1-\epsilon_1$ in the eight-qubit stabilizer game in
  Fig.~\ref{fig:stab-game} for $\epsilon_1 = 4n\epsilon$.
  By the rigidity of the stabilizer game
  (Theorem~\ref{thm:stab-game}), there exist unitary operators $V_u\in
  \Lin(\R_t,\B_u\otimes \tilde{\R}_t)$ for $u\in [n]$ such that
  \begin{equation}
    \label{eq:ms-tilde-approx}
    \hat{Z}_u = \tilde{Z}_u, \quad \Re \tr_{\rho} \bigl( \hat{X}_u
    \tilde{X}_u \bigr) \approx_{\epsilon_1} 1,
  \end{equation}
  where $\tilde{D}_u = V_u^*(D\otimes \I)V_u$ for $D\in \{X,Z\}$ and
  $u\in [n]$.
  Note that we have omitted the superscript for the player index for
  simplicity.

  Define a new strategy $\Strategy'_u$ which is the same as
  $\Strategy_u$ except that reflections $\hat{D}_u$ are replaced with
  $\tilde{D}_u$ for player $(t)$.
  It then follows by Lemma~\ref{lem:approx-stab} and
  Eq.~\eqref{eq:ms-tilde-approx} that strategy $\Strategy'_u$ has
  value at least $1-O(\epsilon_1)$ in the eight-qubit stabilizer game.
  As $X^{\otimes 8}$ and $Z^{\otimes 8}$ are both in the set $\Xi$, it
  follows that $\tilde{D}_u$ and $\bigotimes_{i\ne t}\hat{D}_u^{(i)}$
  are $O(\epsilon_1)$-consistent.

  Similarly, for all $D_u\in Q\in \Power_{n,k}$ where $Q$ contains $k$
  Pauli operators of weight one, consider the state $\rho$,
  reflections $\hat{D}_u^{(i)}$ for $i\ne t$ and reflection
  $\hat{D_u|Q}$ for player $(t)$.
  By the test in \game{Confusion Check}, they form a strategy for the
  stabilizer game with value at least $1-O(\epsilon_2)$ where
  $\epsilon_2 = n^k \epsilon$.
  It follows that $\hat{D_u|Q}$ and $\bigotimes_{i\ne
    t}\hat{D}_u^{(i)}$ are $O(\epsilon_2)$-consistent.

  For any $Q$ that contains $D_u, D_v$, applying
  Lemma~\ref{lem:commute} with $\tilde{D}_u$, $\tilde{D}_v$ as $R_1$,
  $R_2$, $\hat{D_u|Q}$, $\hat{D_v|Q}$ as $S_1$, $S_2$, and
  $\bigotimes_{i\ne t} \hat{D}_u^{(i)}$, $\bigotimes_{i\ne t}
  \hat{D}_v^{(i)}$ as $U_1$ and $U_2$, we have for all $u\ne v \in
  [n]$,
  \begin{equation}
    \label{eq:ms-commute}
    \Re\tr_\rho \bigl( \tilde{D}_u \tilde{D}_v \tilde{D}_u \tilde{D}_v
    \bigr) \approx_{\epsilon_2} 1.
  \end{equation}

  Let $\reg{X}$, $\reg{Y}$ be quantum registers with $n$ qubits each.
  For $u\in [n]$, define $\ket{\Phi}_u$ to be the EPR state between
  the qubits $\bigl( \reg{X}, u \bigr)$ and $\bigl( \reg{Y}, u
  \bigr)$, and define isometry $W_u \in \Lin(\R_t,
  \X_u\otimes\Y_u\otimes \R_t)$ as
  \begin{equation*}
    W_u = (\I \otimes V_u^*) \SWAP_u (\ket{\Phi}_u \otimes V_u),
  \end{equation*}
  where $\SWAP_u$ is the $\SWAP$ gate acting on qubits $\bigl(
  \reg{X}, u \bigr)$ and the first output qubit $\reg{B}_u$ of $V_u$.

  Define isometry $V\in \Lin \bigl(\R_t, \X \otimes \Y \otimes \R_t
  \bigr)$ as the sequential application of $W_1, W_2, \ldots, W_n$,
  \begin{equation}
    \label{eq:ms-V}
    V = W_n W_{n-1} \cdots W_1.
  \end{equation}
  We claim that this choice of $V$ works for the stated claims by
  taking $\R'_t$ to be $\Y \otimes \R_t$.
  Define operators
  \begin{equation}
    \label{eq:ms-check}
    \check{D}_u = V^* (D_u\otimes \I) V\; \text{ for } D_u \in
    \{X_u, Z_u\}.
  \end{equation}

  As $D_u$ and $W_v$ commute for all $v>u$, we have
  \begin{equation*}
    \check{D}_u = W_1^* W_2^* \cdots W_{u-1}^* \tilde{D}_u W_{u-1}
    W_{u-2} \cdots W_1.
  \end{equation*}

  For each $u\in [n]$, define a quantum channel
  \begin{equation}
    \label{eq:ms-T}
    \mathfrak{T}_u(\rho) = \frac{\rho + \tilde{X}_u \rho \tilde{X}_u +
      \tilde{Z}_u \rho \tilde{Z}_u + \tilde{X}_u
      \tilde{Z}_u \rho \tilde{Z}_u\tilde{X}_u}{4}.
  \end{equation}
  By a series of applications of Lemma~\ref{lem:W},
  \begin{equation}
    \label{eq:ms-depolarize}
    \check{D}_u = \mathfrak{T}_1 \circ
    \mathfrak{T}_2 \circ \cdots \circ \mathfrak{T}_{u-1} (\tilde{D}_u).
  \end{equation}
  We remark that the channel $\mathfrak{T}_u$ is exactly the
  depolarizing channel with respect to the qubit defined by the
  anti-commuting pair of reflections $\tilde{X}_u$ and $\tilde{Z}_u$.
  The expression in Eq.~\eqref{eq:ms-depolarize} states that the
  operators $\check{D}_u$ defined in Eq.~\eqref{eq:ms-check} using EPR
  states and $\SWAP$ can be constructed by sequentially depolarizing the
  qubits defined by reflections $\tilde{X}_v$ and $\tilde{Z}_v$ for
  $v=u-1, u-2, \ldots, 1$.

  We then claim that for all $v \le u\in [n]$
  \begin{equation}
    \label{eq:ms-D-vu}
    \Re \tr_\rho \bigl( \tilde{D}_u \, \mathfrak{T}_v \circ
    \mathfrak{T}_{v+1} \circ \cdots \circ \mathfrak{T}_{u-1}
    (\tilde{D}_u) \bigr) \approx_{(u-v)^2 \epsilon_2} 1.
  \end{equation}
  We prove the claim by induction on $v$.
  For $v = u$, the claim follows from the fact that $\tilde{D}_u$ is a
  reflection.
  Assume that the statement holds for some $v>1$ and we prove the case
  for $v-1$.
  Define for simplicity
  \begin{equation*}
    R_v = \mathfrak{T}_v \circ \mathfrak{T}_{v-1} \circ \cdots \circ
    \mathfrak{T}_{u-1} (\tilde{D}_u),
  \end{equation*}
  and the induction hypothesis becomes
  \begin{equation*}
    \Re \tr_\rho \bigl( \tilde{D}_u \, R_v \bigr) \approx_{(u-v)^2
      \epsilon_2} 1.
  \end{equation*}

  For $v-1$, we have by the induction hypothesis, the existence of
  consistency reflections for $\tilde{X}_{v-1}$,
  Eq.~\eqref{eq:ms-commute} and Lemma~\ref{lem:approx-stab},
  \begin{equation*}
    \Re \tr_\rho \bigl(\tilde{D}_u \tilde{X}_{v-1} R_v \tilde{X}_{v-1}
    \bigr) \approx_{(u-v+1)^2\epsilon_2} 1.
  \end{equation*}
  By a similar argument that adds $\tilde{Z}_{v-1}$ and
  $\tilde{X}_{v-1}\tilde{Z}_{v=1}$ in the expression, we have
  \begin{equation*}
    \Re \tr_\rho \bigl( \tilde{D}_u \, \mathfrak{T}_{v-1} (R_v) \bigr)
    \approx_{(u-v+1)^2 \epsilon_2} 1,
  \end{equation*}
  which proves the claim in Eq.~\eqref{eq:ms-D-vu} for $v-1$.

  Taking $v=1$ in Eq.~\eqref{eq:ms-D-vu}, we have for all $u\in [n]$,
  \begin{equation*}
    \Re \tr_\rho \bigl( \tilde{D}_u \check{D}_u \bigr) \approx_{n^2
      \epsilon_2} 1.
  \end{equation*}
  Lemma~\ref{lem:approx-stab} and Eq.~\eqref{eq:ms-tilde-approx} then
  imply that, for $\epsilon_3 = n^2 \epsilon_2$,
  \begin{equation}
    \label{eq:ms-1-1}
    \Re \tr_\rho \bigl( \hat{D}_u \check{D}_u \bigr) \approx_{
      \epsilon_3} 1,
  \end{equation}
  which proves the approximation in Eq.~\eqref{eq:ms-1} for Pauli
  operator of weight one by choosing $\kappa$ large enough.

  We then prove Eq.~\eqref{eq:ms-2} for the case where $Q \subset
  \Pauli_{n,1}$, containing commuting Pauli operators of weight one
  only.
  First recall that
  \begin{equation*}
    \Re\tr_\rho \Bigl( \hat{D}_u \otimes \Bigl( \bigotimes_{i\ne
      t}\hat{D}_u^{(i)} \Bigr) \Bigr) \approx_{\epsilon_1} 1.
  \end{equation*}
  By Lemma~\ref{lem:approx-stab} and the consistency analysis in the
  beginning of the proof, we have
  \begin{equation*}
    \Re\tr_\rho \bigl( \hat{D}_u \hat{D_u|Q} \bigr)
    \approx_{\epsilon_2} 1.
  \end{equation*}
  By Eq.~\eqref{eq:ms-1-1},
  \begin{equation*}
    \Re\tr_\rho \bigl( \check{D}_u \hat{D_u|Q} \bigr)
    \approx_{\epsilon_3} 1.
  \end{equation*}
  With Lemma~\ref{lem:derived}, this proves the approximation in
  Eq.~\eqref{eq:ms-2} for $Q \subseteq \Pauli_{n,1}$.

  For the proof of Eq.~\eqref{eq:ms-1}, we analyze the \game{Parity
    Check} part of the game.
  Consider a strategy $\Strategy'$ that is the same as $\Strategy$ but
  with reflections $\hat{D}_u$ replaced by $\check{D}_u$ and
  measurements $\hat{Q}$ replaced by $\check{Q}$ for $Q\subseteq
  \Pauli_{n,1}$ for all players.
  It then follows by the claims proved above and the monotonicity of
  trace distance that the value of strategy $\Strategy'$ is at least
  $1-O(\epsilon_4)$ for $\epsilon_4 = n^k\epsilon_3^{1/2}$.
  For $i\in [8]$, and $P = \prod_{u\in J} D_u$, define operators
  \begin{equation*}
    \tilde{P}^{(i)} \defeq \prod_{u\in J}\, \check{D}_u^{(i)},\quad
    \check{P}^{(i)} \defeq V_i^* (P\otimes \I) V_i,
  \end{equation*}
  and write $\tilde{P}$ for $\tilde{P}^{(t)}$, $\check{P}$ for
  $\check{P}^{(t)}$.
  When strategy $\Strategy'$ is used, the bit $a^{(i)}$ defined in the
  game corresponds to the measurement defined by $\check{P}^{(i)}$ for
  $i\ne t$.

  As strategy $\Strategy'$ has value at least $1-O(\epsilon_4)$, the
  test in \game{Parity Check} implies that for $P\in \Pauli_{n,k}$
  \begin{equation}
    \label{eq:ms-parity}
    \Re\tr_\rho \Bigl( \hat{P} \otimes \Bigl( \bigotimes_{i\ne t} \check{P}^{(i)}
    \Bigr) \Bigr) \approx_{\epsilon_5} 1,
  \end{equation}
  for $\epsilon_5 = \abs{\Pauli_{n,k}}\epsilon_4$.
  Similarly, the test in \game{Stabilizer Check} implies that
  \begin{equation*}
    \Re\tr_\rho \Bigl( \check{D}_u \otimes \Bigl( \bigotimes_{i\ne t}
    \check{D}_u^{(i)} \Bigl) \Bigr) \approx_{\epsilon_5} 1.
  \end{equation*}
  By Lemma~\ref{lem:approx-stab}, this implies that
  \begin{equation*}
    \Re\tr_\rho \Bigl( \tilde{P} \otimes
    \Bigl( \bigotimes_{i\ne t} \tilde{P}^{(i)} \Bigl) \Bigr)
    \approx_{\epsilon_5} 1.
  \end{equation*}
  By Lemma~\ref{lem:overline},
  \begin{equation*}
    \Re\tr_\rho \bigl( \check{P}^{(i)} \cdot \tilde{P}^{(i)} \bigr)
    \approx_{\epsilon_3} 1,
  \end{equation*}
  for all $i\in [r]$ and therefore, by Lemma~\ref{lem:approx-stab-3},
  \begin{equation}
    \label{eq:ms-all-check}
    \Re\tr_\rho \Bigl( \check{P} \otimes
    \Bigl( \bigotimes_{i\ne t} \check{P}^{(i)} \Bigl) \Bigr)
    \approx_{r \epsilon_5} 1.
  \end{equation}
  Together with Eq.~\eqref{eq:ms-parity} and this equation,
  Lemma~\ref{lem:approx-stab-3} implies
  \begin{equation*}
    \Re\tr_\rho \bigl( \hat{P} \cdot \check{P} \bigr)
    \approx_{r \epsilon_5} 1,
  \end{equation*}
  which completes the proof for Eq.~\eqref{eq:ms-1}.

  We now prove the statement in~\eqref{eq:ms-2}.
  Consider a strategy $\check{\Strategy}$ that is the same as
  $\Strategy$ but replaces all reflection $\hat{P}$ with $\check{P}$.
  It then follows that the value of the strategy $\check{\Strategy}$
  is at least $1-O(n^kr^{1/2}\epsilon_5^{1/2})$.
  In strategy $\check{\Strategy}$, the players measures $\check{P}$
  for question $P$.
  The test in \game{Pauli Check} implies that for $\epsilon_6 =
  \abs{\Power_{n,k}} n^k\epsilon_5^{1/2}$
  \begin{equation*}
    \Re\tr_\rho \Bigl( \hat{P|Q} \otimes \Bigl( \bigotimes_{i\ne t}
    \check{P}^{(i)} \Bigr) \Bigr) \approx_{\epsilon_6} 1.
  \end{equation*}
  Applying Lemma~\ref{lem:approx-stab-3} again to this equation and
  Eq.~\eqref{eq:ms-all-check}, we have
  \begin{equation*}
    \Re\tr_\rho \bigl( \hat{P|Q} \cdot \check{P} \bigr)
    \approx_{\epsilon_6} 1,
  \end{equation*}
  which completes the proof for Eq.~\eqref{eq:ms-2} using
  Lemma~\ref{lem:derived}.

  The second part of the theorem follows from a similar argument used
  to prove the second part of Theorem~\ref{thm:stab-game}.
\end{proof}

We note that we haven't tried to optimize the dependence of the
approximation error on $n$ and $\epsilon$ as it is not important for
our work. With a more careful analysis, it should be possible to have
a much better dependence on these parameters.

\section{Propagation Games, Constraint Propagation Games and Rigidity}
\label{sec:prop}

\subsection{Propagation Games}

In this section, we define the propagation game, an extended nonlocal
game that checks the propagation of a sequences of reflections.
Let $R_1, R_2, \ldots, R_n$ be $n$ reflections.
Let $\mathfrak{R} = \bigl( R_{\zeta_i} \bigr)_{i=1}^N$ be a sequence
of reflections with indices $\zeta_i \in [n]$.
The propagation game is an extended nonlocal game that checks the
propagation of the reflection sequence $\mathfrak{R}$ on the player's
system.

Let $\bigl( v_i \bigr)_{i=0}^N$ be an increasing sequence of integers
of length $N+1$.
The propagation graph $G=(V,E)$ of the reflections $\mathfrak{R} =
\bigl( R_{\zeta_i} \bigr)_{i=1}^N$ over the vertex sequence $\bigl(
v_i \bigr)_{i=0}^N$ is the labeled graph with vertex set $V = \{ v_i
\mid i=0, 1, \ldots, N \}$, edge set $E = \{ (v_{i-1}, v_i) \}$ and
edge labels $\Gamma_e = R_{\zeta_i}$ for $e=(v_{i-1},v_i)$.

For each propagation graph $G$, we can define a propagation game as in
Fig.~\ref{fig:prop-game}.
In the game, the referee possesses a clock register $\reg{S}$ with
associated Hilbert space $\complex^V$.
The question is sampled from the set $[n]$, the index set for the
reflections.
The player is expected to perform the two-outcome measurement
corresponding to the reflection $R_j$ for question $j$ and reply with
the measurement outcome.
For each edge $e=(u,v)\in E$, define a projective measurement $\Pi_e =
\bigl\{ \Pi_e^0, \Pi_e^1, \Pi_e^2 \bigr\}$ on $\reg{S}$ where
\begin{equation}
  \label{eq:Pi-e}
  \begin{split}
    \Pi_e^0 & = \frac{(\ket{u}+\ket{v})(\bra{u}+\bra{v})}{2},\\
    \Pi_e^1 & = \frac{(\ket{u}-\ket{v})(\bra{u}-\bra{v})}{2},\\
    \Pi_e^2 & = \I - \Pi_e^0 - \Pi_e^1.
  \end{split}
\end{equation}

\begin{figure}[!htb]
  \begin{shaded}
    \ul{Propagation Game}\\[1em]
    Let $G=(V,E)$ be the propagation graph defined above.
    The referee selects an edge $e \in E$ uniformly at random and
    sends the index $j \in [n]$ for edge label $\Gamma_e = R_j$ to the
    player and receives an answer bit $a$; performs the projective
    measurement $\Pi_e$ on register $\reg{S}$ and accepts if the
    outcome is $2$ or is equal to $a$; rejects otherwise.
  \end{shaded}
  \caption{The propagation game between a referee and a player.}
  \label{fig:prop-game}
\end{figure}

We prove the following lemma for propagation games.
Roughly, it says that the shared state has to be close to a history
state of the computation defined by the sequence of the reflections.

\begin{lemma}
  \label{lem:prop-game}
  The propagation game in Fig.~\ref{fig:prop-game} has nonlocal value
  $1$.
  Furthermore, the following property holds.
  Let $\Strategy = (\rho, \{\hat{R}_j\})$ be a strategy with $\rho \in
  \Density(\S \otimes \R)$ and $\hat{R}_j \in \Herm(\R)$.
  Define $\hat{U} \in \Unitary(\S\otimes \R)$,
  \begin{equation}
    \label{eq:prop-game-u}
    \hat{U} = \sum_{i=0}^N
    \ket{v_i}\bra{v_i} \otimes \hat{R}_{\zeta_{i}} \hat{R}_{\zeta_{i-1}}
    \cdots \hat{R}_{\zeta_1},
  \end{equation}
  and state $\ket{V} \in \S$
  \begin{equation}
    \label{eq:prop-game-V}
    \ket{V} = \frac{1}{\sqrt{\abs{V}}} \sum_{v\in V} \ket{v}.
  \end{equation}
  If strategy $\Strategy$ has value at least $1-\epsilon$, then there
  exists a state $\rho'_{\reg{R}}\in \Density(\R)$ such that
  \begin{equation*}
    \dtr \Bigl(\rho, \hat{U} \bigl( \ket{V}\bra{V} \otimes \rho'_{\reg{R}}
    \bigr) \hat{U}^* \Bigr) \le O(N^{3/2}\epsilon^{1/2}).
  \end{equation*}
\end{lemma}

\begin{proof}
  The player chooses an arbitrary state $\ket{\psi} \in \R$ and
  initialize the registers $\reg{S}$ and $\reg{R}$ in the state $U
  (\ket{V} \otimes \ket{\psi})$ for
  \begin{equation*}
    U = \sum_{i=0}^N \ket{v_i}\bra{v_i} \otimes R_{\zeta_{i}}
    R_{\zeta_{i-1}} \cdots R_{\zeta_1},
  \end{equation*}
  He replies each question $j\in [n]$ with the outcome of the
  measurement defined by reflection $R_j$.
  A direct calculation then concludes the first part of the lemma.

  We now proves the second part of the lemma.
  For edge $e$, let $\zeta_e \in[n]$ be the index of reflection for
  the label $\Gamma_e = R_{\zeta_e}$.
  For a strategy $\Strategy = (\rho, \{\hat{R}_j\})$, the referee
  rejects with probability
  \begin{equation}
    \label{eq:prop-game-pr}
    \begin{split}
      & \E_{e \in E(G)} \tr_\rho \biggl( \Pi_e^0 \otimes
      \frac{\I-\hat{R}_{\zeta_e}}{2} + \Pi_e^1 \otimes
      \frac{\I+\hat{R}_{\zeta_e}}{2} \biggr)\\
      = & \frac{1}{2N} \tr_\rho \sum_{i=1}^N
      \bigl[(\ket{v_{i-1}}\bra{v_{i-1}} + \ket{v_i}\bra{v_i}) \otimes
      \I - (\ket{v_{i-1}}\bra{v_i} + \ket{v_i}\bra{v_{i-1}}) \otimes
      \hat{R}_{\zeta_i} \bigr]\\
      = & \frac{1}{2N} \tr_{\rho'} \sum_{i=1}^N
      (\ket{v_{i-1}}\bra{v_{i-1}} + \ket{v_i}\bra{v_i} -
      \ket{v_{i-1}}\bra{v_i} - \ket{v_i}\bra{v_{i-1}})\\
      = & \frac{1}{2N} \tr_{\rho'} L(G),
    \end{split}
  \end{equation}
  where $\rho' = \hat{U}^*\rho \hat{U}$ and the matrix $L(G)$ is the
  Laplacian matrix of graph $G$.
  As the strategy has value at least $1-\epsilon$, we have
  \begin{equation}
    \label{eq:prop-game-1}
    \tr_{\rho'} L(G) \le 2N\epsilon.
  \end{equation}

  By standard techniques~\cite{Sin93,Chu96}, the matrix $L(G)$ has a
  zero eigenvalue with eigenvector $\ket{V}$ and all other eigenvalues
  are at least $1/(N+1)^2$.
  That is, for projection $\Pi_V = \ket{V}\bra{V}$, we have
  \begin{equation*}
    L(G) \ge \frac{\I - \Pi_V}{(N+1)^2}.
  \end{equation*}
  Together with Eq.~\eqref{eq:prop-game-1}, this proves that
  \begin{equation*}
    \tr_{\rho'} (\Pi_V) \ge 1 - 2N(N+1)^2\epsilon.
  \end{equation*}
  By Lemma~\ref{lem:gentle}, it follows that
  \begin{equation*}
    \dtr(\rho',\ket{V}\bra{V}\otimes \rho'_{\reg{R}}) \le O(N^{3/2}\epsilon^{1/2}),
  \end{equation*}
  where $\rho'_{\reg{R}} \in \Density(\R)$ is defined as
  \begin{equation*}
    \rho'_{\reg{R}} = \frac{\bra{V}\rho'\ket{V}}{\tr_{\rho'}\Pi_V}.
  \end{equation*}
  It then follows that
  \begin{equation}
    \label{eq:prop-game-dtr}
    \dtr \bigl(\rho, \hat{U} (\ket{V}\bra{V} \otimes \rho'_{\reg{R}}) \hat{U}^*
    \bigr) \le O(N^{3/2}\epsilon^{1/2}),
  \end{equation}
  which completes the proof.
\end{proof}

We introduce two types of extensions to the propagation game.
The first type allows controlled reflections of the form
$\Lambda_c(R_j)$ in the sequence of reflections for $j\in [n]$ and $c$
the control qubit index of one of the referee's registers.
In the second extension, the referee may confuse the player by asking
him to perform $k$ pairwise commuting reflections and answer $k$
output bits, even though the referee is interested in the measurement
outcome of one of the reflections.
For this type of extension, the corresponding reflection in the
reflection sequence will have the form $R_{j|q}$, where $j\in q$ and
$q$ a subset of $[n]$ of size $k$.
The reflection sequence now becomes $\mathfrak{U} = \bigl( U_i
\bigr)_{i=1}^N$ where each $U_i$ has one of the following three
possible forms:
\begin{enumerate}
\item $U_i = R_j$ for some reflection index $j\in [n]$,
\item $U_i = \Lambda_c(R_j)$ for $j\in [n]$ and $c$ the control qubit
  index,
\item $U_i = R_{j|q}$ for $j\in q$ and $q$ a subset of $[n]$ of size
  $k$.
\end{enumerate}
In the propagation graph for the sequence $\mathfrak{U} = \bigl( U_i
\bigr)_{i=1}^N$ over vertices $\bigl( v_i \bigr)_{i=0}^N$, the label
for edge $e=(v_{i-1},v_i)$ is $\Gamma_e = U_i$.

The propagation game is updated to accommodate the changes accordingly
in Fig.~\ref{fig:prop-game-ext}.
The referee possesses two registers, the clock register $\reg{S}$ as
before and a register $\reg{X}$ containing the control qubits.
The question to the player is either a $j\in [n]$ or a set $q\subset
[n]$ of size $k$.
The strategy of the player can be described by $\Strategy = \bigl(
\rho, \{\hat{R}_j\}, \{M_q\} \bigr)$ where $\hat{R}_j$ is the
reflection corresponding to the measurement the player performs for
question $j$.
Measurement $\{M_q\}$ is a projective measurement with $k$ outcome
bits for question $q$.
For each $j\in q$, define derived reflections
\begin{equation*}
  \hat{R}_{j|q} = \sum_{b:q\rightarrow \{0,1\}} (-1)^{b(j)} M_q^b.
\end{equation*}
These $k$ reflections then equivalently characterize the measurement
$\{M_q\}$.
The strategy can then be equivalently described by
\begin{equation*}
  \Strategy = \bigl(\rho, \{\hat{R}_j\}, \{\hat{R}_{j|q}\} \bigr).
\end{equation*}

We will refer to this generalization of the propagation game also as
the propagation game as there will be no confusion.
As Lemma~\ref{lem:prop-game}, the following lemma holds for this
extended version of propagation games.

\begin{figure}[!htb]
  \begin{shaded}
    \ul{Propagation Game (Extended)}\\[1em]
    Let $G=(V,E)$ be the propagation graph.
    The referee selects an edge $e \in E$ uniformly at random.
    \begin{enumerate}
    \item If $\Gamma_e = R_j$ for some $j\in [n]$, he sends $j$ to the
      player and receives an answer bit $a$; performs the projective
      measurement $\Pi_e$ on register $\reg{S}$ and accepts if the
      outcome is $2$ or is equal to $a$; rejects otherwise.

    \item If $\Gamma_e = \Lambda_c(R_j)$, sends $j$ to the player and
      receives an answer bit $a$; performs the projective measurement
      $\Pi_e$ on register $\reg{S}$ and accepts if the outcome $t=2$
      and continue otherwise; measures $Z_{\reg{X},c}$ with outcome
      bit $b$ and rejects if $b=0,t=1$ or $b=1, a\oplus t=1$; accepts
      otherwise.

    \item Otherwise, for $\Gamma_e = R_{j|q}$, sends $q$ to the player
      and receives answer $b:q\rightarrow \{0,1\}$; performs the
      projective measurement $\Pi_e$ on register $\reg{V}$ and accepts
      if the outcome is $2$ or equals to $b(j)$; rejects otherwise.
    \end{enumerate}
  \end{shaded}
  \caption{The propagation game between a referee and a player.}
  \label{fig:prop-game-ext}
\end{figure}

\begin{lemma}
  \label{lem:prop-game-ext}
  The propagation game with controlled reflections in
  Fig.~\ref{fig:prop-game-ext} has nonlocal value $1$.
  Furthermore, the following property holds.
  Let $\Strategy = (\rho, \{\hat{R}_j\}, \{\hat{R}_{j|q}\})$ be a
  strategy with $\rho \in \Density(\S \otimes \X \otimes \R)$,
  $\hat{R}_j \in \Herm(\R)$ and $\hat{R}_{j|q}$ the derived
  reflections from the strategy.
  Define $\hat{U}\in \Unitary(\S\otimes \X \otimes \R)$
  \begin{equation}
    \hat{U} = \sum_{i=0}^N
    \ket{v_i}\bra{v_i} \otimes \hat{U}_{i} \hat{U}_{i-1}
    \cdots \hat{U}_{1},
  \end{equation}
  where $\hat{U}_i = \I_{\reg{X}} \otimes \hat{R}_j$ for $U_i = R_j$,
  $\hat{U}_i = \Lambda_{\reg{X},c} (\hat{R}_j)$ for $U_i =
  \Lambda_{c}(R_j)$ and $\hat{U}_i = \I_{\reg{X}}\otimes
  \hat{R}_{j|q}$ for $U_i = R_{j|q}$.
  Define state $\ket{V} \in \S$
  \begin{equation*}
    \ket{V} = \frac{1}{\sqrt{\abs{V}}} \sum_{v\in V} \ket{v}.
  \end{equation*}
  If strategy $\Strategy$ has value at least $1-\epsilon$, then there
  exists a state $\rho'_{\reg{XR}}\in \Density(\X\otimes \R)$ such
  that
  \begin{equation*}
    \dtr \Bigl(\rho, \hat{U} \bigl( \ket{V}\bra{V} \otimes \rho'_{\reg{XR}}
    \bigr) \hat{U}^* \Bigr) \le O(N^{3/2}\epsilon^{1/2}).
  \end{equation*}
\end{lemma}

\begin{proof}
  If $\Gamma_e = \Lambda_c(R_j)$, the referee rejects with probability
  \begin{equation*}
    \tr_\rho \Bigl[ \ket{0}\bra{0}_{\reg{X},c} \otimes \Pi_e^1 \otimes
    \I_{\reg{R}} + \ket{1}\bra{1}_{\reg{X},c} \otimes \Bigl( \Pi_e^1
    \otimes \frac{\I + \hat{R}_j}{2} + \Pi_e^0 \otimes \frac{\I -
      \hat{R}_j}{2} \Bigr)\Bigr].
  \end{equation*}
  If $e=(v_{i-1}, v_i)$, a direct calculation simplifies the above
  expression to
  \begin{equation*}
    \frac{1}{2}\tr_\rho \bigl[(\ket{v_i}\bra{v_i} +
    \ket{v_{i-1}}\bra{v_{i-1}}) \otimes \I_{\reg{XR}} -
    (\ket{v_i}\bra{v_{i-1}} + \ket{v_{i-1}}\bra{v_i})\otimes
    \Lambda_{\reg{X},c}(\hat{R}_j) \bigr].
  \end{equation*}

  Similarly, for $\Gamma_e = R_{j|q}$, the referee rejects with
  probability
  \begin{equation*}
    \frac{1}{2}\tr_\rho \bigl[ (\ket{v_i}\bra{v_i} +
    \ket{v_{i-1}}\bra{v_{i-1}}) \otimes \I_{\reg{R}} -
    (\ket{v_i}\bra{v_{i-1}} + \ket{v_{i-1}}\bra{v_i})\otimes
    \hat{R}_{j|q} \bigr].
  \end{equation*}
  Note that even though the player does not know $j$ when the referee
  sampled $\Gamma_e = R_{j|q}$, he can nevertheless prepare the
  initial state that depends on $R_{j|q}$ as the propagation graph of
  the game is known to the player.

  The proof now follows along similar lines as in the proof for
  Lemma~\ref{lem:prop-game}.
\end{proof}

\subsection{Constraint Propagation Game}

We define an extended nonlocal game, called the constraint propagation
game, for any system of product-form reflection constraints.
A system of product-form reflection constraints is described by a set
of $n$ reflections $R_1, R_2, \ldots, R_n$ and a set of $m$ operator
identities $C_1, C_2, \ldots, C_m$ where each $C_i$ is of the form
\begin{equation*}
  R_{j_{i,1}} R_{j_{i,2}} \cdots R_{j_{i,n_i}} = (-1)^{\tau_i} \I,
\end{equation*}
for $j_{i,1}, j_{i,2}, \ldots, j_{i,n_i} \in [n]$ and $\tau_i\in
\{0,1\}$.
We will sometimes abuse the notion and also use $C_i$ to denote the
product
\begin{equation*}
  R_{j_{i,1}} R_{j_{i,2}} \cdots R_{j_{i,n_i}},
\end{equation*}
so the constraints become $C_i = (-1)^{\tau_i}\I$.
Let $N_i=\sum_{j=1}^i n_j$ be the total number of occurrences of the
reflections in $C_1, C_2, \ldots, C_i$.
Define $N=N_m$ and $N_0=0$ for notational convenience.
The number $N$ is referred to as the size of the constraint system.

We say that the constraint system has a \textit{quantum satisfying
  assignment} if there exists an Hilbert space and $n$ reflections
acting on the space such that all constraints are satisfied.
These definitions are closely related to but different from the
concept of binary constraint system games introduced in~\cite{CM12}.
The formulation here allows more flexible ways of specifying relations
between reflections and do not enforce commutativity between
reflections occurring in the same constraint.
For example, we allow the constraints of the form
\begin{equation*}
  R_1R_2R_1R_2 = \I,\quad R_1R_2R_1R_2 = -\I,
\end{equation*}
representing the commutativity and anti-commutativity of two
reflections $R_1$ and $R_2$ respectively, while they are not
explicitly available in binary constraint system games.

For later convenience, we allow the derived reflections in the
constraints.
In this case, each constraint is now of the form
\begin{equation*}
  U_{i,1} U_{i,2} \cdots U_{i,n_i} = (-1)^{\tau_i} \I,
\end{equation*}
where $U_{i,j}$ is either one of the $n$ reflections $R_j$ or of the
form $R_{j|q}$ for $j\in q$ and $q\subset [n]$ of constant size.
In an satisfying assignment of the constraint system, we require
additionally that the reflections assigned to $R_{j|q}$ and $R_{j'|q}$
for $j,j'\in q$ commute.
But we do not require that $R_{j|q}$ is related to $R_j$ in any way at
this moment.

Let $\mathfrak{U} = \bigl( U_i \bigr)_{i=1}^N$ be the sequence of all
the reflections in $C_1, C_2, \ldots, C_m$ in the order that they
occur.
More specifically, for $1\le i \le m$ and $N_{i-1} < v \le N_i$, $U_v
= U_{i,v-N_{i-1}}$.
Let graph $G_{\prop} = (V_{\prop},E_{\prop})$ be the propagation graph
of the sequence $\mathfrak{U}$ over the vertices $\bigl( i
\bigr)_{i=0}^N$, and graph $G_{\cons} = (V_{\cons}, E_{\cons})$ be the
propagation graph of the sequence $\bigl( (-1)^{\tau_i}\I
\bigr)_{i=1}^m$ over the vertices $\bigl( N_i \bigr)_{i=0}^m$.
Define the constraint graph of the constraint system as $G = (V,E)$
with vertex set $V = V_{\prop}$ and $E = E_{\prop} \cup E_{\cons}$.
More explicitly, the vertex set is $V=\{0, 1, \ldots, N\}$, edge sets
$E_{\prop}$ and $E_{\prop}$ are
\begin{equation*}
  \begin{split}
    E_{\prop} & = \bigl\{ (i-1,i) \mid 1\le i \le N \bigr\},\\
    E_{\cons} & = \bigl\{ (N_{i-1},N_i) \mid 1\le i\le m \bigr\}.
  \end{split}
\end{equation*}
Edge $e=(i-1,i)\in E_{\prop}$ is labeled by $\Gamma_e = U_i$.
Edge $e=(N_{i-1},N_i)\in E_{\cons}$ is labeled by the operator
$\Gamma_e = (-1)^{\tau_i}\I$ on the right hand side of the constraint
$C_i$.

The constraint propagation game defined by the system of reflection
constraints is an extended nonlocal game between a referee and a
player.
The referee holds a quantum register $\reg{S}$ associated with Hilbert
space $\S = \complex^{V}$ for $V$ being the vertex set of graph $G$.
The referee will either sample and send $j\in [n]$ and expect the
player to perform the measurement represented by the corresponding
reflection $R_j$ and respond with the measurement outcome, or send
$q\subset [n]$ and expect the player to perform all the measurement
represented by the corresponding reflection with index in $q$ and send
back all the measurement outcomes.
The constraint propagation game is given in
Fig.~\ref{fig:cons-prop-game}.
Intuitively, in the \game{Propagation Check} part, the referee checks
the propagation of the reflections $U_i$ for $i=1,2,\ldots, N$ on the
state of the player.
In the \game{Constraint Check} part, the referee tests whether the
constraints are satisfied assuming correct propagations enforced by
the \game{Propagation Check}.
Note that in this part, the referee does not ask any question to the
player as the reflections are proportional to identities.

\begin{figure}[!htb]
  \begin{shaded}
    \ul{Constraint Propagation Game}\\[1em]
    Let $G, G_{\prop}, G_{\cons}$ be the constraint graph and two
    corresponding propagation graphs defined above.
    The referee does the following with equal probability:
    \begin{enumerate}
    \item \game{Propagation Check}.
      Plays the \game{Propagation Game} specified by the propagation
      graph $G_{\prop}$.

    \item \game{Constraint Check}.
      Selects an edge $e \in E_{\cons}$ uniformly at random; measures
      $\Pi_e$ on register $\reg{V}$ and accepts if the outcome is $2$
      or is equal to the sign $\tau$ in $\Gamma_e = (-1)^{\tau} \I$;
      rejects otherwise.
    \end{enumerate}
  \end{shaded}
  \caption{The constraint propagation game defined by a system of
    reflection constraints.}
  \label{fig:cons-prop-game}
\end{figure}

We prove the following lemma about constraint propagation games.

\begin{lemma}
  \label{lem:cons-prop-game}
  For any system of reflection constraints that has a quantum
  assignment, the nonlocal value of the corresponding propagation game
  is $1$.
  Moreover, for any strategy $\Strategy = \bigl( \rho, \{\hat{R}_j\},
  \{\hat{R}_{j|q}\} \bigr)$ that has value at least $1-\epsilon$, the
  constraints $C_i$'s are approximately satisfied in the following
  sense.
  Let $N$ be the size of the constraint system.
  Let $p_0$ and $\rho_0 \in \Density(\R)$ be the probability and the
  post-measurement state if a computational basis measurement on $\S$
  is performed on $\rho$ and the outcome is $0$.
  Define
  \begin{equation*}
    \hat{C}_i = \hat{U}_{i,1} \hat{U}_{i,2} \cdots \hat{U}_{i,n_i}.
  \end{equation*}
  Then, there exists a constant $\kappa$ such that,
  \begin{equation}
    \label{eq:cp-game-c}
    \Re \tr_{\rho_0} \hat{C}_i \approx_{N^\kappa\epsilon^{1/\kappa}} (-1)^{\tau_i},
  \end{equation}
  and
  \begin{equation}
    \label{eq:cp-game-pr}
    p_0 \approx_{N^\kappa\epsilon^{1/\kappa}} \frac{1}{N+1}.
  \end{equation}
\end{lemma}

\begin{proof}
  If the constraint system has a quantum assignment, choose $\R$ to be
  the Hilbert space of the assignment and let $R_j, R_{j|n} \in
  \Herm(\R)$ be the reflections that satisfy all the constraints.
  The player chooses an arbitrary state $\ket{\psi} \in \R$ and
  initialize the registers $\reg{S}$ and $\reg{R}$ in the state $U
  (\ket{V} \otimes \ket{\psi})$ for $U$ and $\ket{V}$ defined as
  \begin{equation*}
    U = \sum_{i=0}^N \ket{v_i}\bra{v_i} \otimes U_i U_{i-1} \cdots U_1
  \end{equation*}
  and $\ket{V} = \sum_{v\in V}\ket{v}/\sqrt{\abs{V}}$.
  He replies each question $j\in [n]$ with the outcome of the
  measurement defined by reflection $R_j$ and replies question
  $q\subseteq [n]$ with the measurement outcomes of $R_{j|q}$ for all
  $j\in q$.
  A direct calculation then concludes the first part of the lemma.

  We now proves the second part of the lemma.
  Let $\Strategy = (\rho, \{\hat{R}_j\}, \{\hat{R}_{j|q}\})$ be a
  strategy that has value at least $1-\epsilon$.
  The referee rejects with probability at most $2\epsilon$ in the
  \game{Propagation Check} and, by Lemma~\ref{lem:prop-game-ext},
  \begin{equation}
    \label{eq:cp-game-dtr}
    \dtr(\rho, \tilde{\rho}) \le O(\epsilon'),
  \end{equation}
  for state $\tilde{\rho} = \hat{U} (\ket{V}\bra{V}\otimes
  \rho'_{\reg{R}}) \hat{U}^*$, $\rho'_{\reg{R}} \in \Density(\R)$, and
  $\epsilon' = N^{3/2}\epsilon^{1/2}$.

  We now analyze the checking corresponding to the edges in
  $E_{\cons}$.
  For edge $e_i = (N_{i-1}, N_i)$, define operator
  \begin{equation*}
    H_i = \frac{1-(-1)^{\tau_i}}{2} \Pi_{e_i}^0 +
    \frac{1+(-1)^{\tau_i}}{2} \Pi_{e_i}^1,
  \end{equation*}
  which simplifies to,
  \begin{equation*}
    H_i = \frac{1}{2} \bigl[ \ket{N_{i-1}}\bra{N_{i-1}} +
    \ket{N_{i}}\bra{N_{i}} - (-1)^{\tau_i} \ket{N_{i-1}}\bra{N_{i}} -
    (-1)^{\tau_i} \ket{N_{i}}\bra{N_{i-1}} \bigr],
  \end{equation*}
  by the definitions of $\Pi_{e_i}^0$ and $\Pi_{e_i}^1$.
  Edge $e_i$ is sampled with probability $1/2m$ in the game, and the
  probability that the referee rejects in this case is at most
  $2m\epsilon$ as the strategy has value at least $1-\epsilon$.
  This implies that $\tr_\rho H_i \le 2m\epsilon$, and by
  Eq.~\eqref{eq:cp-game-dtr},
  \begin{equation*}
    \tr_{\tilde{\rho}} H_i \le O(\epsilon').
  \end{equation*}

  By the definition of $\hat{U}$, $\hat{C}_i$ and $\tilde{\rho}$, this
  simplifies to
  \begin{equation*}
    \Re \tr_{\rho'_{\reg{R}}} \bigl( \hat{C}_1 \hat{C}_2
    \cdots \hat{C}_{i-1}\hat{C}_i \hat{C}_{i-1}^* \cdots
    \hat{C}_1^* \bigr) \approx_{N\epsilon'} (-1)^{\tau_i}.
  \end{equation*}

  For $i_0 \le i_1\in[m]$, introduce the notion
  \begin{equation*}
    \hat{C}_{i_0, i_1} = \hat{C}_{i_0}
    \hat{C}_{i_0+1} \cdots \hat{C}_{i_1},
  \end{equation*}
  and define $\hat{C}_{i_0,i_1} = \I$ if $i_0> i_1$.
  The last equation then becomes, for $i\in[m]$,
  \begin{equation}
    \label{eq:cp-game-approx}
    \Re \tr_{\rho'_{\reg{R}}} \bigl( \hat{C}_{1,i} \hat{C}_{1,i-1}^* \bigr)
    \approx_{N\epsilon'} (-1)^{\tau_i}.
  \end{equation}

  Using Lemma~\ref{lem:approx-stab-2}, we have for $1\le j<i$,
  \begin{equation*}
    \Re \tr_{\rho'_{\reg{R}}} \bigl( \hat{C}_{1,j} \hat{C}_i \hat{C}_{1,j}^*
    \bigr) \approx_{\sqrt{N\epsilon'}} \Re \tr_{\rho'_{\reg{R}}} \bigl[
    \bigl(\hat{C}_{1,j} \hat{C}_{1,j-1}^* \bigr)^* \hat{C}_{1,j}
    \hat{C}_i \hat{C}_{1,j}^* \bigl( \hat{C}_{1,j} \hat{C}_{1,j-1}^*
    \bigr) \bigr] = \Re \tr_{\rho'_{\reg{R}}} \bigl(\hat{C}_{1,j-1} \hat{C}_i
    \hat{C}_{1,j-1}^* \bigr).
  \end{equation*}
  A repeated application of the above approximation gives
  \begin{equation*}
    \Re \tr_{\rho'_{\reg{R}}} \hat{C}_i \approx_{m\sqrt{N\epsilon'}} \Re
    \tr_{\rho'_{\reg{R}}} \bigl( \hat{C}_{1,i-1} \hat{C}_i \hat{C}_{1,i-1}^*
    \bigr) \approx_{N\epsilon'} (-1)^{\tau_i}.
  \end{equation*}
  This proves the Eq.~\eqref{eq:cp-game-c} in the lemma but with state
  $\rho'_{\reg{R}}$ instead of $\rho_0$.
  To complete the proof it suffices to bound the distance between
  $\rho'_{\reg{R}}$ and $\rho_0$.

  Consider a computational basis measurement on system $\S$ of state
  $\tilde{\rho}$.
  By the definition of state $\tilde{\rho}$, it is obvious that the
  state left on $\R$ will be $\rho'_{\reg{R}}$ if the measurement
  outcome is $0$.
  Let $\tilde{p}_0 = 1/(N+1)$ be the probability that outcome $0$
  happens.
  By the monotonicity of the trace distance and
  Eq.~\eqref{eq:cp-game-dtr}, $\abs{p_0-\tilde{p}_0} \le \epsilon'$.
  This proves Eq.~\eqref{eq:cp-game-pr}.
  Also by the monotonicity of the trace distance and the triangle
  inequality, we have
  \begin{equation*}
    \begin{split}
      \dtr(\rho_0, \rho'_{\reg{R}}) & = \frac{N+1}{2}
      \norm{\tilde{p}_0 \rho_0
        - \tilde{p}_0 \rho'_{\reg{R}} }_1\\
      & \le \frac{N+1}{2} \left[ \norm{\tilde{p}_0 \rho_0 - p_0
          \rho_0}_1 + \norm{p_0 \rho_0 - \tilde{p}_0
          \rho'_{\reg{R}}}_1\right]\\
      & \le O(N\epsilon'),
    \end{split}
  \end{equation*}
  which completes the proof for Eq.~\eqref{eq:cp-game-c}.
\end{proof}

\subsection{Constraint Propagation Game for Multi-Qubit Pauli
  Operators}

Consider the following constraint system satisfied by the Pauli
operators of weight $k$ on $n$ qubits.
The reflections under consideration will be those in $\Pauli_{n,k}$.
Let $P|Q$ be the same as reflection $P$.
These reflections satisfy the constrains as follows:
\begin{enumerate}
\item $D_uD_vD_uD_v= \I$ for $u\ne v\in [n]$ and $D_u\in\{X_u,Z_u\}$
  and $D_v\in\{X_v,Z_v\}$;
\item $X_uZ_uX_uZ_u = -\I$ for $u\in [n]$;
\item $D_vX_uZ_uX_uZ_uD_v = -\I$ for $u,v\in [n]$, and $D_v \in \{X_v,
  Z_v\}$;
\item $P \cdot \prod_{v\in J} D_v = \I$ for $P\in \Pauli_{n,k}$, $J$
  the support of $P$ and $P = \prod_{v\in J}D_v$;
\item $(P|Q)P = \I$ for $Q\in \Power_{n,k}$ and $P\in Q$.
\end{enumerate}
We refer to this particular constraint system of reflections as the
$(n,k)$-constraint system.
It is easy to see that the size $N_{n,k}$ and the number of
constraints $m_{n,k}$ of the constraint system are at most
polynomially in $n$ for any constant $k$.

Consider the constraint propagation game of the $(n,k)$-constraint
system.
For operators of the form $X_u, Z_u, P, (P|Q)$, we add a hat to denote
the corresponding reflection in the player's strategy.
For example, $\hat{X}_u$, $\hat{Z}_u$ and $\hat{P}$ denote the
corresponding reflections of the player's strategy when receiving
measurement requests of $X_u$, $Z_u$, $P\in \Pauli_{n,k}$
respectively.
Similarly, $\hat{P|Q}$ denotes the derived reflection from the
player's projective measurement $\hat{Q}$ for the question $Q\in
\Power_{n,k}$ and $P\in Q$.

By Lemma~\ref{lem:cons-prop-game}, we can enforce approximate
satisfaction of these constraints on a quantum state.
For example, if $\bigl( \rho, \{ \hat{P} \}, \{ \hat{Q} \}\bigr)$ is a
strategy that has value at least $1-\epsilon$ in the constraint
propagation game defined by the $(n,k)$-constraint system, then for
constant $\kappa$ and $\epsilon' = N_{n,k}^\kappa
\epsilon^{1/\kappa}$,
\begin{equation*}
  \tr_{\rho_0} \bigl( \hat{D}_u \hat{D}_v \hat{D}_u \hat{D}_v \bigr)
  \approx_{\epsilon'} 1,
\end{equation*}
and
\begin{equation*}
  \tr_{\rho_0} \bigl( \hat{X}_u \hat{Z}_u \hat{X}_u \hat{Z}_u)
  \approx_{\epsilon'} -1.
\end{equation*}
These conditions will be helpful to prove rigidity type of results.
But unfortunately, we do not know if these conditions are already
sufficient enough to establish the existence of an isometry $V$ such
that, $\hat{P}$ is close to $P$ under the conjugation of $V$.
To complete the proof, we need to establish these approximation
properties not only on state $\rho_0$, but also on several other
states derived from it.
This is the reason that we will need to consider the following more
complicated game defined by the $(n,k)$-constraint system.

Let $\mathfrak{V}_{n,k} = \bigl( V_i \bigr)_{i=1}^{N_{n,k}}$ be the
sequence of the reflections (derived reflections) of the
$(n,k)$-constraint system.
Let $\mathfrak{W} = \bigl( W_i \bigr)_{i=1}^{N'}$ be the concatenation
of sequences
\begin{equation*}
  \mathfrak{V}_{n,k}, \Lambda_{2i-1}(X_i), \mathfrak{V}_{n,k},
  \Lambda_{2i}(Z_i)
\end{equation*}
for $i=1, 2, \ldots, n$.
The length of $\mathfrak{W}$ is $N'=2n(N_{n,k}+1)$.

A sequence of Pauli operators is called primitive if it consists of
$XZ$-form Pauli operators of weight one and has length at most $k$.
For any $Q\in \Power_{n,k}$, a derived sequence for $Q$ is a sequence
of the form $P_i|Q$ for $P_i\in Q$ of length at most $k$.
Let $\bigl( \mathfrak{Q}_l \bigr)_{l=1}^L$ be the sequence of all
sequences that is the concatenation of all possible primitive and
derived sequences, including the empty sequence as the first entry
$\mathfrak{Q}_1$.
The length of $\mathfrak{Q}_l$ is denoted as $q_l$.
For a sequence $\mathfrak{R}$ of reflections, let $\mathfrak{R}^*$ be
the sequence of entries in $\mathfrak{R}$ in the reversed order.
For $l\in [L]$, let $\mathfrak{U}_l$ be the concatenation of sequences
\begin{equation*}
  \mathfrak{Q}_l, \mathfrak{W}, \mathfrak{W}^*, \mathfrak{Q}_l^*.
\end{equation*}
The length of $\mathfrak{U}_l$ is $2(N' + q_l)$.
Finally, let $\mathfrak{U} = \bigl( U_i \bigr)_{i=1}^N$ be the
concatenation of sequences $\mathfrak{U}_l$ for $l=1, 2, \ldots, L$.

Let $N_j$ be the number of reflections in the first $j$ constraints of
the $(n,k)$-constraint system.
For $i=0, 1, \ldots, 2n-1$, define integer $N_0^i = i (N_{n,k}+1)$
that marks the vertex index for the start of the $(i+1)$-th occurrence
of $\mathfrak{V}_{n,k}$ in $\mathfrak{W}$.
For $i=0, 1, \ldots, 2n-1$, and $j=0, 1, \ldots, m_{n,k}$, define
integer $N_j^i = N_0^i + N_j$.
For $l\in [L]$, define
\begin{equation*}
  N_j^{i,l} = \sum_{l'=1}^{l-1} 2(N'+q_{l'}) + q_l + N_j^i.
\end{equation*}

Let graph $G_{\prop}=(V_{\prop}, E_{\prop})$ be the propagation graph
of the sequence $\mathfrak{U}$ over the vertex sequence $\{0, 1,
\ldots, N\}$.
For $i=1, 2, \ldots, 2n$, and $l=1, 2, \ldots, L$, let graph
$G_{\cons}^{i,l} = (V_{\cons}^{i,l}, E_{\cons}^{i,l})$ be the
propagation of $\bigl( (-1)^{c_j}\I \bigr)_{j=1}^{m_{n,k}}$ from the
right hand sides of the $(n,k)$-constraint system over vertex
sequences $\bigl( N_j^{i,l} \bigr)_{j=0}^{m_{n,k}}$.
Finally, define the constraint graph $G=(V,E)$ as $V=V_{\prop}$ and
\begin{equation*}
  E = E_{\prop} \cup \biggl( \bigcup_{i=1}^{2n} \bigcup_{l=1}^L
  E_{\cons}^{i,l} \biggr).
\end{equation*}

Define the $(n,k)$-constraint propagation game as in
Fig.~\ref{fig:mc-game}.
The $(n,k)$-constraint propagation game is an extended nonlocal game
between a referee and a player.
The referee possesses two registers $\reg{S}$ and $\reg{X}$.
Register $\reg{S}$ is a clock register with associated Hilbert space
$\complex^V$.
The control register $\reg{X}$ has $2n$ qubits.

\begin{figure}[!htb]
  \begin{shaded}
    \ul{$(n,k)$-Constraint Propagation Game}\\[1em]
    The referee does the following with equal probability:
    \begin{enumerate}
    \item \game{Propagation Check}.
      Plays the propagation game corresponding to the propagation
      graph $G_{\prop}$.

    \item \game{Initialization Check}.
      Measures the register $\reg{S}$, accepts if the outcome is not
      $0$ and continues otherwise; samples $i\in [2n]$, measures
      $X_{\reg{X},i}$ and accepts if the outcomes is $0$; rejects
      otherwise.

    \item \game{Constraint Check}.
      Randomly samples $i\in [2n]$, $l\in [L]$ and an edge $e \in
      E_{\cons}^{i,l}$; measures $\Pi_e$ on register $\reg{S}$ and
      accepts if the outcome is $2$ or is equal to $\tau$ for
      $\Gamma_e = (-1)^{\tau}\I$; rejects otherwise.
    \end{enumerate}
  \end{shaded}
  \caption{The constraint propagation game for the $(n,k)$-constraint
    system of Pauli operators on $n$ qubits of weight $k$.}
  \label{fig:mc-game}
\end{figure}

\begin{theorem}
  \label{thm:mc-game}
  The $(n,k)$-constraint propagation game has value $1$.
  Furthermore, there is a constant $\kappa$ such that for any strategy
  $\Strategy = \bigl( \rho, \{ \hat{P} \}, \{ \hat{Q} \} \bigr)$ that
  has value at least $1-\epsilon$, there exists an isometry $V\in
  \Lin(\R, \B^{\otimes n} \otimes \R')$ such that the following
  properties hold
  \begin{itemize}
  \item For all $P\in \Pauli_{n,k}$
    \begin{equation}
      \label{eq:mc-P}
      \dis_{\rho_0} \bigl( \hat{P}, \check{P} \bigr) \le
      O(N^\kappa \epsilon^{1/\kappa}),
    \end{equation}
    and, for all $Q\in \Power_{n,k}$,
    \begin{equation}
      \label{eq:mc-Q}
      \dis_{\rho_0} \bigl( \hat{Q}, \check{Q} \bigr) \le
      O(N^\kappa \epsilon^{1/\kappa}),
    \end{equation}
    where $\check{P} = V^* (P\otimes \I) V$, and $\check{Q}$ is the
    measurement of $k$ Pauli operators in $Q$ after the application of
    isometry $V$.

  \item The probability $p_0$ satisfies
    \begin{equation*}
      p_0 \approx_{N^\kappa \epsilon^{1/\kappa}} \frac{1}{N+1}.
    \end{equation*}
  \end{itemize}
  In the statement, $\rho_0$ is the reduced state on register
  $\reg{R}$ when the computational basis measurement on $\reg{S}$ has
  outcome $0$, and $p_0$ is the probability of outcome $0$ when
  measuring $\reg{S}$.
\end{theorem}

We prove the following lemmas before proving the above theorem.

\begin{lemma}
  \label{lem:overline2}
  Let $R_1, R_2, \ldots, R_k\in \Herm(\Y)$ be $k$ pairwise commuting
  reflections, $V\in \Lin(\X,\Y)$ an isometry, $\rho \in
  \Density(\X\otimes \Z)$ a quantum state.
  Define $\check{R}_i = V^* R_i V$.
  If for all $x\in \{0,1\}^k$, and $\tilde{R}_x =
  \prod_{i=1}^kR_i^{x_i}$,
  \begin{equation*}
    \Re\tr_\rho \bigl( \tilde{R}_x^* R_i \check{R}_i \tilde{R}_x
    \bigr) \approx_\epsilon 1,
  \end{equation*}
  then
  \begin{equation*}
    \Re\tr_\rho \biggl[ V^* \Bigl( \prod_{i=1}^k R_i \Bigr) V \,
    \prod_{i=1}^k \Bigl( \check{R}_i \Bigr) \biggr]
    \approx_{\sqrt{\epsilon}} 1.
  \end{equation*}
\end{lemma}

\begin{proof}
  By the Cauchy-Schwarz inequality, it is easy to show
  \begin{equation}
    \Re\tr_\rho \biggl[ V^* \Bigl( \prod_{i=1}^k R_i \Bigr) V \,
    \prod_{i=1}^k \Bigl( \check{R}_i \Bigr) \biggr]
    \approx_{\sqrt{\epsilon}}
    \Re\tr_\rho \biggl[ V^* \Bigl( R_k VV^* \prod_{i=1}^{k-1} R_i
    \Bigr) V \, \prod_{i=1}^k \Bigl( \check{R}_i \Bigr) \biggr].
  \end{equation}
  By this equation and Lemma~\ref{lem:approx-stab-2}, we have
  \begin{equation*}
    \Re\tr_\rho \biggl[ V^* \Bigl( \prod_{i=1}^k R_i \Bigr) V \,
    \prod_{i=1}^k \Bigl( \check{R}_i \Bigr) \biggr]
    \approx_{\sqrt{\epsilon}} \Re\tr_\rho \biggl[ R_k V^* \Bigl(
    \prod_{i=1}^{k-1} R_i \Bigr) V \, \prod_{i=1}^{k-1} \Bigl(
    \check{R}_i \Bigr) R_k \biggr].
  \end{equation*}
  A repeated application of the above procedure then proves the claim
  in the lemma.
\end{proof}

\begin{lemma}
  \label{lem:derived2}
  Let $M = \bigl\{ M^a \bigr\}$ be a projective measurement of $k$-bit
  outcome on quantum register $\reg{X}$ and $R_1, R_2, \ldots, R_k$ be
  its derived reflections.
  Let $N = \bigl\{ N^a \bigr\}$ be a projective measurement of $k$-bit
  outcome on quantum register $\reg{Y}$ and $S_1, S_2, \ldots, S_k$ be
  its derived reflections.
  Let $V\in \Lin(\X,\Y)$ be an isometry and $\rho\in
  \Density(\X\otimes \Z)$ a quantum state.
  For $i\in [k]$, define $\check{S}_i = V^* S_i V$.
  Let $\check{N}$ be the quantum measurement that measures $N$ after
  the application of isometry $V$.

  If for $i\in [k]$, and all $x\in \{0,1\}^k$, $\tr_\rho \bigl(
  \tilde{R}_x R_i \check{S_i} \tilde{R}_x \bigr) \approx_\epsilon 1$
  where $\tilde{R}_x = \prod_{}R_i^{x_i}$, then
  \begin{equation*}
    \dis_\rho (M, \check{N}) \le O(\epsilon^{1/4}).
  \end{equation*}
\end{lemma}

\begin{proof}
  By the first part of the proof for Lemma~\ref{lem:derived}, we have
  \begin{equation*}
    \sum_a \tr_\rho (M^a \check{N}^a)
    = \frac{1}{2^k} \sum_{x\in \{0,1\}^k} \tr_\rho \biggl[ \Bigl(
    \prod_{i=1}^k R_i^{x_i} \Bigr) V^* \Bigl( \prod_{i=1}^k
    S_i^{x_i} \Bigr) V \biggr].
  \end{equation*}

  For all $x\in \{0,1\}^k$, we have
  \begin{equation*}
    \begin{split}
      \tr_\rho \biggl[ \Bigl(\prod_{i=1}^k R_i^{x_i} \Bigr) V^* \Bigl(
      \prod_{i=1}^k S_i^{x_i} \Bigr) V \biggr] &
      \approx_{\sqrt{\epsilon}} \tr_\rho \biggl[ \Bigl(\prod_{i=1}^k
      R_i^{x_i} \Bigr) V^* \Bigl(\prod_{i=1}^{k-1} S_i^{x_i} \Bigr) V
      \check{S}_k^{x_k} \biggr]\\
      & \approx_{\sqrt{\epsilon}} \tr_\rho \biggl[
      R_k^{x_k}\Bigl(\prod_{i=1}^{k-1} R_i^{x_i} \Bigr) V^*
      \Bigl(\prod_{i=1}^{k-1} S_i^{x_i} \Bigr) V
      R_k^{x_k} \biggr].\\
    \end{split}
  \end{equation*}
  The proof completes by repeating the above approximation procedure.
\end{proof}

\begin{proof}[Proof of Theorem~\ref{thm:mc-game}]
  For $\mathfrak{U} = \bigl( U_i \bigr)_{i=1}^N$, define a sequence
  $\mathfrak{U}' = \bigl( U'_i \bigr)_{i=1}^N$ as follows.
  Operator $U'_i$ is defined to be $P$ if $U_i = P$ or $U_i = (P|Q)$
  for $P\in \Pauli_{n,k}$, $Q\in \Power_{n,k}$, and is defined to be
  $\Lambda_{\reg{X},c}(D_u)$ if $U_i = \Lambda_c(D_u)$ for $D_u \in
  \{X_u, Z_u\}$ and $u\in [n]$.

  Define unitary operator
  \begin{equation*}
    U' = \sum_{i=0}^N \ket{i}\bra{i} \otimes U'_i U'_{i-1} \cdots U'_1,
  \end{equation*}
  and
  \begin{equation*}
    \ket{\Phi}_{\reg{X}} = \frac{1}{2^n} \sum_{x\in \{0,1\}^{2n}}
    \ket{x}.
  \end{equation*}
  If the player chooses state $U'(\ket{V}_{\reg{S}}
  \ket{\Phi}_{\reg{X}} \ket{\psi}_{\reg{R}})$ for some $\ket{\psi} \in
  \R$ and measure honestly, the referee accepts with certainty.
  This proves the first part of the theorem.

  We now prove the second part of the theorem.
  Define a sequence $\hat{\mathfrak{U}} = \bigl( \hat{U}_i
  \bigr)_{i=1}^N$ as follows.
  Operators $\hat{U}_i$ is defined to be $\hat{P}$ if $U_i = P$,
  $\hat{P|Q}$ if $U_i = (P|Q)$ and $\Lambda_{\reg{X},c}(\hat{P})$ if
  $U_i = \Lambda_c(P)$.
  For $i_0<i_1 \in [N]$, introduce the notion
  \begin{equation*}
    \hat{U}_{i_1\leftarrow i_0} = \hat{U}_{i_1} \hat{U}_{i_1 - 1}
    \cdots \hat{U}_{i_0}.
  \end{equation*}
  Define unitary operator
  \begin{equation*}
    \hat{U} = \sum_{i=0}^N \ket{i}\bra{i} \otimes \hat{U}_{i\leftarrow 1},
  \end{equation*}

  As the strategy $\Strategy$ has value at least $1-3\epsilon$ in the
  \game{Propagation Check} part, Lemma~\ref{lem:prop-game-ext} implies
  that there is a state
  \begin{equation*}
    \rho'_0 \in \Density(\X \otimes \R),
  \end{equation*}
  such that
  \begin{equation}
    \label{eq:mc-dtr}
    \dtr(\rho, \rho') \le O(\epsilon_1)
  \end{equation}
  for $\rho' = \hat{U}(\ket{V}\bra{V}\otimes \rho'_0) \hat{U}^*$ and
  $\epsilon_1 = N^{3/2}\epsilon^{1/2}$.

  Define a strategy $\Strategy'$ that is the same as $\Strategy$
  except that the shared state is $\rho'$ instead of $\rho$.
  By monotonicity of the trace distance, strategy $\Strategy'$ has
  value at least $1-O(\epsilon_1^{1/2})$.
  Define Hamiltonian
  \begin{equation*}
    H_{\text{in}} = \frac{1}{2n} \sum_{i=1}^{2n} \ket{0}
    \bra{0}_{\reg{S}} \otimes \frac{(\I-X)_{\reg{X},i}}{2}.
  \end{equation*}
  The referee rejects $\Strategy'$ in the \game{Initialization Check}
  part with probability
  \begin{equation*}
    \tr_{\rho'} \bigl(H_{\text{in}} \bigr) = \frac{1}{N} \tr_{\rho'_0}
    \Bigl[\sum_{i=1}^{2n} \frac{(\I-X)_{\reg{X},i}}{2} \Bigr] \le
    O(\epsilon_1^{1/2}).
  \end{equation*}
  Define states $\hat{\rho}_0 \in \Density(\R), \tilde{\rho}_0 \in
  \Density(\X \otimes \R)$ as
  \begin{equation*}
    \hat{\rho}_0 = \frac{\bra{\Phi}_{\reg{X}} \,\rho'_0\,
      \ket{\Phi}_{\reg{X}}}{\tr_{\rho'_0} \bigl( \ket{\Phi} \bra{\Phi}_{\reg{X}}
      \bigr)}, \quad \tilde{\rho}_0 = \ket{\Phi}\bra{\Phi}_{\reg{X}}
    \otimes \hat{\rho}_0.
  \end{equation*}
  By the spectrum of $H_{\text{in}}$ and Lemma~\ref{lem:gapped}, it
  follows that
  \begin{equation}
    \label{eq:mc-dtr2}
    \dtr(\rho'_0, \tilde{\rho}_0) \le O(\epsilon_2),
  \end{equation}
  where $\epsilon_2 = n^{1/2}N^{1/2}\epsilon_1^{1/4}$.
  Define state $\tilde{\rho}$ as
  \begin{equation*}
    \tilde{\rho} = \hat{U} (\ket{V}\bra{V}\otimes \tilde{\rho}_0)
    \hat{U}^*.
  \end{equation*}
  By Eqs.~\eqref{eq:mc-dtr}, ~\eqref{eq:mc-dtr2} and the triangle
  inequality, we have
  \begin{equation*}
    \dtr(\rho, \tilde{\rho}) \le O(\epsilon_2).
  \end{equation*}

  For $i=0,1, \ldots, 2n-1$ and $l=1,2,\ldots, L$, define states
  $\rho_{i,l}, \tilde{\rho}_{i,l} \in \Density(\X \otimes \R)$,
  \begin{equation*}
    \rho_{i,l} = \frac{\bra{N_0^{i,l}} \rho
      \ket{N_0^{i,l}}}{\tr_{\rho} \bigl(
      \ket{N_0^{i,l}}\bra{N_0^{i,l}} \bigr)}, \quad
    \tilde{\rho}_{i,l} = \frac{\bra{N_0^{i,l}} \tilde{\rho}
      \ket{N_0^{i,l}}}{\tr_{\tilde{\rho}} \bigl(
      \ket{N_0^{i,l}}\bra{N_0^{i,l}} \bigr)}.
  \end{equation*}
  By the definition of $\tilde{\rho}$, it is easy to see
  \begin{equation}
    \label{eq:mc-tri}
    \tilde{\rho}_{i,l} = \hat{U}_{N_0^{i,l} \leftarrow 1}
    \,\tilde{\rho}_0\, \hat{U}_{N_0^{i,l} \leftarrow 1}^*.
  \end{equation}
  As the product of all reflections in $\mathfrak{U}_l$ reduces to
  $\I$, we have
  \begin{equation}
    \label{eq:mc-states}
    \tilde{\rho}_{0,l} = \tilde{Q}_l^* \,\tilde{\rho}_{0}\,
    \tilde{Q}_l,
  \end{equation}
  for $\tilde{Q}_l = \hat{P}_{1} \hat{P}_{2} \cdots \hat{P}_{q_l}$
  where $\mathfrak{Q}_l = \bigl( P_j \bigr)_{j=1}^{q_l}$.
  We note that, as in the discussion before the construction of the
  $(n,k)$-constraint propagation game, states of the form in
  Eq.~\eqref{eq:mc-states} are the states on which we need to
  establish the approximate satisfaction of the $(n,k)$-constraint
  system.

  Consider strategy $\tilde{\Strategy}$ which is the same as
  $\Strategy$ except that the shared state becomes $\tilde{\rho}$
  instead of $\rho$.
  The value of $\tilde{\Strategy}$ is at least
  $1-O(\epsilon_2^{1/2})$.
  Consider the \game{Constraint Check} part.
  For $j\in [m_{n,k}]$, if edge $e_j = (N_{j-1}^{i,l},N_j^{i,l}) \in
  E_{\cons}^{i,l}$ is sampled, the referee rejects $\tilde{\Strategy}$
  with probability
  \begin{equation*}
    \tr_{\tilde{\rho}} H_j^i \le O \bigl( nLm_{n,k}\epsilon_2^{1/2}
    \bigr),
  \end{equation*}
  where
  \begin{equation*}
    H_j^{i,l} = \frac{1}{2} \Bigl[ \bigl(
    \ket{N_{j-1}^{i,l}}\bra{N_{j-1}^{i,l}} +
    \ket{N_j^{i,l}}\bra{N_j^{i,l}} \bigr) \otimes \I - (-1)^{\tau_j}
    \bigl(\ket{N_{j-1}^{i,l}}\bra{N_j^{i,l}} +
    \ket{N_j^{i,l}}\bra{N_{j-1}^{i,l}} \bigr) \Bigr].
  \end{equation*}
  By the definition of $\tilde{\rho}$ and $\tilde{\rho}_{i,l}$, this
  is equivalent to
  \begin{equation*}
    \Re\tr_{\tilde{\rho}_{i,l}} \bigl( \hat{U}_{N^{i,l}_j\leftarrow
      N^{i,l}_0+1}^* \hat{U}_{N^{i,l}_{j-1}\leftarrow N^{i,l}_0+1}
    \bigr) \approx_{\epsilon_3} (-1)^{\tau_j},
  \end{equation*}
  for $\epsilon_3 = nLm_{n,k}N\epsilon_2^{1/2}$.
  This can be further simplified to
  \begin{equation*}
    \Re\tr_{\tilde{\rho}_{i,l}} \bigl( \hat{C}_1 \hat{C}_2 \cdots
    \hat{C}_{j} \hat{C}_{j-1}^* \cdots \hat{C}_1^* \bigr)
    \approx_{\epsilon_3} (-1)^{\tau_j},
  \end{equation*}
  or equivalently, using the shorthand notion introduced in the proof
  of Lemma~\ref{lem:cons-prop-game},
  \begin{equation*}
    \Re\tr_{\tilde{\rho}_{i,l}} \bigl( \hat{C}_{1,j} \,
    \hat{C}_{1,j-1}^* \bigr) \approx_{\epsilon_3} (-1)^{\tau_j}.
  \end{equation*}

  Using similar arguments as in the proof of
  Lemma~\ref{lem:cons-prop-game}, this proves that for $i=0,1,\ldots,
  2n-1$, $l\in [L]$, $j \in [m_{n,k}]$ and $\epsilon_4 =
  m_{n,k}\epsilon_3^{1/2}$,
  \begin{equation}
    \label{eq:mc-approx-sat-tilde}
    \Re \tr_{\tilde{\rho}_{i,l}} \hat{C}_j \approx_{\epsilon_4} (-1)^{\tau_j}.
  \end{equation}

  For $i=0,1,\ldots, 2n-1$, and $l\in [L]$, define states
  $\hat{\rho}_{i,l} \in \Density(\R)$ as
  \begin{equation*}
    \hat{\rho}_{i,l} = \tr_{\X} \tilde{\rho}_{i,l}.
  \end{equation*}
  As $\hat{C}_j$ is an operator in $\Unitary(\R)$,
  Eq.~\eqref{eq:mc-approx-sat-tilde} can be written as
  \begin{equation}
    \label{eq:mc-approx-sat}
    \Re \tr_{\hat{\rho}_{i,l}} \hat{C}_j \approx_{\epsilon_4} (-1)^{\tau_j}.
  \end{equation}

  By the definition of $\tilde{\rho}_{i,l}$ in Eq.~\eqref{eq:mc-tri},
  we have for $i=1,2,\ldots, 2n-1$,
  \begin{equation*}
    \tilde{\rho}_{i,l} = \hat{U}_{N_0^{i,l}\leftarrow N_0^{i-1,l}+1}
    \,\tilde{\rho}_{i-1,l}\, \hat{U}_{N_0^{i,l}\leftarrow N_0^{i-1,l}+1}^*.
  \end{equation*}
  Equivalently, for $i=1,2,\ldots, n$,
  \begin{equation*}
    \tilde{\rho}_{2i-1,l} = \Lambda_{2i-1}(\hat{X}_i)
    \hat{C}_{1,m_{n,k}}^* \tilde{\rho}_{2i-2,l} \hat{C}_{1,m_{n,k}}
    \Lambda_{2i-1}(\hat{X}_i),
  \end{equation*}
  and for $i=1,2,\ldots, n-1$,
  \begin{equation*}
    \tilde{\rho}_{2i,l} = \Lambda_{2i}(\hat{Z}_i)
    \hat{C}_{1,m_{n,k}}^* \tilde{\rho}_{2i-1,l} \hat{C}_{1,m_{n,k}}
    \Lambda_{2i}(\hat{Z}_i).
  \end{equation*}

  Taking partial trace over $\X$ on the above two equations, we have
  for $i=1,2,\ldots, 2n-1$,
  \begin{equation*}
    \hat{\rho}_{i,l} = \mathfrak{F}_i \bigl( \hat{C}_{1,m_{n,k}}^*
    \hat{\rho}_{i-1,l} \hat{C}_{1,m_{n,k}} \bigr),
  \end{equation*}
  where $\mathfrak{F}_1, \mathfrak{F}_2, \ldots, \mathfrak{F}_{2n-1}$
  are quantum channels defined as
  \begin{equation*}
    \begin{split}
      \mathfrak{F}_{2i-1}(\rho) & = \frac{\rho + \hat{X}_i \rho
        \hat{X}_i}{2}, \text{ for } i=1,2,\ldots, n,\\
      \mathfrak{F}_{2i} (\rho) & = \frac{\rho + \hat{Z}_i \rho
        \hat{Z}_i}{2}, \text{ for } i=1,2,\ldots,n-1.
    \end{split}
  \end{equation*}

  Define states $\check{\rho}_{i,l} \in \Density(\R)$ for
  $i=0,1,\ldots, 2n-1$,
  \begin{equation*}
    \check{\rho}_{i,l} = \mathfrak{F}_i \circ \mathfrak{F}_{i-1} \circ
    \cdots \circ \mathfrak{F}_1 (\hat{\rho}_{0,l}).
  \end{equation*}

  We claim that Eq.~\eqref{eq:mc-approx-sat} then implies, for
  $\epsilon_5 = nm_{n,k}\epsilon_4$,
  \begin{equation}
    \label{eq:mc-approx-sat-check}
    \Re \tr_{\check{\rho}_{i,l}} (\hat{C}_j) \approx_{\epsilon_5} (-1)^{\tau_j}.
  \end{equation}
  In fact, by Eq.~\eqref{eq:mc-approx-sat} and the definition of
  $\hat{\rho}_{i,l}$, we have
  \begin{equation*}
    \Re \tr_{\hat{\rho}_{i-1,l}} \Bigl(\hat{C}_{1,m_{n,k}}
    \mathfrak{F}_i \bigl( \hat{C}_j \bigr) \hat{C}_{1,m_{n,k}}^*
    \Bigr) \approx_{\epsilon_4} (-1)^{\tau_j}.
  \end{equation*}
  Using Lemma~\ref{lem:approx-stab} and Eq.~\eqref{eq:mc-approx-sat},
  the above approximation is simplified to
  \begin{equation*}
    \Re \tr_{\hat{\rho}_{i-1,l}} \Bigl(
    \mathfrak{F}_i \bigl( \hat{C}_j \bigr)
    \Bigr) \approx_{m_{n,k}\epsilon_4} (-1)^{\tau_j}.
  \end{equation*}
  The claim in Eq.~\eqref{eq:mc-approx-sat-check} the follows by a
  repeated application of the above procedure and the choice of
  $\epsilon_5$.

  Lemma~\ref{lem:XZ} applied to Eq.~\eqref{eq:mc-approx-sat}
  and~\eqref{eq:mc-approx-sat-check} with constraint $X_uZ_uX_uZ_u =
  -\I$ implies the existence of unitary operators $V_u\in \Lin(\R,\B_u
  \otimes \R')$ such that
  \begin{equation*}
    \tilde{Z}_u = \hat{Z}_u,
  \end{equation*}
  and
  \begin{equation}
    \label{eq:mc-cond-1}
    \begin{split}
      \Re\tr_{\hat{\rho}_{i,l}} \bigl( \tilde{X}_u \hat{X}_u \bigr) &
      \approx_{\epsilon_4} 1,\\
      \Re\tr_{\check{\rho}_{i,l}} \bigl( \tilde{X}_u \hat{X}_u \bigr)
      & \approx_{\epsilon_5} 1,
    \end{split}
  \end{equation}
  where $\tilde{D}_u = V_u^*(D\otimes \I) V_u$ for $D\in \{X,Z\}$.

  Similarly, using the condition for constraints $D_vX_uZ_uX_uZ_uD_v =
  -\I$, we have
  \begin{equation}
    \label{eq:mc-cond-2}
    \begin{split}
      \Re\tr_{\hat{\rho}_{i,l}} \bigl( \hat{D}_v \tilde{X}_u \hat{X}_u
      \hat{D}_v
      \bigr) & \approx_{\epsilon_4} 1,\\
      \Re\tr_{\check{\rho}_{i,l}} \bigl( \hat{D}_v \tilde{X}_u
      \hat{X}_u \hat{D}_v \bigr) & \approx_{\epsilon_5} 1,
    \end{split}
  \end{equation}

  As in the proof of Theorem~\ref{thm:ms-game}, define isometry $V$ as
  in Eq.~\eqref{eq:ms-V} and quantum channels $\mathfrak{T}_u$ as in
  Eq.~\eqref{eq:ms-T} and operators
  \begin{equation*}
    \check{D}_u = \mathfrak{T}_1 \circ \mathfrak{T}_2 \circ \cdots
    \circ \mathfrak{T}_{u} (\tilde{D}_u).
  \end{equation*}

  Define operators $R$ and $R'$ as
  \begin{equation*}
    \begin{split}
      R & = \mathfrak{T}_2 \circ \cdots \circ
      \mathfrak{T}_{u} (\tilde{D}_u),\\
      R' & = \frac{R + \tilde{Z}_1 R \tilde{Z}_1 }{2}.
    \end{split}
  \end{equation*}

  For $l\in [L]$, we have
  \begin{equation*}
    \begin{split}
      \Re\tr_{\hat{\rho}_{0,l}} \bigl( \check{D}_u \hat{D}_u \bigr) &
      = \frac{1}{2}\Re\tr_{\check{\rho}_{0,l}} \bigl( R' \hat{D}_u +
      \tilde{X}_1 R' \tilde{X}_1 \hat{D}_u \bigr)\\
      & \approx_{\epsilon_5^{1/2}}
      \frac{1}{2}\Re\tr_{\check{\rho}_{0,l}} \bigl( R' \hat{D}_u +
      \hat{X}_1 R' \hat{X}_1 \hat{D}_u \bigr)\\
      & \approx_{\epsilon_5^{1/2}}
      \frac{1}{2}\Re\tr_{\check{\rho}_{0,l}} \bigl( R' \hat{D}_u +
      \hat{X}_1 R' \hat{D}_u \hat{X}_1 \bigr)\\
      & = \Re\tr_{\check{\rho}_{1,l}} (R' \hat{D}_u),\\
    \end{split}
  \end{equation*}
  where the first approximation is by Eqs.~\eqref{eq:mc-cond-1}
  and~\eqref{eq:mc-cond-2}, and the second is by
  Eq.~\eqref{eq:mc-approx-sat-check}.
  Similarly, we have
  \begin{equation*}
    \begin{split}
      \Re\tr_{\check{\rho}_{1,l}} (R' \hat{D}_u) & =
      \frac{1}{2}\Re\tr_{\check{\rho}_{1,l}} \bigl( R \hat{D}_u +
      \tilde{Z}_1 R \tilde{Z}_1 \hat{D}_u \bigr)\\
      & = \frac{1}{2}\Re\tr_{\check{\rho}_{1,l}} \bigl( R \hat{D}_u +
      \hat{Z}_1 R \hat{Z}_1 \hat{D}_u \bigr)\\
      & \approx_{\epsilon_5^{1/2}}
      \frac{1}{2}\Re\tr_{\check{\rho}_{1,l}} \bigl( R \hat{D}_u +
      \hat{Z}_1 R \hat{D}_u \hat{Z}_1\bigr)\\
      & = \Re\tr_{\check{\rho}_{2,l}} (R\hat{D}_u),
    \end{split}
  \end{equation*}

  Repeating the above procedure, we have for $l\in [L]$ and
  $\epsilon_6=n\epsilon_5^{1/2}$,
  \begin{equation}
    \label{eq:mc-cond-3}
    \Re\tr_{\hat{\rho}_{0,l}} \bigl( \check{D}_u \hat{D}_u \bigr)
    \approx_{\epsilon_6} 1.
  \end{equation}
  This proves Eq.~\eqref{eq:mc-P} in the theorem for the case $P\in
  \Pauli_{n,1}$ with state $\hat{\rho}_0$ by taking $l=0$.

  Consider constraints of the form $P \prod_{u\in J} D_u = \I$ for $P
  = \prod_{u\in J} D_u$, we have by Eq.~\eqref{eq:mc-approx-sat},
  \begin{equation*}
    \Re\tr_{\hat{\rho}_{0,l}} \bigl( \hat{P} \, \prod_{u\in J} \hat{D}_u
    \bigr) \approx_{\epsilon_4} 1.
  \end{equation*}
  By Eq.~\eqref{eq:mc-cond-3}, we have
  \begin{equation*}
    \Re\tr_{\hat{\rho}_{0}} \bigl( \hat{P} \, \prod_{u\in J} \check{D}_u
    \bigr) \approx_{\epsilon_6^{1/2}} 1.
  \end{equation*}
  Finally, by Lemma~\ref{lem:overline2} and
  Lemma~\ref{lem:approx-stab-3}, we have
  \begin{equation*}
    \Re\tr_{\hat{\rho}_{0}} \bigl( \hat{P}\, \check{P} \bigr)
    \approx_{\epsilon_6^{1/2}} 1.
  \end{equation*}

  By constraints of the form $(P|Q)P = \I$ in
  Eq.~\eqref{eq:mc-approx-sat}, we have
  \begin{equation*}
    \Re\tr_{\hat{\rho}_{0,l}} \bigl( \hat{P|Q}\, \hat{P} \bigr)
    \approx_{\epsilon_4} 1,
  \end{equation*}
  and using Eq.~\eqref{eq:mc-cond-3} and Lemma~\ref{lem:overline2},
  this implies that
  \begin{equation*}
    \Re\tr_{\hat{\rho}_{0,l'}} \bigl(\hat{P|Q}\, \check{P}\bigr)
    \approx_{\epsilon_6^{1/2}} 1,
  \end{equation*}
  for all index $l'$ that corresponds to sequences of the
  concatenations of empty primitive sequence and derived reflections.
  Lemma~\ref{lem:derived2} then completes the proof for state
  $\hat{\rho}_0$.

  Finally, using a similar argument as in the proof for
  Lemma~\ref{lem:cons-prop-game}, the statements in the theorem are
  proved for the state $\rho_0 = \tr_{\X}\rho_{0,0}$.
\end{proof}

We mention that even though the definition of the $(n,k)$-constraint
system game is an extended nonlocal game with a single player, it is
straightforward to extend it the case of $r$ players.
For this, consider a copy of the $(n,k)$-constraint system for each
player and take the union of all the constraints.
The referee then does the same as in the one player case and direct
questions of reflections of the $i$-th copy of the constraint system
to the player $(i)$.
Similar rigidity results can be established for this $r$-player
version of the $(n,k)$-constraint propagation game.

The use of the $(n,k)$-constraint propagation game is to enforce that
the players follow the measurement specifications. In order to check
other properties, such as the correct propagation of some quantum
computation, we need to first measure the clock register $\reg{S}$ and
continue only when the result is $0$. This is not very efficient and
introduces another polynomial loss in efficiency. This is not
important in our case. But this loss may be recovered by using a more
intricate decoding of the players' answers when measurement outcome
other than $0$ appears in the $(n,k)$-constraint propagation game.

\section{From Interactive Proofs to Nonlocal Games}
\label{sec:proof}
\subsection{Localization with Honest Players}

In this section, we show how to transform an $r$-prover quantum
interactive proof system to a game between a quantum referee and $r$
honest players, in which the referee measures and asks each player to
measure at most constant number of qubits.
Therefore the questions in the game consist of at most logarithmic
number of bits.
We start with the following lemma proved in~\cite{KKMV08}.

\begin{lemma}
  For any $r,m\in \poly$, $c'$ and $s'$ satisfying $c'-s'\in
  \poly^{-1}$, there exists an $s\in 1 - \poly^{-1}$, such that
  $\QMIP*(r,m,c',s') \subseteq \QMIP*(r,3,1,s)$.
\end{lemma}

It therefore suffices to start with $r$-prover, $3$-message quantum
interactive proof systems with perfect completeness.
Recall that $\reg{V}$ is the private quantum register of the verifier
$V$, $\reg{P}_i$ for $i\in [r]$ is the private quantum register of
prover $P_i$, and $\reg{M}_i$ for $i\in [r]$ are the quantum registers
for the message qubits.
The registers $\reg{V}$ and $\reg{M}_i$ consist of $q_V, q_M\in \poly$
qubits respectively, while there are no restrictions on the sizes of
$\reg{P}_i$'s.
The associated Hilbert spaces of these registers are denoted as $\V,
\M_i, \P_i$ respectively.
Registers $\reg{M}, \reg{P}$ refer to the collection of quantum
registers $\bigl( \reg{M}_i \bigr)_{i=1}^r$ and $\bigl( \reg{P}_i
\bigr)_{i=1}^r$ respectively, and $\M = \bigotimes_{i=1}^r\M_i$ and
$\P = \bigotimes_{i=1}^r\P_i$ are their associated Hilbert spaces.
Let $(V^1, V^2)$ and $(W^1, W^2, \ldots, W^r)$ be the polynomial-time
quantum verifier and the quantum provers' circuits for the $3$-message
interactive proof system.
For simplicity, we assume that both $V^1$ and $V^2$ consist of $L$
elementary gates from some universal gate set specified below.
If they have different size, one can add extra elementary gates that
act on auxiliary qubits.
Define $T=2L+1$ be the total number of time steps including the
provers' action as in Fig.~\ref{fig:QMIP}.

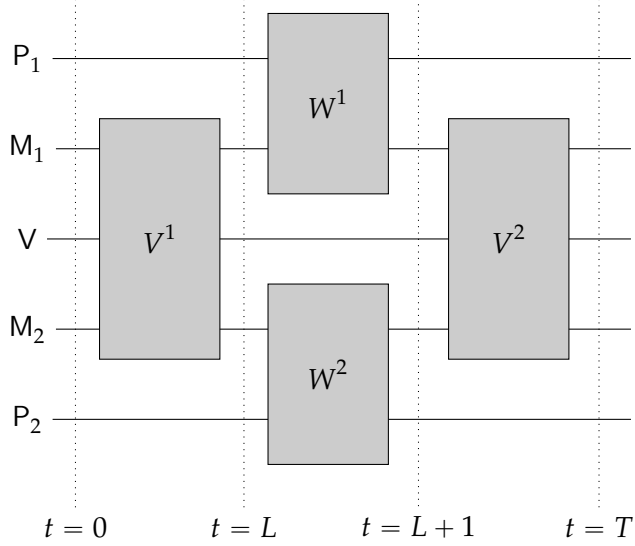
\begin{figure}[!htb]
  \centering
    \begin{tikzpicture}[scale=.8,
      V/.style={draw, minimum height = 3.2cm, minimum width = 1.6cm,
        fill=ChannelColor},
      P/.style={draw, minimum height = 2.4cm, minimum width = 1.6cm,
        fill=ChannelColor}]

      \node (inP1) at (-2,3) {$\reg{P}_1$};
      \node (inM1) at (-2,1.5) {$\reg{M}_1$};
      \node (inV) at (-2,0) {$\reg{V}$};
      \node (inM2) at (-2,-1.5) {$\reg{M}_2$};
      \node (inP2) at (-2,-3) {$\reg{P}_2$};

      \node (outP1) at (8.2,3) {};
      \node (outM1) at (8.2,1.5) {};
      \node (outV) at (8.2,0) {};
      \node (outM2) at (8.2,-1.5) {};
      \node (outP2) at (8.2,-3) {};

      \draw (inM1) -- (outM1);
      \draw (inM2) -- (outM2);
      \draw (inP1) -- (outP1);
      \draw (inP2) -- (outP2);
      \draw (inV) -- (outV);

      \node[V] (V1) at (0.2,0) {$V^1$};
      \node[P] (P1) at (3,2.25) {$W^1$};
      \node[P] (P2) at (3,-2.25) {$W^2$};
      \node[V] (V2) at (6,0) {$V^2$};

      \node (t0) at (-1.2,-4.8) {$t=0$};
      \node (tL) at (1.6,-4.8) {$t=L$};
      \node (tL1) at (4.5,-4.8) {$t=L+1$};
      \node (tT) at (7.5,-4.8) {$t=T$};

      \draw[dotted] (t0) -- (-1.2,4);
      \draw[dotted] (tL) -- (1.6,4);
      \draw[dotted] (tL1) -- (4.5,4);
      \draw[dotted] (tT) -- (7.5,4);
    \end{tikzpicture}
  \caption{An illustration of two-prover, three-message quantum
    interactive proof system with verifier's circuits $(V^1, V^2)$ and
    provers' circuits $(W^1, W^2)$.}
  \label{fig:QMIP}
\end{figure}

In the $r$-prover, $3$-message interactive proof system, the provers
initialize a state $\ket{\Psi}\in \M\otimes \P$, send registers
$\reg{M}_i$ to the verifier.
The verifier then applies $V^1 \in \Unitary(\V\otimes \M)$ and sends
$\reg{M}_i$ to prover $P_i$.
The provers apply $W^i \in \Unitary(\M_i \otimes \P_i)$ and sends
$\reg{M}_i$ back to the verifier.
Finally the verifier applies $V^2 \in \Unitary(\V\otimes \M)$ and
accepts if and only if the first qubit $\reg{V}$ measures to $1$.
Define projection $\Pi_{\text{acc}} = \ket{1}\bra{1}_{\reg{V},1}$.
For any verifier described by $(V^1, V^2)$, The maximum acceptance
probability of the provers is given by
\begin{equation*}
  \MAP(V^1,V^2) = \sup \norm{ \Pi_{\text{acc}} V^2 W V^1
    \bigl( \ket{0^{q_V}}_{\reg{V}} \ket{\Psi}_{\reg{M},\reg{P}} \bigr)}^2,
\end{equation*}
where $W = \bigotimes_{i=1}^r W^i$, and the supreme is taken over all
possible Hilbert spaces $\P_i$, all quantum state $\ket{\Psi}$ and all
quantum provers $W^i \in \Unitary(\M_i\otimes \P_i)$.
For any language $A \in \QMIP*(r,3,1,s)$, we have $\MAP(V^1,V^2) = 1$
if $x\in A$ and $\MAP(V^1, V^2)\le s$ if $x\not\in A$ by the
completeness and the soundness of the proof system.

Our transformation from this three-message interactive proof system to
a one-round multi-player game with honest players can be regarded as a
generalization of the circuit-to-Hamiltonian transformation of Kitaev
to the interactive setting.
The multi-player game consists of a referee and $r$ players, playing
the role of the verifier and provers respectively.
The referee possesses a clock register $\reg{C}$ and a register
$\reg{V}$.
For each $i\in [r]$, player $(i)$ possesses register $\reg{B}_i$, a
copy of the $L+1$-th qubit in $\reg{C}$, and two registers $\reg{M}_i,
\reg{P}_i$.
We use unary clock encoding for the clock register $\reg{C}$
consisting of $T$ qubits.
The legal states of the register are spanned by states of the form
$\bigket{1^t0^{T-t}}$.
Let $\reg{T}$ be the collection of the clock register $\reg{C}$ and
all registers $\bigl(\reg{B}_i \bigr)_{i=1}^r$.
The legal clock states are spanned by
\begin{equation*}
  \bigket{\hat{t}}_{\reg{T}} = \bigket{1^t0^{T-t}}_{\reg{C}} \otimes
  \Bigl( \bigotimes_{i=1}^r\; \bigket{\delta_t}_{\reg{B}_i} \Bigr),
\end{equation*}
where $\delta_t \in \{0,1\}$ equals $0$ if $t\le L$ and equals $1$
otherwise.
We will use $\C,\T,\B_i$ to denote the corresponding Hilbert spaces of
register $\reg{C}, \reg{T}, \reg{B}_i$.
Define Hilbert space $\H$ to be $\T \otimes \V \otimes \M\otimes \P$.

In the game, there are two possible types of questions that the
referee may ask.
The first type is a measurement specification of either one or several
commuting Pauli operators on qubits in registers $\reg{B}_i$ and
$\reg{M}_i$.
The players are honest in the sense that they will perform the
measurements corresponding to the received Pauli operators and reply
with the measurement outcome.
The second type consists of only one special question $\xap$, which
asks player $(i)$ to measure $X$ on $\reg{B}_i$ after the application
of the prover's circuit $(W^i)^*$ conditioned on the qubit in
$\reg{B}_i$.
The player is however not required to follow this protocol exactly.

The game proceeds as follows.
The players first prepare a state $\rho$ in all the registers of the
referee and the players.
The players are not allowed to communicate after this initialization
step.
The referee sends questions to the players as in
Fig.~\ref{fig:honest-game} and the players respond honestly in the
sense described above.
Finally, the referee determines whether to accept or reject based on
the questions, answers and his own measurement outcomes.
The strategy of the players can be described by $\Strategy = (\rho,
\hat{\xap}^{(i)})$, for state $\rho \in \Density(\H)$ and reflection
$\hat{\xap}^{(i)}$ the players applies for question $\xap$.

In order to use $XZ$-form Pauli measurements in the game, we will
assume that circuit $V^1, V^2$ uses two elementary gates---the Toffoli
gate and the Hadamard gate~\cite{Shi02}.
We further assume that the each Hadamard gate on a qubit in $\reg{V}$
(or $\reg{M}_i$) appears in pair with another Hadamard gate on
$\reg{V}$ ($\reg{M}_i$ respectively).
This can be easily achieved by adding a dummy qubit to these registers
and it is a technique first used in~\cite{BJSW16} to simplify the
design of a zero-knowledge proof for \QMA.
With this convention of the verifier circuit, the referee will play
the \game{Hadamard Check} and \game{Toffoli Check} given in
Figs.~\ref{fig:hadamard-game} and~\ref{fig:toffoli-game} to check the
propagation of the verifier's circuits.

\begin{figure}[!htb]
  \begin{shaded}
    \ul{Hadamard Check}\\[1em]
    Let $u_1,u_2$ be the two qubits that the two Hadamard gates act on
    in the Hadamard check at time $t$, the referee measures
    $X_{\reg{C},t}$ with outcome $x$.
    He samples $j\in \{0,1\}$ uniformly at random and does the
    following:
    \begin{enumerate}
    \item If $j=0$, measures or asks the players to measure
      $X_{u_1}X_{u_2}, Z_{u_1}Z_{u_2}$ (if both $u_1,u_2$ are qubits
      form register $\reg{M}$) and let $a_1, a_2$ be the two outcome
      bits; rejects if $x \oplus a_1 = x \oplus a_2 = 1$ and accepts
      otherwise.
    \item If $j=1$, measures or asks the players to measure
      $X_{u_1}Z_{u_2}, Z_{u_1}X_{u_2}$ and let $a_1,a_2$ be the two
      outcome bits; rejects if $x\oplus a_1 = x\oplus a_2 = 1$ and
      accepts otherwise.
    \end{enumerate}
  \end{shaded}
  \caption{The protocol that checks the Hadamard gate propagation at
    time step $t$.}
  \label{fig:hadamard-game}
\end{figure}

\begin{lemma}
  \label{lem:hadamard-game}
  Let $\rho$ be the shared state used in the \game{Hadamard Check} and
  $U_t$ be the corresponding doubled Hadamard gate, the referee
  rejects with probability
  \begin{equation*}
    \frac{1}{4} \tr_\rho \bigl[ \I - X_{\reg{C},t} \otimes U_t \bigr].
  \end{equation*}
\end{lemma}

\begin{proof}
  The referee rejects with probability
  \begin{equation*}
    \begin{split}
      \frac{1}{2} \tr_\rho \Bigl[ \frac{\I-X_{\reg{C},t}}{2} \otimes
      \frac{\bigl[ (\I+XX)(\I+ZZ) \bigr]_{u_1,u_2}}{4} +
      \frac{\I+X_{\reg{C},t}}{2} \otimes \frac{\bigl[ (\I-XX)(\I-ZZ)
        \bigr]_{u_1,u_2}}{4} + \\
      \frac{\I-X_{\reg{C},t}}{2} \otimes \frac{\bigl[ (\I+XZ)(\I+ZX)
        \bigr]_{u_1,u_2}}{4} + \frac{\I+X_{\reg{C},t}}{2} \otimes
      \frac{\bigl[ (\I-XZ)(\I-ZX) \bigr]_{u_1,u_2}}{4} \Bigr].
    \end{split}
  \end{equation*}
  The above express simplifies to
  \begin{equation*}
    \frac{1}{4} \tr_\rho \bigl[ \I - X_{\reg{C},t} \otimes U_t \bigr]
  \end{equation*}
  by a direct calculation.
\end{proof}

\begin{figure}[!htb]
  \begin{shaded}
    \ul{Toffoli Check}\\[1em]
    Let $u_1,u_2,u_3$ be the three qubits the Toffoli gate acts on
    with $u_3$ the target qubit.
    In the Toffoli check at time $t$, the referee measures
    $X_{\reg{C},t}$ with outcome $x$.
    He samples $j\in \{0,1\}$ uniformly at random, accepts if $j=1$
    and continues otherwise.
    He measures or asks the players to measure $Z_{u_1}, Z_{u_2},
    X_{u_3}$ (if any of $u_1,u_2,u_3$ is a qubit from register
    $\reg{M}$) and let $a_1, a_2, a_3$ be the three outcome bits;
    rejects if $a_1=a_2=1, x\oplus a_3=1$ or $a_1a_2 = 0, x=1$ and
    accepts otherwise.
  \end{shaded}
  \caption{The protocol that checks the Toffoli gate propagation at
    time step $t$.}
  \label{fig:toffoli-game}
\end{figure}

\begin{lemma}
  \label{lem:toffoli-game}
  Let $\rho$ be the shared state used in the \game{Toffoli Check} and
  $U_t$ be the corresponding Toffoli gate, the referee rejects with
  probability
  \begin{equation*}
    \frac{1}{4} \tr_\rho \bigl[ \I - X_{\reg{C},t} \otimes U_t \bigr].
  \end{equation*}
\end{lemma}

\begin{proof}
  Let $u_1,u_2,u_3$ be the qubits that $U_t$ acts on.
  The referee rejects with probability
  \begin{equation*}
    \frac{1}{2} \tr_\rho \Bigl[\frac{\I-X_{\reg{C},t}}{2} \otimes
    \bigl(\I - \ket{11}\bra{11} \bigr)_{u_1,u_2} +
    \frac{\I-X_{\reg{C},t}X_{u_3}}{2} \otimes \ket{11}
    \bra{11}_{u_1,u_2} \Bigr] = \frac{1}{4} \tr_\rho \bigl[\I -
    X_{\reg{C},t}\otimes U_t \bigr].
  \end{equation*}
\end{proof}

Note that we have scaled down the rejection probability by a half
using the random bit $j$ on purpose so that both checks have the
rejection probabilities of the same form.

\begin{figure}[!htb]
  \begin{shaded}
    \ul{Multi-Player One-Round Game for $\QMIP*$ with Honest Players}\\[1em]
    The referee performs the following checks with equal probability:
    \begin{enumerate}

    \item \game{Clock Check}.
      The referee checks the validity of clock states.
      He does the following with equal probability:
      \begin{enumerate}
      \item Randomly samples $t\in [T-1]$; measures $Z_{\reg{C},t},
        Z_{\reg{C},t+1}$ and rejects if the outcomes are $0,1$
        respectively; accepts otherwise.
      \item Randomly samples $i\in [r]$; sends measurement
        specification $Z_{\reg{B}_i}$ to player $(i)$; measures
        $Z_{\reg{C},L+1}$ and rejects if the outcome is different from
        the player's answer bit; accepts otherwise.
      \end{enumerate}

    \item \game{Verifier Propagation Check}.
      The referee checks propagation of the verifier steps:
      \begin{enumerate}
      \item Samples $t\in [T]\setminus \{L+1\}$ uniformly at random;
      \item Measures $Z_{\reg{C},t-1}$ and $Z_{\reg{C},t+1}$ (if
        $t=1$, assume that the first measurement always has outcome
        $1$, and if $t=T$, assume that the second measurement always
        has outcome $0$); accepts if the outcomes are not $1,0$
        respectively and continues otherwise;
      \item If $U_t$ is a Toffoli gate, does the \game{Toffoli Check}
        at time $t$; accepts or rejects as in the \game{Toffoli
          Check};
      \item If $U_t$ consists of two Hadamard gates, does the
        \game{Hadamard Check} at time $t$; accepts or rejects as in
        the \game{Hadamard Check}.
      \end{enumerate}

    \item \game{Prover Propagation Check}.
      The referee checks the propagation of the provers' step:
      \begin{enumerate}
      \item Measures $Z_{\reg{C},L}$, $Z_{\reg{C},L+2}$ and accepts if
        the outcomes are not $1, 0$ respectively; continues otherwise;
      \item Sends $\xap$ to player $(i)$ and receives $a^{(i)}$ back;
      \item Measures $X_{\reg{C},L+1}$ and accepts if the outcome
        $a=\bigoplus a^{(i)}$; rejects otherwise.
      \end{enumerate}

    \item \game{Initialization Check}.
      The referee checks that the state is correctly initialized:
      \begin{enumerate}
      \item Measures $Z_{\reg{C},1}$ and accepts if the outcome is
        $1$; continues otherwise;
      \item Samples $j\in \reg{V}$; measures $Z_{\reg{V},j}$ and
        accepts if the outcome is $0$; rejects otherwise.
      \end{enumerate}

    \item \game{Output Check}.
      The output qubit should indicate acceptance in the interactive
      proof:
      \begin{enumerate}
      \item Measures $Z_{\reg{C},T}$ and accepts if the outcome is
        $0$; continues otherwise;
      \item Measures $Z_{\reg{V},1}$ and accepts if outcome is $1$;
        rejects otherwise.
      \end{enumerate}
    \end{enumerate}
  \end{shaded}
  \caption{The multi-player one-round game for \QMIP* with honest
    players.}
  \label{fig:honest-game}
\end{figure}

We use the following lemma in the analysis.

\begin{lemma}
  \label{lem:error}
  Let $h,p,s$ be positive real numbers such that $s\in 1-\poly^{-1}$,
  $h\in \poly$, $p\in \poly^{-1}$ and $p\in 1 - \poly^{-1}$.
  Let $\kappa\ge 0$ be a constant.
  Then for function
  \begin{equation*}
    f(\epsilon) = (1-p)(1-\epsilon) + p
    \min \bigl(1,s+h\epsilon^{1/\kappa} \bigr),
  \end{equation*}
  we have $\max_\epsilon f(\epsilon) \in 1-\poly^{-1}$.
\end{lemma}

\begin{proof}
  If $\epsilon \le [(1-s)/2h]^\kappa$, then
  \begin{equation*}
    f(\epsilon) \le 1-p + p(1+s)/2 = 1 - p(1-s)/2 \in 1 - \poly^{-1}.
  \end{equation*}
  Otherwise,
  \begin{equation*}
    f(\epsilon) \le (1-p)(1-\epsilon) + p = 1 - (1-p)\epsilon \in 1 -
    \poly^{-1}.
  \end{equation*}
\end{proof}

\begin{theorem}
  \label{thm:honest-game}
  For $r\in \poly$, $s\in 1 - \poly^{-1}$, there exists a real number
  $s'\in 1 - \poly^{-1}$ such that, for any language $A \in
  \QMIP*(r,3,1,s)$ and an instance $x$, the following properties hold
  for the game in Fig.~\ref{fig:honest-game}:
  \begin{enumerate}
  \item If $x\in A$, the referee accepts with certainty;
  \item If $x\not\in A$, the referee accepts with probability at most
    $s'$.
  \end{enumerate}
\end{theorem}

\begin{proof}[Proof of Theorem~\ref{thm:honest-game}]
  We first prove the case $x\in A$.
  Let $(V^1, V^2)$ and $(W^1, W^2, \ldots, W^r)$ be the
  polynomial-time quantum verifier and provers' circuits in the
  three-message quantum interactive proof system.
  Let $U_1, U_2, \ldots, U_L$ and $U_{L+2},U_{L+3},\ldots, U_T$ be the
  $L$ elementary gates in $V^1$ and $V^2$ respectively.
  Define $U_{L+1} = \bigotimes_{i=1}^r W^i$.
  Let $\ket{\Psi}$ be the state that the provers initializes in
  registers $\reg{M}, \reg{P}$.
  In the multi-player game, the honest players share state
  \begin{equation*}
    \frac{1}{\sqrt{T+1}}\sum_{t=0}^T \bigket{\hat{t}}_{\reg{T}} \otimes
    U_tU_{t-1}\cdots U_1 \bigl( \bigket{0^{q_V}}_{\reg{V}}
    \bigket{\Psi}_{\reg{M},\reg{P}} \bigr).
  \end{equation*}
  The referee in the game then accepts with certainty, a fact which
  can be verified directly but it will also become clear in the next
  part of the proof.
  This proves the first part of the theorem.

  Now we prove the second part and suppose that $x\not\in A$.
  Let $\Strategy = (\rho, \hat{\xap})$ be the strategy of the players.
  Define three games $G_2$, $G_3$ and $G_4$ derived from the honest
  player game $G$ as follows.
  In game $G_4$, the referee performs the \game{Initialization Check}
  and \game{Output Check} with equal probability.
  In game $G_3$, the referee does the \game{Prover Propagation Check}
  and $G_4$ with probability $1/3$ and $2/3$ respectively.
  In game $G_2$, the referee does the \game{Verifier Propagation
    Check} and $G_3$ with probability $1/4$ and $3/4$.
  Finally, game $G$ is equivalent to the game in which the referee
  does the \game{Clock Check} and $G_2$ with probability $1/5$ and
  $4/5$ respectively.
  We will analyze the five checks in the game sequentially as follows.

  \textit{Step 1}.
  For the \game{Clock Check}, define a Hamiltonian
  \begin{equation*}
    H_{\text{clock}} = \frac{1}{2(T-1)} \sum_{t=1}^{T-1}
    \ket{0}\bra{0}_{\reg{C},t} \otimes \ket{1}\bra{1}_{\reg{C},t+1} +
    \frac{1}{2r} \sum_{a\in \{0,1\}}\sum_{i=1}^r
    \ket{a}\bra{a}_{\reg{C},L+1} \otimes \bigl[\I -
    \ket{a}\bra{a}\bigr]_{\reg{B}_i}.
  \end{equation*}
  In this proof, we assume that the Hamiltonians are operators in
  $\Herm(\H)$.
  The referee then rejects in the \game{Clock Check} subgame with
  probability $\epsilon = \tr_\rho (H_{\text{clock}})$.
  The Hamiltonian $H_{\text{clock}}$ has an eigenspace
  $S_{\legal}\subseteq \H$ with eigenvalue $0$ spanned by the legal
  clock states of the form $\ket{\hat{t}}_{\reg{T}}
  \ket{\xi}_{\reg{V},\reg{M},\reg{P}}$, and the nonzero eigenvalues
  are least $\Omega(1/h)$ for $h=\max(T,r)$.
  We then have, by Lemma~\ref{lem:gapped},
  \begin{equation*}
    \dtr(\rho,\rho_{\legal}) \le O(\sqrt{h\epsilon}),
  \end{equation*}
  where $\rho_{\legal} = \Pi_{\legal} \rho \Pi_{\legal} /
  \tr_\rho(\Pi_{\legal})$ for $\Pi_{\legal}$ is the projection onto
  $S_{\legal}$.
  Hence, by the monotonicity of the trace distance, the referee
  accepts in game $G$ with probability at most
  \begin{equation*}
    \frac{1}{5}(1-\epsilon) + \frac{4}{5} \min \bigl( 1,s_2 +
    c\sqrt{h\epsilon} \bigr),
  \end{equation*}
  where $s_2$ is the maximum acceptance probability of the referee in
  game $G_2$ for the shared state $\rho$ supported on $S_{\legal}$,
  namely, $\rho = \Pi_{\legal} \rho \Pi_{\legal}$.
  Lemma~\ref{lem:error} then reduces our problem to proving that $s_2
  \in 1 - \poly^{-1}$.

  \textit{Step 2}.
  Suppose that the state $\rho$ satisfies $\rho = \Pi_{\legal} \rho
  \Pi_{\legal}$.
  We want to prove that the acceptance probability $s_2$ of game $G_2$
  is at most $1 - \poly^{-1}$ under this condition.
  Define a Hamiltonian
  \begin{equation*}
    \begin{split}
      H_{\propv} & = \frac{1}{4(T-1)} \Bigl[
      \ket{0}\bra{0}_{\reg{C},2}\otimes (\I - X_{\reg{C},1} \otimes
      U_1) + \sum_{\substack{t:1<t<T\\t\ne L+1}}
      \ket{10}\bra{10}_{\reg{C},t-1,t+1}\otimes (\I-X_{\reg{C},t}
      \otimes U_t) + \\
      &\qquad\qquad \ket{1}\bra{1}_{\reg{C},T-1} \otimes (\I -
      X_{\reg{C},T} \otimes U_T) \Bigr].
    \end{split}
  \end{equation*}

  It is easy to verify that
  \begin{equation*}
    H_{\propv}{\restriction_{S_{\legal}}} = \frac{1}{4(T-1)}
    \sum_{\substack{t:1\le t\le T\\t\ne L+1}}
    \Bigl[ \bigl( \ket{\hat{t-1}}\bra{\hat{t-1}} +
    \ket{\hat{t}}\bra{\hat{t}} \bigr)_{\reg{T}} \otimes \I -
    \bigl( \ket{\hat{t}}\bra{\hat{t-1}} +
    \ket{\hat{t}}\bra{\hat{t-1}} \bigr)_{\reg{T}} \otimes U_t \Bigr]
  \end{equation*}

  By Lemmas~\ref{lem:hadamard-game} and~\ref{lem:toffoli-game}, the
  referee rejects in the \game{Verifier Propagation Check} with
  probability
  \begin{equation*}
    \tr_\rho H_{\propv} = \tr_\rho \bigl(
    H_{\propv}{\restriction_{S_{\legal}}} \bigr).
  \end{equation*}

  Define unitary operator $U \in \Unitary (S_{\legal})$, and states
  $\ket{\omega^1}$, $\ket{\omega^2}$ as
  \begin{equation*}
    U = \sum_{t=0}^{L} \ket{\hat{t}}\bra{\hat{t}}_{\reg{T}} \otimes
    U_{t+1}^* U_{t+2}^* \cdots U_L^* + \sum_{t=L+1}^T
    \ket{\hat{t}}\bra{\hat{t}}_{\reg{T}} \otimes U_t U_{t-1} \cdots
    U_{L+2},
  \end{equation*}
  and
  \begin{equation*}
    \begin{split}
      \ket{\omega^1} & = \frac{1}{\sqrt{L+1}}\; \sum_{t=0}^{L}
      \;\;\ket{\hat{t}}_{\reg{T}},\\
      \ket{\omega^2} & = \frac{1}{\sqrt{L+1}} \sum_{t=L+1}^T
      \ket{\hat{t}}_{\reg{T}}.
    \end{split}
  \end{equation*}

  Define two subspaces $S_{\propv}^1$ and $S_{\propv}^2$ of
  $S_{\legal}$ as the spans respectively of states of the form
  \begin{equation*}
    \begin{split}
      \ket{\psi^1} & = U (\ket{\omega^1}\otimes \ket{\psi^1_{L}}),\\
      \ket{\psi^2} & = U (\ket{\omega^2}\otimes \ket{\psi^2_0}),
    \end{split}
  \end{equation*}
  for states $\ket{\psi^1_{L}}, \ket{\psi^2_0}\in
  \V\otimes\M\otimes\P$.
  Let $S_{\propv}$ be the direct sum of $S_{\propv}^1$ and
  $S_{\propv}^2$.
  Let $\Pi_{\propv}^1$, $\Pi_{\propv}^2$ and $\Pi_{\propv}$ be the
  projections onto subspaces $S_{\propv}^1$, $S_{\propv}^2$ and
  $S_{\propv}$ respectively.
  It is easy to verify that $S_{\propv}$ is the $0$-eigenspace of
  $H_{\propv}\restriction_{S_{\legal}}$.

  The Hamiltonian $H_{\propv}{\restriction_{S_{\legal}}}$ is the sum
  of two propagation checking Hamiltonians as in the circuit to
  Hamiltonian construction.
  It follows that the nonzero eigenvalues of
  $H_{\propv}{\restriction_{S_{\legal}}}$ is at least $\Omega(1/T^3)$.
  Suppose that the referee rejects in the \game{Verifier Propagation
    Check} part of $G_2$ with probability $\epsilon$.
  Define state
  \begin{equation*}
    \rho_{\propv} = \Pi_{\propv} \rho \Pi_{\propv} / \tr_\rho (\Pi_{\propv}).
  \end{equation*}
  It then follows similarly that
  \begin{equation*}
    \tr_\rho \Pi_{\propv} \ge 1 - O(T^3\epsilon),
  \end{equation*}
  and
  \begin{equation*}
    \dtr(\rho, \rho_{\propv}) \le O(T^{3/2}\epsilon^{1/2}).
  \end{equation*}

  The referee then accepts strategy $\Strategy$ in game $G_2$ with
  probability at most
  \begin{equation*}
    \frac{1}{4}(1-\epsilon) + \frac{3}{4} \min \bigl( 1, s_3 +
    cT^{3/2}\epsilon^{1/2} \bigr),
  \end{equation*}
  where $s_3$ is the maximum acceptance probability of the referee in
  game $G_3$ for the shared state $\rho$ supported on $S_{\propv}$.
  Lemma~\ref{lem:error} then reduces the problem to showing that $s_3
  \in 1 - \poly^{-1}$.

  \textit{Step 3}.
  Suppose that the state $\rho$ is supported on $S_{\propv}$.
  The aim is to show that the referee will reject with at least with
  inverse polynomial probability in game $G_3$.

  Define an operator
  \begin{equation*}
    H_{\propp} = \frac{1}{2}\ket{10}\bra{10}_{\reg{C},L,L+2} \otimes
    \Bigl[\I - X_{\reg{C},L+1}\otimes \Bigl(\bigotimes_{i=1}^r
    \hat{\xap}^{(i)} \Bigr)\Bigr].
  \end{equation*}

  The restriction $H_{\propp}\restriction_{S_{\propv}}$ can be
  computed as
  \begin{equation*}
    \begin{split}
      H_{\propp}\restriction_{S_{\propv}} & = \bigl(\Pi_{\propv}^1 +
      \Pi_{\propv}^2 \bigr) H_{\propp} \bigl(\Pi_{\propv}^1 +
      \Pi_{\propv}^2 \bigr)\\
      & = \frac{1}{L+1} \Bigl[ U \bigl( \ket{\omega^1} \bra{\hat{L}} +
      \ket{\omega^2} \bra{\hat{L+1}} \bigr) R \bigl(\ket{\hat{L}}
      \bra{\omega^1} + \ket{\hat{L+1}} \bra{\omega^2} \bigr) U^*
      \Bigr],
    \end{split}
  \end{equation*}
  where
  \begin{equation*}
    R = \frac{1}{2}\Bigl[ \I - X_{\reg{C},L+1} \otimes
    \Bigl(\bigotimes_{i=1}^r \hat{\xap}^{(i)} \Bigr) \Bigr].
  \end{equation*}
  Define operator $B$ as
  \begin{equation*}
    B = \ket{\omega^1}\bra{\hat{L}} + \ket{\omega^2}\bra{\hat{L+1}}.
  \end{equation*}

  The referee rejects in the \game{Prover Propagation Check} part of
  $G_3$ with probability
  \begin{equation*}
    \tr_\rho H_{\propp} = \tr_{\rho} \bigl( H_{\propp}\restriction_{S_{\propv}}
    \bigr) = \frac{1}{L+1} \tr_{\rho'} R,
  \end{equation*}
  for $\rho' = B^*U^*\rho UB$.
  If the rejection probability is $\epsilon$, it then follows that
  \begin{equation}
    \label{eq:XXX}
    \tr_{\rho'} \Bigl[ X_{\reg{C},L+1} \otimes \Bigl(\bigotimes_{i=1}^r
    \hat{\xap}^{(i)} \Bigr) \Bigr] = 1 - 2(L+1)\epsilon.
  \end{equation}

  On the other hand,
  \begin{equation}
    \label{eq:ZZ}
    \tr_{\rho'} \bigl[ Z_{\reg{C},L+1} \otimes Z_{\reg{B}_i} \bigr] =
    \tr_\rho \bigl[ UB \bigl( Z_{\reg{C},L+1} \otimes Z_{\reg{B}_i}
    \bigr) B^* U^* ] = \tr_\rho \Pi_{\propv} = 1,
  \end{equation}
  for all $i\in [r]$.
  These two conditions and Lemma~\ref{lem:approx-stab} then imply that
  \begin{equation*}
    \tr_{\rho'} \bigl[ \hat{\xap}^{(i)} Z_{\reg{B}_i} \hat{\xap}^{(i)}
    Z_{\reg{B}_i} \bigr] \approx_{L\epsilon} -1,
  \end{equation*}
  for $i\in [r]$.

  By Lemma~\ref{lem:XZ}, there exist unitary operators $W_i \in
  \Unitary(\B_i\otimes \M_i \otimes \P_i)$, such that for all $i\in
  [r]$
  \begin{equation}
    \label{eq:propp-Z}
    Z_{\reg{B}_i} = W_i^* (Z\otimes \I) W_i,
  \end{equation}
  and
  \begin{equation}
    \label{eq:propp-X}
    \dis_{\rho'} \bigl( \hat{\xap}^{(i)}, W_i^* (X\otimes \I) W_i \bigr) \le
    O(\sqrt{L\epsilon}).
  \end{equation}

  Define operator
  \begin{equation*}
    \tilde{W} = \Bigl( \bigotimes_{i=1}^r W_i \Bigr).
  \end{equation*}
  These two conditions and Eqs.~\eqref{eq:XXX} and~\eqref{eq:ZZ} then
  implies that the state
  \begin{equation*}
    \tilde{\rho} =  \tilde{W} \rho' \tilde{W}^*
  \end{equation*}
  satisfies that
  \begin{equation*}
    \tr_{\tilde{\rho}} \Bigl[ X_{\reg{C},L+1}\otimes \Bigl( \bigotimes_{i=1}^r
    X_{\reg{B}_i} \Bigr) \Bigr] \ge 1 - O(r\sqrt{L\epsilon}),
  \end{equation*}
  and
  \begin{equation*}
    \tr_{\tilde{\rho}} \bigl( Z_{\reg{C},L+1} \otimes Z_{\reg{B}_i}
    \bigr) = 1.
  \end{equation*}

  State $\tilde{\rho}$ is therefore approximately stabilized by
  $X_{\reg{C},L+1}\otimes \Bigl( \bigotimes_{i=1}^r X_{\reg{B}_i}
  \Bigr)$ and $Z_{\reg{C},L+1} \otimes Z_{\reg{B}_i}$.
  These operators generate the stabilizer for the \textrm{GHZ} state
  \begin{equation*}
    \ket{\Phi_{\GHZ}} = \frac{\ket{0}_{\reg{C},L+1} \ket{0^r}_{\reg{B}} +
      \ket{1}_{\reg{C},L+1} \ket{1^r}_{\reg{B}}}{\sqrt{2}}.
  \end{equation*}
  Let $\Pi_{\GHZ}$ be the projection to state $\ket{\Phi_{\GHZ}}$ and
  $h$ be $r\sqrt{L}$.
  We have
  \begin{equation*}
    \tr_{\tilde{\rho}} \Pi_{\GHZ} \ge 1 - O(h\sqrt{\epsilon}).
  \end{equation*}
  Define state
  \begin{equation*}
    \tilde{\rho}_{\GHZ} = \frac{\Pi_{\GHZ} \, \tilde{\rho} \,
      \Pi_{\GHZ}}{\tr_{\tilde{\rho}} \Pi_{\GHZ}}.
  \end{equation*}
  By Lemma~\ref{lem:gentle},
  \begin{equation}
    \label{eq:propp-state-approx}
    \dtr (\tilde{\rho}, \tilde{\rho}_{\GHZ}) \le O(h^{1/2}
    \epsilon^{1/4}).
  \end{equation}

  By the condition in~\eqref{eq:propp-Z}, it follows that each unitary
  $W_i$ has the block form
  \begin{equation*}
    W_i = \ket{0}\bra{0}_{\reg{B}_i} \otimes W_i^0 +
    \ket{1}\bra{1}_{\reg{B}_i} \otimes W_i^1.
  \end{equation*}
  Define
  \begin{equation*}
    \tilde{W}^0 = \bigotimes_{i=1}^r W_i^0, \tilde{W}^1 =
    \bigotimes_{i=1}^4 W_i^1.
  \end{equation*}

  As $\rho$ is supported on $S_{\propv}$, it follows by the definition
  of the state $\tilde{\rho}$ that the state $\tilde{\rho}$ is
  supported on states spanned by states of the form
  \begin{equation*}
    \ket{\hat{L}} \ket{\xi}, \ket{\hat{L+1}} \ket{\xi'}.
  \end{equation*}
  Without loss of generality, we may assume that the state $\rho$ in
  the strategy is a pure state.
  Then $\rho'$, $\tilde{\rho}$ and $\tilde{\rho}_{\GHZ}$ are all pure
  states.
  It follows that the pure state corresponding to
  $\tilde{\rho}_{\GHZ}$ can be written as
  \begin{equation*}
    \ket{\tilde{\Psi}_{\GHZ}} =
    \frac{\ket{\hat{L}}+\ket{\hat{L+1}}}{\sqrt{2}} \otimes \ket{\psi}.
  \end{equation*}

  By Eq.~\eqref{eq:propp-state-approx}, it follows that state
  \begin{equation*}
    \ket{\Psi} = U B W^* \ket{\tilde{\Psi}_{\GHZ}} = \frac{U
      \bigl[\ket{\omega^1} \bigl( \tilde{W}^0 \bigr)^* \ket{\psi} +
      \ket{\omega^2} \bigl( \tilde{W}^1 \bigr)^* \ket{\psi}
      \bigr]}{\sqrt{2}}
  \end{equation*}
  is a good approximation of the state $\rho$.

  Define state
  \begin{equation*}
    \ket{\psi_0} = U_1^* U_2^* \cdots U_L^* \bigl( \tilde{W}^0
    \bigr)^* \ket{\psi},
  \end{equation*}
  and unitary
  \begin{equation*}
    W^i = \bigl( \tilde{W}^1_i \bigr)^* \tilde{W}^0_i.
  \end{equation*}

  Let $\hat{\Strategy}$ be the strategy that uses the state and
  reflection in the following equation
  \begin{equation}
    \label{eq:propp-strategy}
    \frac{1}{\sqrt{T+1}} \sum_{t=0}^T \ket{\hat{t}} \otimes U_t
    U_{t-1} \cdots U_1 \ket{\psi_0},\quad
    \Lambda(W^i) \bigl( X\otimes \I \bigr) \Lambda(W^i)^*.
  \end{equation}
  They form a strategy that is accepted with probability
  $1-O(h^{1/2}\epsilon^{1/4})$.

  Assuming the condition that the state is supported on $S_{\propv}$,
  the strategy is accepted in game $G_3$ with probability at most
  \begin{equation*}
    \frac{1}{3}(1-\epsilon) + \frac{2}{3} \min \bigl( 1,s_4 + c
    h^{1/2}\epsilon^{1/4} \bigr),
  \end{equation*}
  where $s_4$ is the maximum acceptance probability of the referee in
  game $G_4$ if the players use a strategy of the form in
  Eq.~\eqref{eq:propp-strategy}.
  By Lemma~\ref{lem:error}, the problem then reduces to proving $s_4
  \in 1 - \poly^{-1}$.

  \textit{Step 4}.
  Define Hamiltonian
  \begin{equation*}
    H_{\text{in}} = \frac{1}{q_V} \ket{0}\bra{0}_{\reg{C},1} \otimes
    \sum_{j=1}^{q_V} \ket{1}\bra{1}_{\reg{V},j}.
  \end{equation*}

  For state $\rho$ of the form in~\eqref{eq:propp-strategy}, the
  referee rejects with probability
  \begin{equation*}
    \tr_{\rho} H_{\text{in}} = \frac{1}{q_V(T+1)} \bra{\psi_0}
    \sum_{j=1}^{q_V} \ket{1}\bra{1}_{\reg{V},j} \ket{\psi_0}.
  \end{equation*}

  Let $\Pi_{\text{in}}$ be the projection
  \begin{equation*}
    \Pi_{\text{in}} = \ket{0^{q_V}}\bra{0^{q_V}}_{\reg{V}}.
  \end{equation*}

  Suppose the referee rejects with probability $\epsilon$ in
  \game{Initialization Check}, then
  \begin{equation*}
    \bra{\psi_0} \Pi_{\text{in}} \ket{\psi_0} \ge 1 - O(h\epsilon),
  \end{equation*}
  where $h=q_V(T+1)$.

  Define $\rho_{\text{in}}$ be density matrix of the pure state
  \begin{equation*}
    \Pi_{\text{in}} \ket{\psi_0} / \norm{ \Pi_{\text{in}} \ket{\psi_0} }
  \end{equation*}
  Then
  \begin{equation*}
    \dtr(\rho_{\text{in}}, \ket{\psi_0}\bra{\psi_0}) \le
    O(\sqrt{h\epsilon}).
  \end{equation*}

  The referee rejects with probability
  \begin{equation*}
    \frac{1}{2} (1-\epsilon) + \frac{1}{2} \min \bigl(
    1,s+c\sqrt{h\epsilon} \bigr).
  \end{equation*}
  By Lemma~\ref{lem:error}, it is at most $1-\poly^{-1}$ as $s\in 1 -
  \poly^{-1}$.
  This completes the proof.
\end{proof}

\subsection{Extended Nonlocal Game for QMIP}

In this section, we transform the honest player game in the previous
subsection to an extended nonlocal game.

The referee possesses registers $\reg{C}$ and $\reg{V}$ as in the
honest player game and additional registers $\reg{S}$ and $\reg{X}$.
The register $\reg{S}$ will be used as the clock register for the
constraint propagation subgame and its size $q_S$ will be determined
correspondingly.
The player $(i)$ possess registers $\reg{B}_i$, $\reg{M}_i$ and
$\reg{P}_i$ as in the honest player game.
The questions have the same form as in the honest player game, which
can be either a measurement specification or the special question
$\xap$.
But the players are not required to play honestly anymore.
The game is specified in Fig.~\ref{fig:extended-game}.

For later convenience, we now use unary clock encoding in register
$\reg{S}$ and, to accommodate this change, the measurement $\Pi_e$
used in the $(n,k)$-constraint propagation game will be updated
accordingly as follows.

For $t\in [q_S]$ and edge $e=(t-1,t)$, the measurement $\Pi_e$ is
\begin{equation*}
  \begin{split}
    \Pi_e^0 & = \ket{10}\bra{10}_{\reg{S},t-1,t+1} \otimes
    \frac{\I + X_{\reg{S},t}}{2}\\
    \Pi_e^1 & = \ket{10}\bra{10}_{\reg{S},t-1,t+1} \otimes
    \frac{\I - X_{\reg{S},t}}{2}\\
    \Pi_e^2 & = (\I-\ket{10}\bra{10})_{\reg{S},t-1,t+1}.
  \end{split}
\end{equation*}
The measurement can be easily implemented by $X,Z$ measurements on the
$t-1$, $t$ and $t+1$-th qubit of register $\reg{S}$.

For edge $e=(t_1,t_2)$ with $t_2-t_1 =k$, the measurement $\Pi_e$ is
\begin{equation*}
  \begin{split}
    \Pi_e^0 & = \frac{1}{2}\ket{10}\bra{10}_{\reg{S},t_1,t_2+1}
    \otimes
    \sum_{a,b\in\{0,1\}}\ket{a^k}\bra{b^k}_{\reg{S},t_1+1,\ldots,t_2}\\
    \Pi_e^1 & = \frac{1}{2}\ket{10}\bra{10}_{\reg{S},t_1,t_2+1}
    \otimes \sum_{a,b\in\{0,1\}}(-1)^{a\oplus
      b}\ket{a^k}\bra{b^k}_{\reg{S},t_1+1,\ldots,t_2}\\
    \Pi_e^2 & = \I - \Pi_e^0 - \Pi_e^1.
  \end{split}
\end{equation*}
For any constant $k$, the measurement can be implemented using
collective $X,Z$ measurements on constant number of qubits.

\begin{figure}[!htb]
  \begin{shaded}
    \ul{Extended Nonlocal Game for $\QMIP*$}\\[1em]
    The referee does the following with equal probability:
    \begin{enumerate}
    \item \game{Clock Check}.
      Randomly samples $t\in [q_S-1]$; measures $Z_{\reg{S},t},
      Z_{\reg{S},t+1}$ and rejects if the outcomes are $0,1$
      respectively; accepts otherwise.

    \item \game{Constraint Propagation}.
      Plays the $(n,k)$-constraint propagation game with the
      $r$-players using registers $\reg{S}$ and $\reg{X}$ and accepts
      or rejects accordingly.

    \item \game{Output Check}.
      Measures $Z_{\reg{S},1}$ with outcome $a$; plays the honest
      player game as in Fig~\ref{fig:honest-game} using registers
      $\reg{C}$ and $\reg{V}$; rejects if $a=0$ and the honest player
      game rejects; accepts otherwise.
    \end{enumerate}
  \end{shaded}
  \caption{The extended nonlocal game for \QMIP*.}
  \label{fig:extended-game}
\end{figure}

\begin{theorem}
  \label{thm:extended-game}
  For any $r\in \poly$, $s\in 1 - \poly^{-1}$, there is a $s'\in 1 -
  \poly^{-1}$ such that for any language $A \in \QMIP*(r,3,1,s)$ and
  instance $x$, the extended game in Fig~\ref{fig:extended-game} has
  the property that
  \begin{enumerate}
  \item If $x\in A$, the referee accepts with certainty;
  \item If $x\not\in A$, the referee accepts with probability at most
    $s'$.
  \end{enumerate}
\end{theorem}

\begin{proof}
  If $x\in A$, it is easy to see that the players can will the game
  with certainty.

  Consider the case for $x\not\in A$.
  Define game $G_2$ to be the game where the referee plays the
  \game{Constraint Propagation} and \game{Output Check} parts with
  equal probability.
  Let the strategy be $\Strategy = \bigl(\rho, \{ \hat{P} \}, \{
  \hat{Q} \} \bigr)$.
  Suppose that the referee rejects with probability $\epsilon$ in the
  \game{Clock Check}.
  It then follows that
  \begin{equation*}
    \dtr \bigl(\rho, \rho_{\legal}\bigr) \le O(\sqrt{q_S\epsilon}).
  \end{equation*}
  where $\rho_{\legal} = \Pi_{\legal} \rho \Pi_{\legal} / \tr_\rho
  \Pi_{\legal}$ and $\Pi_{\legal}$ be the projection to the legal
  clock subspace as in the proof in Theorem~\ref{thm:honest-game}.
  Let $s_2$ be the maximum acceptance probability in game $G_2$ for
  players who share a state supported on the legal clock subspace.
  Then the referee accepts with probability
  \begin{equation*}
    \frac{1}{3}(1-\epsilon) + \frac{2}{3} \min \bigl(1,s_2 +
    c\sqrt{q_S\epsilon}\bigr).
  \end{equation*}
  Lemma~\ref{lem:error} then reduced the problem to proving that $s_2
  \in 1 - \poly^{-1}$.

  We now analyze game $G_2$ with states supported on the legal clock
  subspace.
  Suppose the referee rejects with probability $\epsilon$ in the
  \game{Constraint Propagation}.
  Theorem~\ref{thm:mc-game} then implies that the players must play
  approximately honestly for the measurement specification type of
  questions.
  That is, there is a constant $\kappa$, such that there exists an
  isometry $V_i$
  \begin{equation*}
    \begin{split}
      \dis_{\rho_0} (\hat{P}, \check{P}) & \le O(n^\kappa
      \epsilon^{1/\kappa}),\\
      \dis_{\rho_0} (\hat{Q}, \check{Q}) & \le O(n^\kappa
      \epsilon^{1/\kappa}),
    \end{split}
  \end{equation*}
  where $\check{P} = V_i^* (P\otimes \I) V_i$ and $\check{Q}$ is the
  measurement that measures Pauli operators in $Q$ after application
  of $V_i$.
  Furthermore, the probability $p_0$ that outcome $0$ occurs in the
  \game{Output Check} satisfies
  \begin{equation*}
    p_0 \approx_{n^{\kappa} \epsilon^{1/\kappa}} \frac{1}{q_S}.
  \end{equation*}

  Consider the strategy $\Strategy'$ with all measurements $\hat{P}$
  and $\hat{Q}$ replaced by $\check{P}$ and $\check{Q}$ for all
  players.
  Strategy $\Strategy'$ is accepted with probability at most $s_1 \in
  1 - \poly^{-1}$ conditioned on the event that the measurement
  $Z_{\reg{S},1}$ in \game{Output Check} has outcome $0$, as promised
  by Theorem~\ref{thm:honest-game}.

  Therefore the referee accepts with probability at most
  \begin{equation}
    \frac{1}{2}(1-\epsilon) + \frac{1}{2} \min \Bigl( 1,
    1-p_0 + p_0
    \bigl( s_1 + cn^\kappa \epsilon^{1/\kappa} \bigr) \Bigr).
  \end{equation}

  The proof follows from the above equation and Lemma~\ref{lem:error}.
\end{proof}

\subsection{Nonlocal Games for QMIP}

In this section, further transform the extended nonlocal game to a
nonlocal game.
The idea is to introduce eight extra players $(1'), (2'), \ldots,
(8')$ and let them to share an encoding of the referee's state and
measure the local $X$, $Z$ measurements the referee does in the
extended nonlocal game.
As in the extended nonlocal game, the referee's measurements on his
registers correspond to at most $k$ Pauli operators of weight at most
$k$, the referee in the following nonlocal game can delegate the
measurement to the eight extra players and the rigidity theorem for
the $(n,k)$-stabilizer game guarantees that the players have to
measure honestly.

\begin{figure}[!htb]
  \begin{shaded}
    \ul{Nonlocal Game for $\QMIP*$}\\[1em]
    The referee does the following with equal probability:
    \begin{enumerate}
    \item Plays the $(n,k)$-stabilizer game in Fig.~\ref{fig:ms-game}
      with players $(1'), (2'), \ldots, (8')$ and rejects or accepts
      accordingly.
    \item Simulates the extended nonlocal game using logical $X,Z$
      measurements.
    \end{enumerate}
  \end{shaded}
  \caption{The nonlocal game for \QMIP*.}
  \label{fig:nonlocal-game}
\end{figure}

\begin{theorem}
  \label{thm:nonlocal-game}
  For $r\in \poly$, $s\in 1 - \poly^{-1}$ there is an $s'\in 1 -
  \poly^{-1}$.
  For any language $A \in \QMIP*(r,3,1,s)$ and an instance $x$, the
  nonlocal game in Fig~\ref{fig:nonlocal-game} satisfies that
  \begin{enumerate}
  \item If $x\in A$, the referee accepts with certainty;
  \item If $x\not\in A$, the referee accepts with probability at most
    $s'$.
  \end{enumerate}
\end{theorem}

\begin{proof}
  Suppose that the referee rejects with probability $\epsilon$ in the
  first part.
  Theorem~\ref{thm:ms-game} then implies that the players must play
  approximately honestly for the measurement specification type of
  questions.
  Therefore the referee accepts with probability at most
  \begin{equation}
    \frac{1}{2}(1-\epsilon) + \frac{1}{2} \min \bigl( 1, s_1 + cn^\kappa
    \epsilon^{1/\kappa} \bigr),
  \end{equation}
  where $s_1$ is the maximum acceptance probability for the second
  part where all the eight players $(1'), (2'), \ldots, (8')$ measures
  honestly.
  By Theorem~\ref{thm:extended-game}, it follows that $s_1\in 1 -
  \poly^{-1}$.
  The proof then follows from Lemma~\ref{lem:error}.
\end{proof}

Our main theorem (Theorem~\ref{thm:main}) follows by observing that
the questions of the nonlocal game in Fig.~\ref{fig:nonlocal-game} are
measurement specification of at most $k$ Pauli operators of weight at
most $k$ and can be encoded with logarithmically number of bits.
\bibliographystyle{acm}

\bibliography{NEXP-entangled}


\end{document}